\documentclass[a4paper,12pt]{article}  
\usepackage[left=1in, right=1in, top=1in, bottom=1in]{geometry} %,includefoot

\usepackage{amsthm}
\usepackage{amsmath}
\usepackage{amssymb}
\usepackage{mathrsfs} % Added by TF
\usepackage{graphicx}
\usepackage[
numberedbib,
natbibapa
]{apacite}

\usepackage[dvipsnames]{xcolor}
\usepackage{stackengine}
\usepackage{placeins}
\usepackage{arydshln}
\usepackage[colorlinks=true,linkcolor=NavyBlue, citecolor=NavyBlue, urlcolor=NavyBlue, linktocpage=true]{hyperref} 
\usepackage{bbm}
\usepackage{textcomp}
\usepackage{enumerate}
\usepackage{mathtools}
\usepackage{tabularx}
\usepackage{rotating}
\usepackage{threeparttable}

\usepackage{color}
\usepackage{framed}
\usepackage{comment}
\definecolor{shadecolor}{gray}{0.9}

\usepackage{dsfont} % Added by TF
\usepackage[font=small]{caption}
\usepackage{subcaption} % Added by TF
\usepackage{graphicx}
\usepackage{pdfpages}
\usepackage{url}
\usepackage[english]{babel}
\usepackage{setspace}
\usepackage{dashrule}
\usepackage{booktabs}
\usepackage{multirow}
\usepackage{tikz}
\usetikzlibrary{automata, positioning, arrows}
\tikzset{
        ->,  % makes the edges directed
        node distance=5.5cm, % specifies the minimum distance between two nodes. Change if n
        every state/.style={thick, fill=gray!10}, % sets the properties for each ÕstateÕ n
        initial text=$ $, % sets the text that appears on the start arrow
        }

\usepackage{appendix}

\theoremstyle{plain}  %default 
\newtheorem{theorem}{theorem}[section] 
\newtheorem{lemma}[theorem]{Lemma} 
\newtheorem{proposition}[theorem]{Proposition} 
\newtheorem{corollary}[theorem]{Corollary} 

\theoremstyle{definition} 
\newtheorem{definition}[theorem]{Definition}

\newtheorem{assump}[theorem]{Assumption}

\theoremstyle{remark} 

\newcommand{\E}{\mathbb{E}}
\renewcommand{\P}{\mathbb{P}}
\newcommand{\Var}{\mathrm{Var}}

\newcommand{\Cov}{{Cov}} 
\newcommand{\Cor}{\mathrm{Cor}}

\newcommand{\OR}{{OR}}

\newcommand{\MSC}{{MSC}}

\newcommand{\F}{\mathcal{F}}
\DeclareMathOperator{\Pb}{\mathbb{P}}

\newcommand{\toP}{\stackrel{\mathbb{P}}{\to}}
\newcommand{\tod}{\stackrel{d}{\to}}

\newcommand{\pa}{\mathsf{p}}
\newcommand{\qa}{\mathsf{q}}
\newcommand{\ra}{\mathsf{r}}

\newcommand{\p}{\widehat{{p}}_n}
\newcommand{\q}{\widehat{{q}}_n}
\renewcommand{\r}{\widehat{{r}}_n}
\newcommand{\hsigma}{\widehat{\sigma}_n}

\renewcommand{\mp}{m^+}
\newcommand{\hmp}{\widehat{m}_n^+}

\newcommand{\mn}{m^-}
\newcommand{\hmn}{\widehat{m}_n^-}

\newcommand{\mV}{\mathbf{V}}
\newcommand{\mW}{\mathbf{W}}
\newcommand{\mZ}{\mathbf{Z}}
\newcommand{\moW}{\overline{\mathbf{W}}}

\newcommand{\myC}{{C}}
\newcommand{\myCn}{	\widehat{{C}}_n}
\newcommand{\myQ}{{Q}}
\newcommand{\myQn}{	\widehat{{Q}}_n}
\newcommand{\myV}{{V}}
\newcommand{\myT}{{T}}
\newcommand{\myPC}{{PC}}
\newcommand{\myTC}{{TC}}

\newcommand{\myY}{{Y}}
\newcommand{\myFH}{{FH}}

\renewcommand{\H}{\mathcal{H}}

\def\be{\begin{equation} \label}
\def\ee{\end{equation}}

\usepackage[colorinlistoftodos,textsize=tiny]{todonotes}
\newcommand{\Comments}{1}
\newcommand{\mynote}[2]{\ifnum\Comments=1\textcolor{#1}{#2}\fi}
\newcommand{\mytodo}[2]{\ifnum\Comments=1%
  \todo[linecolor=#1!80!black,backgroundcolor=#1,bordercolor=#1!80!black]{#2}\fi}

\ifnum\Comments=1               % fix margins for todonotes
\fi

\ifnum\Comments=1               % fix margins for todonotes
\fi

\ifnum\Comments=1               % fix margins for todonotes
\fi

\begin{document}

	\title{Measuring Dependence between Events\thanks{We thank Tobias Fissler, Gery Geenens, Tilmann Gneiting and Sander Muns, seminar participants at Goethe University Frankfurt (2024) and Heidelberg Institute for Theoretical Studies (2023) and conference participants at the Workshop on Dependence models, Vines, and their Applications at TU Munich (2024), at CMStatistics 2023 in Berlin, Statistical Week 2023 in Dortmund and the 10th HKMetrics Workshop in Karlsruhe (2023) for helpful comments. 
			Marc-Oliver Pohle is grateful for support by the Klaus Tschira Foundation. 
			Timo Dimitriadis gratefully acknowledges support of the Deutsche Forschungsgemeinschaft (DFG, German Research Foundation) through grant 502572912.}}
	
	\author{Marc-Oliver Pohle\thanks{Karlsruhe Institute of Technology, Institute of Statistics, Blücherstraße 17, 76185 Karlsruhe, Germany, and Heidelberg Institute for Theoretical Studies, e-mail: \href{mailto: pohle@kit.edu}{pohle@kit.edu}} 
	\and Timo Dimitriadis\thanks{Goethe University Frankfurt, RuW Building, Theodor-W.-Adorno-Platz 4, 60323 Frankfurt, Germany, and Heidelberg Institute for Theoretical Studies, e-mail: \href{mailto: dimitriadis@econ.uni-frankfurt.de}{dimitriadis@econ.uni-frankfurt.de}}
	\and Jan-Lukas Wermuth\thanks{Goethe University Frankfurt, RuW Building, Theodor-W.-Adorno-Platz 4, 60323 Frankfurt, Germany, e-mail: \href{mailto: wermuth@econ.uni-frankfurt.de}{wermuth@econ.uni-frankfurt.de}}}
	
	\maketitle	

\begin{abstract}	
	\noindent
	Measuring dependence between two events, or equivalently between two binary random variables, amounts to expressing the dependence structure inherent in a $2\times 2$ contingency table in a real number between $-1$ and 1. Numerous such dependence measures exist, but there is little theoretical guidance on how they compare and on their advantages and shortcomings. Thus, practitioners might be overwhelmed by the problem of choosing a suitable measure. We provide a set of natural desirable properties that a \emph{proper} dependence measure should fulfill. We show that Yule's $\myQ$ and the little-known Cole coefficient are proper, while the most widely-used measures, the phi coefficient and all contingency coefficients, are improper. They have a severe attainability problem, that is, even under perfect dependence they can be very far away from $-1$ and $1$, and often differ substantially from the proper measures in that they understate strength of dependence. 
	The structural reason is that these are measures for equality of events rather than of dependence. 
	% We uncover the structural reason for these problems: they are rather measures of equality of events than of dependence. 
	We derive the (in some instances non-standard) limiting distributions of the measures and illustrate how asymptotically valid confidence intervals can be constructed. In a case study on drug consumption we demonstrate how misleading conclusions may arise from the use of improper dependence measures. 
\end{abstract}


% \color{red}
% ToDo:
% \begin{itemize}
%     \item Sentence about perfect dependence and AE comment
%     \item Pseudo-annals paper zitieren? Scheint ja etwas am Thema vorbei zu gehen...
%     \item Gary zitieren? 
%     \item JRSSB template.
%     \item Danksagung: Gary
% \end{itemize}
% \color{black}

\onehalfspacing
% \doublespacing

\section{Introduction}

Analyzing dependence between two random variables is a fundamental task in statistics. The arguably simplest setting is the one of two binary random variables or, equivalently, two events, which is fully characterized by two marginal and a joint probability, and  is conveniently represented in a $2 \times 2$ contingency table. It is often useful to quantify the dependence structure inherent in the contingency table with a dependence measure, that is, a number lying between $-1$ and 1 and indicating direction of dependence via its sign and strength of dependence via the closeness of its absolute value to 1. However, there exist many dependence measures for this task and hardly any theoretical guidance on the advantages and shortcomings of those measures and on their relationships. %Thus, practitioners might be overwhelmed by the problem of choosing a suitable dependence measure. 
In this paper, we provide such a theoretical guidance based on a formal discussion of dependence concepts for events, an axiomatic approach to dependence measures and a theoretical analysis of existing measures and their relationships, leading to clear recommendations for statistical practice. We also develop statistical inference for the recommended measures.

\subsection{Motivating Example}

The quest for the most suitable dependence measure in the binary case already led to fierce debates among two of the greats of statistics, see \citet{Ekstrom2011} for an account of the so-called Pearson-Yule debate. We consider a motivating data example popular in that period, namely the relation between vaccination against and survival of smallpox. The left part of Table \ref{tab:smallpox_contingency_table} contains absolute marginal and joint frequencies for the events $A=\{ \text{vaccinated}\}$ and $B=\{ \text{recovery}\}$ and their complements for the smallpox epidemic in Leicester in 1892/93 taken from \citet{yule1912}. Looking at the table, the relation between the two events seems positive and quite strong. However, of course a formal way to quantify this dependence is required. 

\begin{table}[tb]
    \small
	\begin{minipage}{.5\linewidth}
		\centering
		\textbf{Contingency Table}
		\begin{tabular}{cccc}
			\addlinespace
			\toprule
			& recovery &  death &  \\
			\midrule 
			vaccinated         &  197   &     2  &  199     \\     
			unvaccinated           &   139   &     19  & 158     \\  
			&   336 & 21 & 357\\
			\bottomrule
		\end{tabular}
	\end{minipage}
	\begin{minipage}{.5\linewidth}
		\centering
		\textbf{Estimated Dependence Measures}
		\begin{tabular}{ccc}
			\addlinespace
			\toprule
			measure & value& 90\% CI    \\
			\midrule 
			Phi coefficient & 0.23 & $[0.16,0.30]$  \\
			Yule's $\myQ$ & 0.86 & $[0.59,0.96]$ \\
			Cole's coefficient $\myC$ & 0.83 & $[0.44,0.96]$ \\
			\bottomrule
		\end{tabular}
	\end{minipage}
	\caption{Data from the Leicester smallpox epidemic 1892/93 taken from \cite{yule1912}. 
		Left table: $2\times 2$ contingency table with absolute frequencies.
		Right table: Three estimated dependence measures and 90\% confidence intervals based on the Fisher transformation; see Section \ref{sec:asymptotics} for details.} 
	\label{tab:smallpox_contingency_table}
\end{table}

While numerous measures for this task have been proposed (see \citet{warrens2019similarity} for a comprehensive overview of such measures and their use in different disciplines), 
their theoretical properties and relationships have hardly been studied. Thus, it is difficult to pick a measure in practice and, if multiple measures are used, to understand why they often lead to different conclusions. Due to the lack of theoretical guidance, practitioners understandably are likely to choose the most popular measures. Table \ref{tab:google_hits} contains the hits of a recent search for eight measures on Google and Google Scholar as a proxy for their popularity in practice and academia together with their estimated values on the smallpox data example of Table \ref{tab:smallpox_contingency_table}.
The most popular measures can clearly be found on the left side, namely the phi coefficient (referred to as Matthews correlation coefficient in the machine learning literature \citep{matthews1975}) and contingency coefficients, of which Cram\'er's V is the most well-known representative.\footnote{Note that contingency coefficients are applicable to general contingency tables. Thus, the hits may overstate the importance of Cram\'er's V for the specific $2\times 2$ case.} Those measures take the rather small value of 0.23, suggesting weak positive dependence between the two events. In contrast, the measures on the right side lie between 0.57 and 0.86, indicating a rather strong dependence and being more in line with the impression one gets when looking at the contingency table. We demonstrate in this article that there is a reason for this impression and that indeed the measures on the right-hand side are theoretically superior to the ones on the left-hand side. % and that it thus may be detrimental to choose the dependence measure by popularity.
%More precisely, we show in this article that the phi coefficient and the contingency coefficients lack the fundamental property of attainability---meaning that they are often far away from $\pm 1$ despite a very strong or even perfect dependence of the events---and can thus be severely misleading. 
The former, even though less widely-used, fulfill a set of crucial properties of dependence measures and thus give a sensible assessment of dependence. 
%The striking mismatch of the estimated values on the left and right sides of Table \ref{tab:google_hits} demonstrates the severity of the attainability issue and the related and potentially dangerous misinterpretations.

% The left side of Table \ref{tab:smallpox_contingency_table} presents estimates for the data example for three dependence measures, namely the popular phi coefficient together with the less popular Yule's $\myQ$ and the Cole coefficient $\myC$.
% The latter two will however arise as being theoretically superior in the course of this paper.
% While the phi coefficient suggests a weak dependence with a value of 0.23, Yule's $\myQ$ and the Cole coefficient $\myC$ equal 0.86 and 0.83, respectively.
% Hence, they indicate a much stronger dependence, which is more in line with the impression one gets when looking at the contingency table. 

\begin{table}[tb] 		
    \small
	\begin{minipage}{.5\linewidth}
		\centering
		\textbf{Improper Dependence Measures}
		\begin{tabular}{cccc}
        	\addlinespace
			\toprule
			measure & value& Google & Scholar    \\
			\midrule 
			Cram\'er's V & $0.23$ & $312{,}000$ & $29{,}200$ \\
			Matthews cor. & $0.23$ & $333{,}000$ & $42{,}200$ \\
			Phi coef. & $0.23$ & $122{,}000$ & $26{,}500$ \\ 
			Pearson cont.~coef. & $0.23$ & $3{,}590$ & $515$ \\
			%Tschuprow's T & 0.23 & 8,640 & 216 \\
			\bottomrule
		\end{tabular}
	\end{minipage}
    \hfill
	\begin{minipage}{.5\linewidth}
		\centering
		\textbf{Proper Dependence Measures}
		\begin{tabular}{cccc}
			\addlinespace
			\toprule
			measure & value& Google & Scholar    \\
			\midrule 
			%Odds ratio & 7.05 & 86,200,000 & 1,330,000
			Tetrach.~cor. & $0.61$ & $53{,}900$ & $10{,}100$ \\
			Yule's $\myQ$ & $0.86$ & $14{,}500$ & $2{,}180$ \\
			Yule's $\myY$ & $0.57$ & $4{,}110$ & $439$ \\
			Cole's coef. & $0.83$ & $368$ & $136$ \\
			%Digby's H & $ $ & 0.62 & 481 & 2 \\
			\bottomrule
		\end{tabular}
	\end{minipage}
	\caption{Estimates of some improper (left) and proper (right) measures of dependence for the smallpox data example and Google and Google Scholar hits of those measures (7th November 2025).}% Cramér's V (50,400/2,110 hits) was spelled Cramer's V}	
	\label{tab:google_hits}
\end{table}

%In this paper we analyze the theoretical properties of dependence measures for events or binary random variables and analyze how they relate to each other. The relevance of this is at least twofold: Firstly, binary variables are very common and measuring dependence between them is a common task in statistical practice. We give clear recommendations, which measures to use and which to avoid. Further, we explain in which situations the suitable measures differ from each other and why and how. Secondly, having a sound theoretical understanding of dependence and its measurement in the fundamental case of two binary random variables may enhance the understanding of the more general case of two arbitrary random variables and foster research in this field.

\subsection{Plan of the Paper}

A sound discussion of dependence measures and their properties requires underlying dependence concepts as a fundament: In Section \ref{sec:dependence_concepts} we introduce concepts of positive and negative, of stronger and weaker and of perfect dependence between two events. In Section \ref{sec:desirable_properties} we provide a formal definition of dependence measures for events, postulate desirable properties for them and call a dependence measure proper if it fulfills them.

After a discussion of the covariance in the case of binary random variables, which serves as a building block for most of the other measures, Section \ref{sec:linear_measures} treats the popular, but improper, measures. We show that the phi coefficient, which is just a special case of the classical Pearson correlation coefficient, lacks the fundamental property of attainability. That is, it does in general not take the values 1 and $-1$ under perfect positive and negative dependence.
We further demonstrate that it may take very small (absolute) values in those cases and usually strongly understates strength of dependence, where the severity of this problem depends on the marginal event probabilities. 
The phi coefficient consequently cannot provide a reliable assessment of strength of dependence and the usual interpretation of values close to 0 indicating weak and absolute values close to 1 indicating strong dependence does not apply here. 
While the phi coefficient can hence not be regarded as a proper dependence measure, we show that it is instead useful for a different task, namely measuring closeness to equality of two events. %, which amounts to jointly measuring dependence and a restriction on the marginal distributions. 
This is for example required in the evaluation of binary classifiers, where equality of classification and outcome characterizes a perfect classifier. We also discuss several popular contingency coefficients and show that for $2 \times 2$ contingency tables they are simple functions of the phi coefficient and thus inherit its shortcomings. 

Section \ref{sec:proper_measures} is dedicated to proper dependence measures. Cole's coefficient is a very natural measure, just normalizing the covariance with its values under perfect positive and negative dependence. Yule's $\myQ$ uses a different normalization, not requiring a case distinction between positive and negative dependence. We also touch upon a whole class of measures that generalize $\myQ$, including Yule's $\myY$, and on tetrachoric correlation. A measure that falls out of the classical framework of dependence measures lying between $-1$ and 1, but is extremely popular and nicely interpretable, is the odds ratio, which is closely connected to Yule's $\myQ$ as well. We show that it is proper too, albeit with respect to a modified set of axioms that suits its scale and multiplicative nature (see Proposition \ref{prop:properties_OR}). 

Our focus for the rest of the paper then lies on the phi coefficient, Yule's $\myQ$ and Cole's coefficient as basically all the other measures discussed are closely related to the first two and share their properties. We first compare their values under different marginal event probabilities, which demonstrates that the phi coefficient differs in general quite strongly from the proper measures and sheds light on their relation as well. 

In Section \ref{sec:asymptotics} we derive the asymptotic distributions of their empirical counterparts $\myCn$, $\myQn$ and $\widehat{\phi}_n$, under assumptions covering independent and identically distributed observations, but also time series data. $\myQn$ and $\widehat{\phi}_n$ are asymptotically normal, allowing for tests and confidence intervals in the classical way. For $\myCn$ the limiting distribution is nonstandard and involves case distinctions with respect to the value of the true $\myC$ due to the case distinction used in its normalization. 
While testing is straightforward, we construct possibly conservative confidence intervals by a Bonferroni-corrected combination of inverted tests that are built on the different asymptotic distributions. 
For all three measures, we also derive the limiting distribution of the Fisher transformation, which leads to a better approximation than the classical limit distribution when the measures lie close to $-1$ or 1 and employ them for testing and confidence intervals. 
E.g., the right side of Table \ref{tab:smallpox_contingency_table} shows $90\%$-confidence intervals based on the Fisher transformation for the data example.

Section \ref{sec:application} presents an application to drug use data, where we analyze the interdependence between the consumption of different drugs. It illustrates the relevance of using proper dependence measures: The phi coefficient suggests a very weak interdependence, while the proper measures indicate that it is quite strong, being in line with the gateway hypothesis of drug use.

Section \ref{sec:conclusion} concludes. 
The Appendix contains all proofs, further details on some of the discussed dependence measures and on asymptotics, additional tables and graphs and simulations investigating the finite sample performance of our confidence intervals. An accompanying R package called \texttt{BCor} is available at \href{https://github.com/jan-lukas-wermuth/BCor}{https://github.com/jan-lukas-wermuth/BCor} and replication material is available under  \href{https://github.com/jan-lukas-wermuth/replication_BCor}{https://github.com/jan-lukas-wermuth/replication\_BCor}.

\subsection{Literature}

The most popular dependence measures for binary random variables were already introduced around the year 1900. According to \citet{Ekstrom2011}, the phi coefficient was independently proposed by \citet{pearson1900}, \citet{boas1909} and \citet{yule1912}. Yule's $\myQ$ was introduced by \cite{yule1900} and tetrachoric correlation in \cite{pearson1900}. \cite{yule1912} is a foundational paper too, amongst others introducing Yule's $\myY$ as a variant of Yule's $\myQ$. In the subsequent decades, several contingency coefficients have been proposed such as Pearson's contingency coefficient \citep{Pearson1904}, Cramér's V \citep{Cramer1946} and Tschuprow's T \citep{Tschuprow1925}. Cole's coefficient was introduced by the ecologist \cite{cole1949} and remained largely unnoticed by the statistical literature.

Despite this strong early interest and the widespread use of these measures in applications, the theoretical literature on them has surprisingly been very slim since those early years. \citet[Chapter 11]{bishop2007discrete} and \cite{warrens2008association} are noteworthy exceptions, taking into account some theoretical properties. 

The literature on dependence measures for general random variables is much larger, see \citet{Balakrishnan2009,Mari2001,Tjostheim2022} for overviews. The axiomatic approach to the study of general dependence measures originating in \cite{Renyi1959} is the role model for our treatment of the binary case. 
However, the literature usually restricts itself to the case of continuous random variables (e.g.\ \citet{Schweizer1981}, \citet{Embrechts2002}), which excludes the binary case, and entails important differences.
\cite{neslehova2007rank} marks a noteworthy exception, providing a discussion of rank correlations in the discrete case from a copula perspective.

Asymptotic sampling variances for estimators of $\phi$ and $\myQ$ are already discussed in the original papers cited above, but restricted to the case of independent and identically distributed data. Inference for Cole's coefficient is developed in this paper for the first time.

\section{Dependence Concepts for Events} \label{sec:dependence_concepts}

\subsection{Setup}

Consider a probability space $(\Omega, \mathcal{F},\P)$ and two events $A,B \in \mathcal{F}$ with $0<\P(A)<1$ and $0<\P(B)<1$. We are interested in measuring direction and strength of dependence between those two events, or equivalently, between the two dichotomous or binary random variables $X := \mathds{1}_A$ and $Y :=\mathds{1}_B$, which take the value 1 if $A$ or $B$, respectively, happen and 0 if not. Denote the complements of $A$ and $B$ by $\overline{A}$ and $\overline{B}$. The joint distribution is fully characterized by three of the four joint probabilities $\P(A \cap B)$, $\P(A \cap \overline{B})$, $\P(\overline{A} \cap B)$ and $\P(\overline{A} \cap \overline{B})$ or one joint and the two marginal probabilities, e.g.,  $\P(A \cap B)$, $\P(A)$ and $\P(B)$. It can conveniently be represented in a $2 \times 2$ contingency table as in Table \ref{tab:contingency_table}.

\begin{table}[tb]
	\centering
    \small
		\begin{tabular}{cccc}
			\toprule
			& $B$ &  $\overline{B}$ &  \\
			\midrule 
			$A$         &   $\P(A \cap B)$   &     $\P(A \cap \overline{B})$  & $\P(A)$     \\     
			$\overline{A}$           &   $\P(\overline{A} \cap B)$   &     $\P(\overline{A} \cap \overline{B})$  & $\P(\overline{A})$     \\  
			&   $\P(B)$ & $\P(\overline{B})$ & 1\\
			\bottomrule
		\end{tabular}
	\caption{A $2 \times 2$ contingency table.}
	\label{tab:contingency_table}
\end{table}

%\begin{table}
%	\begin{center}
%		\begin{tabular}{cccc}
%			\toprule
%			
%			& $B$ &  $\overline{B}$ &  \\
%			\midrule 
%			$A$         &   $a$   &     $b$  & $a+b$     \\     
%			$\overline{A}$           &   $c$   &     $d$  & $c+d$     \\  
%			&   $a+c$ & $b+d$& 1\\
%			\bottomrule
%		\end{tabular}
%	\end{center}
%	\caption{2x2 contingency table}
%	\label{tab:contingency_table}
%\end{table}

%\begin{table}
%\begin{center}
%\begin{tabular}{cccc}
%\toprule
%& \multicolumn{3}{c}{$Y$} \\
%\cmidrule{2-4}
%    $X$ & $\overline{B} &  $B$ & Total \\
%\midrule 
%    $\overline{A}$         &   $p_{\overline{A}\overline{B}}$   &     $p_{\overline{A}B}$  & $p_{\overline{A}}$     \\     
%    $A$           &   $p_{A\overline{B}}$   &     $p_{AB}  & $p_{A}$     \\  
%    Total       &   $p_{\overline{B}}$ & $p_{B}$& 1\\
%\bottomrule
%\end{tabular}
%\end{center}
%\caption{2x2 contingency table}
%\label{tab:contingency_table}}
%\end{table}

\subsection{Direction of Dependence and Dependence Ordering}

We now define some fundamental dependence concepts. Even nominal dichotomous variables can be treated as ordinal when it comes to measuring dependence as there is always a natural ordering, namely the event happening or not. Thus, it makes sense to not only analyze strength, but also direction of dependence in this setting. We speak of positive (negative) dependence when $A$ makes $B$ more (less) likely or vice versa. Recall the ubiquitous definition of independence of $A$ and $B$ via
$\P(A\cap B)= \P(A)\P(B)$.
A surprisingly much lesser-known, but equally fundamental, definition is the following \citep{falk1983}, which is a special case of the notion of quadrant dependence for arbitrary random variables introduced by \citet{lehmann1966}.

\begin{definition}[Positive and Negative Dependence] \label{def:positive_dependence}
	Two events $A,B \in \mathcal{F}$ are positively dependent if  $\P(A\cap B)\ge \ \P(A)\P(B)$. They are negatively dependent if $\P(A\cap B)\le \ \P(A)\P(B).$
\end{definition}

An equivalent expression for positive (negative) dependence can be found in terms of conditional probabilities, i.e.\ $\P(A|B) \ge (\le)\ \P(A)$ or $\P(B|A) \ge (\le)\ \P(B)$.

The following definition captures a natural way of ordering pairs of events in terms of strength of dependence if the respective marginal event probabilities are equal. The more likely two events are to occur together, that is, the larger their intersection is, the stronger positively dependent they are and vice versa for negative dependence.

\begin{definition}[Stronger Dependence] \label{def:stronger_dependence}
	Consider two pairs of events $A,B$ and $A^*,B^*$ with $\P(A)=\P(A^*)$ and $\P(B)=\P(B^*)$. Then $A$ and $B$ are stronger positively dependent than $A^*$ and $B^*$ if $\P(A\cap B) \ge \P(A^*\cap B^*)$. $A$ and $B$ are stronger negatively dependent than $A^*$ and $B^*$ if $\P(A\cap B) \le \P(A^*\cap B^*)$. They are equally dependent if $\P(A\cap B) = \P(A^*\cap B^*)$.
\end{definition}

Again, stronger positive (stronger negative) [equal] dependence can equivalently be expressed in terms of conditional probabilities: $\P(A|B) \ge (\le) [=] \ \P(A^*|B^*)$ or $\P(B|A) \ge (\le) [=] \ \P(B^*|A^*)$. This definition can be seen as a special case of the dependence ordering for random variables with identical marginals discussed by \citet{yanagimoto1969partial} and \citet{tchen1980inequalities}. While for non-binary random variables this ordering is only partial, here in the binary case we get a complete ordering.% for pairs of random variables with identical event probabilities. One of the key goals of measures of dependence is to provide a sensible ordering if the event probabilities are not identical. 

\subsection{Perfect Dependence}

We now introduce concepts of perfect positive and negative dependence between events. We are not aware of a discussion of this key issue in the literature, even though a thorough treatment of dependence measures is not possible without such concepts. Fortunately, they are very natural as well. From Definition \ref{def:stronger_dependence} we know that for fixed event probabilities the dependence between $A$ and $B$ gets stronger in the positive direction if the intersection $A \cap B$ grows. The dependence gets maximal if $A$ and $B$ occur together as often as possible. This is the case if the intersection reaches its maximal size, that is, if $\P( A \setminus B) = \P (A \cap \overline{B} ) = 0$ or $\P( B \setminus A) = 0$. Essentially, this means that the smaller event is a subset of the larger and thus the intersection is equal to the smaller event. Conversely, the dependence gets strongest in the negative sense if $A$ and $B$ occur together as little as possible, that is, if their intersection reaches the minimal size. This is equivalent to the intersection of one set and the complement of the other reaching the maximal size. 

\begin{definition}[Perfect Dependence] \label{def:perfect_dependence}
	$A$ and $B$ are perfectly positively dependent if $\P( A \setminus B) = 0$ or $\P( B \setminus A) = 0$. They are perfectly negatively dependent if $\P( A \setminus \overline{B}) = 0$ or $\P( \overline{B} \setminus A) = 0$. We call $A$ and $B$ perfectly dependent if they are perfectly positively or perfectly negatively dependent.
\end{definition}  

We establish three characterizations of perfect dependence.

\begin{proposition}[Characterizations of Perfect Dependence] \label{prop:characterizations} 
	$ $ \\[-0.7cm]
	\begin{enumerate}[(i)]
		\item $A$ and $B$ are perfectly positively dependent if and only if $\P(A \cap B) = \min(\P(A),\P(B))$. They are perfectly negatively dependent if and only if $\P(A \cap B) = \max(0,\P(A) + \P(B) -1)$. It further holds that $\max(0,\P(A) + \P(B) -1) \leq \P(A \cap B) \leq \min(\P(A),\P(B))$.
		\item $A$ and $B$ are perfectly positively dependent if and only if it holds that $P(A|B)=1$ or $P(B|A)=1$. They are perfectly negatively dependent if and only if it holds that $P(A|\overline{B})=1$ or $P(\overline{B}|A)=1$.
		\item $A$ and $B$ are perfectly positively dependent if and only if it holds that $\P(A \cap \overline{B})=0$ or $\P(\overline{A} \cap B)=0$. They are perfectly negatively dependent if and only if it holds that $\P(A \cap B)=0$ or $\P(\overline{A} \cap \overline{B})=0$.
	\end{enumerate}
\end{proposition}

Part (i) of this Proposition expresses perfect dependence in terms of the relationship between the joint probability of the events $A$ and $B$ and their marginal probabilities. The bounds established there are crucial for the construction and understanding of dependence measures. They are the essence behind and at the same time special cases of the Fr\'{e}chet--Hoeffding bounds for bivariate cumulative distribution functions (CDFs) \citep{frechet1951tableaux, hoeffding1940} and we call them Fr\'{e}chet--Hoeffding bounds for probabilities. Part (i) also establishes that our definition of perfect dependence from Definition \ref{def:perfect_dependence} falls under the notion of co- and countermonotonicity of two random variables, which is usually defined via the  Fr\'{e}chet--Hoeffding bounds for CDFs \citep{Embrechts2002}, when considering the two indicators $\mathds{1}_A$ and $\mathds{1}_B$. From now on we therefore use the terms perfect positive (negative) dependence and comonotonicity (countermonotonicity) of events interchangeably. Part (ii) reformulates Definition \ref{def:perfect_dependence} in terms of conditional probabilities. Part (iii) states that perfect positive dependence is equivalent to at least one of the elements of the secondary diagonal of Table \ref{tab:contingency_table} being 0 and perfect negative dependence is equivalent to at least one of the elements of the primary diagonal being 0. As a consequence, $A$ and $B$ are perfectly dependent if and only if any of the joint probabilities in the contingency table is 0.

Finally, we state a lemma on dependence concepts for complements for later use.
\begin{lemma}[Dependence Concepts for Complements] \label{lemma:complement_dependence}
$ $ \\[-0.7cm]
\begin{enumerate}[(i)]
	\item $A$ and $B$ are positively (negatively) dependent if and only if $A$ and $\overline{B}$ are negatively (positively) dependent.
 	\item $A$ and $B$ are stronger positively (negatively) dependent than $A^*$ and $B^*$ if and only if $A$ and $\overline{B}$ are stronger negatively (positively) dependent than $A^*$ and $\overline{B^*}$.
   	\item $A$ and $B$ are perfectly positively (negatively) dependent if and only if $A$ and $\overline{B}$ are perfectly negatively (positively) dependent.
\end{enumerate}
\end{lemma}

\section{Desirable Properties of Dependence Measures} \label{sec:desirable_properties}

We want to express the dependence between two events $A$ and $B$ or the corresponding binary random variables by a single number, a dependence measure.

\begin{definition}[Dependence Measure]
	A dependence measure for the events $A$ and $B$ is a mapping $\delta: D  \rightarrow \mathbb{R}$, $(\P(A),\P(B),\P(A \cap B)) \mapsto \delta(\P(A),\P(B),\P(A \cap B))$, where $D := \big\{(\pa,\qa,\ra) \in (0,1)^3 \mid  \max(0,\pa+\qa-1) \leq \ra \le  \min(\pa, \qa) \big\}$. We write $\delta(A,B):= \delta(\P(A),\P(B),\P(A \cap B))$.
\end{definition}

We consider dependence measures that indicate direction as well as strength of dependence. To be useful, they should fulfill certain properties. Many such sets of properties have been put forward since \cite{Renyi1959}, which, however, usually focus on the case of continuous random variables (e.g.\ \citet{Schweizer1981}, \citet{Embrechts2002}, \citet{Mari2001}, \citet{Balakrishnan2009}, \citet{fissler2023generalised}).
Hence, the following definition establishes such a desirable set of properties for the binary case.

\begin{definition}[Proper Dependence Measure] \label{def:proper_measure}
	We call a dependence measure $\delta(A,B)$ for the events $A$ and $B$ \emph{proper} if it fulfills the following properties.
	\begin{enumerate}[(A)] 
		\item \emph{Normalization}: $-1 \le \delta(A,B) \le 1$.
		\item \emph{Independence}: $\delta(A,B)=0$ if and only if $A$ and $B$ are independent.
		\item \emph{Attainability}: $\delta(A,B)=1 \ (-1)$ if and only if $A$ and $B$ are perfectly positively (negatively) dependent.
		\item \emph{Monotonicity}: Let $A^*$ and $B^*$ be two further events with $\P(A)=\P(A^*)$ and $\P(B)=\P(B^*)$. Then,  $\delta(A,B) \ge (\le) \ \delta(A^*,B^*)$ if and only if $A$ and $B$ are stronger positively (negatively) dependent than $A^*$ and $B^*$.
		\item \emph{Symmetry}: $\delta(A,B)=\delta(
		B,A)$ and $\delta(A,\overline{B})=-\delta(A,B)$.
	\end{enumerate}
\end{definition}

Usually, directed dependence measures lie in the interval $[-1,1]$ (property (A)), where negative values indicate negative and positive values indicate positive dependence. (B) and (C) make sure that $\delta(A,B)$ correctly indicates the extreme cases of independence as well as perfect positive and negative dependence.  Note that, contrary to the case of two general random variables, a dependence measure in the binary case is not a summary measure in the sense of having to condense information (a complicated dependence structure characterized fully by a joint CDF or copula) into a single number, but it usually characterizes the full dependence structure in this single number in the sense that given the marginal event probabilities $\P(A)$ and $\P(B)$ and the dependence measure $\delta(A,B)$, we can recover $\P(A \cap B)$ (indeed, this is possible for all measures discussed in this paper except for the contingency coefficients). As a consequence of this, the ``only if'' in (B) is infeasible for directed measures of dependence for arbitrary random variables \citep{Embrechts2002}, but in the binary case it is feasible. (D) is a very natural property, making sure that stronger dependence leads to more extreme values of $\delta(A,B)$. It is related to the property of coherence proposed by \cite{scarsini1984measures}. 
(C) and (D) ensure that a dependence measure is a sensible indicator of strength of dependence in that larger absolute values of the measure indicate stronger dependence and absolute values close to 0 and 1 indicate weak and strong dependence, respectively.
Properties (B) and (D) also imply that the measure correctly indicates the direction of dependence as formalized through the following Proposition.
% A further property that one could demand is that the measure correctly indicates direction of dependence, but this is already implied by (D).

\begin{proposition} 
	\label{cor:pos_neg_dependence}
	For a dependence measure $\delta(A,B)$ satisfying $(B)$ and $(D)$ from Definition \ref{def:proper_measure}, it holds that $\delta(A,B) \geq (\leq)\  0$ if and only if $A$ and $B$ are positively (negatively) dependent.
%	Let $\delta(A,B)$ be a dependence measure for which $(B)$ and $(D)$ from Definition \ref{def:proper_measure} hold. 
%	Then, it holds that $\delta(A,B) \geq (\leq)\  0$ if and only if $A$ and $B$ are positively (negatively) dependent.
\end{proposition}

The first part of (E) is just classical symmetry and the second part demands that when considering the complement of one event (or, in other words, exchanging rows or columns of the corresponding contingency table), the sign of the dependence measure should be reversed. It also implies that $\delta(\overline{A},\overline{B}) = \delta(A,B)$.

\section{Improper Dependence Measures} \label{sec:linear_measures}

\subsection{Covariance} \label{subsec:covariance}

The definition of positive (negative) dependence (see Definition \ref{def:positive_dependence}) suggests a natural way to measure direction and strength of dependence: Consider the deviation of the joint probability from the joint probability under independence, $\P(A\cap B) - \P(A)\P(B)$, which equals the covariance of $\mathds{1}_A$ and $\mathds{1}_B$.

\begin{definition}[Covariance]
We call 
$\Cov (A,B) := \Cov( \mathds{1}_A,\mathds{1}_B) = \P(A\cap B) - \P(A)\P(B)$
the covariance of the events $A$ and $B$. 
\end{definition}

The covariance can be rewritten as
\begin{equation} \label{eq:covariance_representation}
	\Cov (A,B)	= \P(A\cap B)\P(\overline{A}\cap \overline{B})-\P(\overline{A}\cap B)\P(A \cap \overline{B}),
\end{equation}
that is, as the difference of the products of the elements of the primary and the secondary diagonal of the corresponding contingency table. Covariance is the foundational building block for most dependence measures discussed in this article.

It already fulfills all but one of the properties of a proper dependence measure. Using the Fr\'echet--Hoeffding bounds for probabilities from Proposition \ref{prop:characterizations}, we get what we call Fr\'echet--Hoeffding bounds for covariance,
\begin{equation} \label{eq:covbounds}
	\max \big( 0,\P(A) + \P(B) -1 \big) - \P(A) \P(B) \leq \Cov (A,B) \leq \min \big( \P(A),\P(B) \big) - \P(A) \P(B).
\end{equation}
They imply $ - \frac 1 4 \leq \Cov (A,B) \leq \frac 1 4$. Thus, property (A) is fulfilled, but property (C) is not. Properties (B) and (D) follow by definition of independence and stronger positive (negative) dependence, (E) directly by representation \eqref{eq:covariance_representation}.

\begin{proposition} \label{prop:covariance_properties}
	$\Cov (A,B)$ fulfills properties (A), (B), (D) and (E) from Definition \ref{def:proper_measure}, but not (C).
\end{proposition}

Thus, if one found a suitable normalization that made the covariance attainable, a proper dependence measure would arise. Indeed, the phi coefficient, Cole's $\myC$ and Yule's $\myQ$, which we discuss below, all turn out to be normalized versions of the covariance.

\subsection{Phi Coefficient} 

The phi coefficient is just the classical Pearson correlation coefficient (as well as grade correlation and Kendall's $\tau_b$, which all coincide in this setting)
%; see \cite{pohle2025inference} for population and sample definitions of these coefficients)
of $\mathds{1}_A$ and $\mathds{1}_B$ and thus a normalized version of $\Cov(A,B)$. 
\begin{definition}[Phi Coefficient] \label{def:phi_coefficient}
	The phi coefficient is defined as
	$$\phi(A,B) := \Cor( \mathds{1}_A,\mathds{1}_B) = \frac{\Cov (A,B)}{\sqrt{\P(A)(1-\P(A))\P(B)(1-\P(B))}}.$$
\end{definition}

As the normalization only depends on the marginal event probabilities, the phi coefficient inherits properties (B), (D) and (E) from the covariance and also fulfills (A) by applying the Cauchy-Schwarz inequality. However, by the Fr\'echet--Hoeffding bounds for the covariance from \eqref{eq:covbounds}, it is bounded by what we call Fr\'echet--Hoeffding bounds for the phi coefficient, that is, its values under perfect negative and positive dependence:
\begin{equation} \label{eq:phibounds}
\frac{\max(0,\P(A) + \P(B) -1) - \P(A) \P(B)}{\sqrt{\P(A)(1-\P(A))\P(B)(1-\P(B))}} \leq \phi(A,B) \leq \frac{\min(\P(A),\P(B)) - \P(A) \P(B)}{\sqrt{\P(A)(1-\P(A))\P(B)(1-\P(B))}}.
\end{equation}
As these bounds are in general not equal to $-1$ and $1$ and depend on the marginal event probabilities $\P(A)$ and $\P(B)$, the phi coefficient is not attainable. 

\begin{proposition} \label{prop:phi_properties}
	$\phi (A,B)$ fulfills properties (A), (B), (D) and (E) from Definition \ref{def:proper_measure}, but not (C).
\end{proposition}

Figure \ref{fig:phibounds} visualizes the bounds from $\eqref{eq:phibounds}$ as a function of the marginal event probabilities. It shows that the values of the phi coefficient under perfect positive and negative dependence can be very far away from $1$ and $-1$. This seriously compromises the interpretability of this measure as it cannot reliably indicate strength of dependence, which should be indicated by $|\phi(A,B)|$ being close to 1. For most combinations of $\P(A)$ and $\P(B)$ such values cannot be reached and instead very small absolute values arise under strong dependence. 
The phi coefficient is only attainable for the special case of $\P(A)=\P(B)=0.5$ (see Proposition \ref{prop:characterizations_linear_dependence} below). This shortcoming may be surprising given that the phi coefficient is just a special case of the popular and widely-used Pearson correlation. However, Pearson correlation has serious attainability problems itself and should be used with care \citep{Embrechts2002}. 
% In the binary case, they become particularly serious (as binary random variables are particularly far away from fulfilling the conditions for attainability of Pearson correlation formulated in \citet[Lemma 3.1]{fissler2023generalised}).
In the binary case, they become particularly serious (as two binary random variables are particularly far away from fulfilling the conditions for attainability of Pearson correlation formulated in \citet[Lemma 3.1]{fissler2023generalised}, that is, from being symmetric and of the same type).

\begin{figure}[tb]
	\centering
	\includegraphics[width=0.8\linewidth]{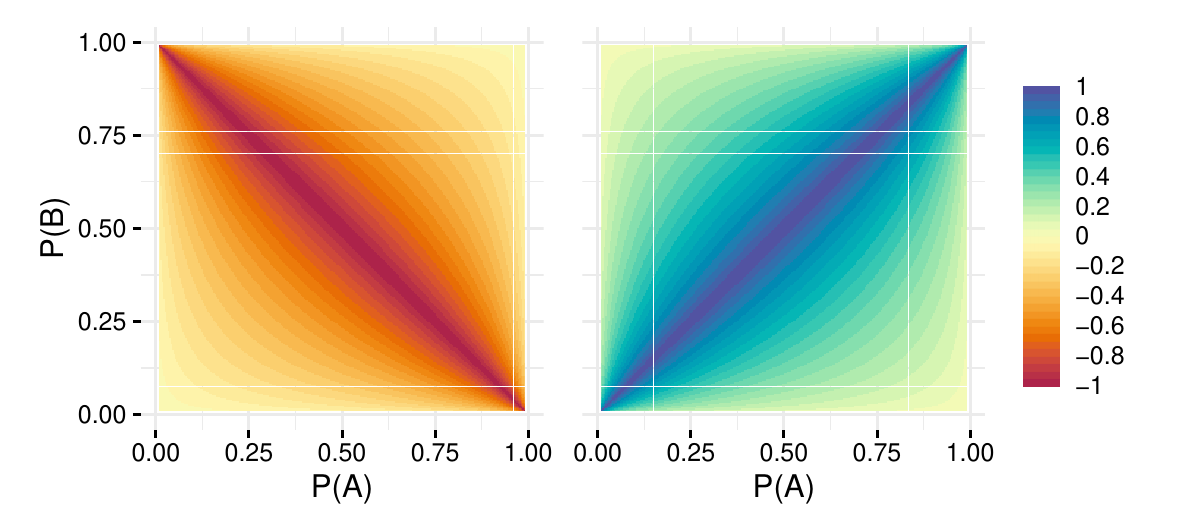}
	\caption{Lower and upper bound of the phi coefficient from \eqref{eq:phibounds} as a function of $\P(A)$ and $\P(B)$.}
	\label{fig:phibounds}
\end{figure}

Thus, the phi coefficient is not a particularly useful measure of dependence for events or binary random variables. Instead, it may be useful for a different task, which becomes clear by analyzing under which conditions $\phi$ takes the values $-1$ and 1.

\begin{proposition} \label{prop:characterizations_linear_dependence} Define the symmetric difference of $A$ and $B$ as $A \Delta B := (A \setminus B) \cup (B \setminus A) = (A \cup B) \setminus (A \cap B)$. The following statements are equivalent:
	\begin{enumerate}[(i)]
		\item $\phi(A,B)=1$ ($\phi(A,B)=-1$).
		\item $\P (A \Delta B) = 0$ ($\P (A \Delta \overline{B}) = 0$).
		\item $A$ and $B$ are perfectly positively (negatively) dependent and $\P(A)=\P(B)$ ($\P(A)=\P(\overline{B})$).
		\item $\P(A \cap \overline{B})=0$ and $\P(\overline{A} \cap B)=0$ ($\P(A \cap B)=0$ and $\P(\overline{A} \cap \overline{B})=0$).
		\item $\mathds{1}_A=\mathds{1}_B$ ($\mathds{1}_A=1-\mathds{1}_B$) almost surely.
		\item $\mathds{1}_A$ and $\mathds{1}_B$ are increasing (decreasing) functions of each other almost surely.
		\item $\mathds{1}_A$ and $\mathds{1}_B$ are increasing (decreasing) linear functions of each other almost surely.
	\end{enumerate}
\end{proposition}

Considering (ii), the phi coefficient essentially measures how close two events $A$ and $B$ are to being equal, $A=B$, for positive dependence and to one set being equal to the complement of the other, $A=\overline{B}$, for negative dependence (which can be regarded as being maximally far away from being equal), which is however different from our characterization of perfect dependence in Definition \ref{def:perfect_dependence}.
By (iii), this amounts to jointly measuring closeness to perfect dependence \emph{and} to a restriction on the marginal distributions, $\P(A)=\P(B)$ in the case of positive (and $\P(A)=\P(\overline{B})$ in the case of negative) dependence.
We however want to measure dependence of arbitrary events and hence take the marginal event probabilities $\P(A)$ and $\P(B)$ as given, which makes the restriction $\P(A)=\P(B)$ (or $\P(A)=\P(\overline{B})$) unacceptable.
%  When measuring dependence we are not interested in this additional restriction as we rather take the marginal event probabilities $\P(A)$ and $\P(B)$ as given. 
The restrictions $\P(A)=\P(B)$ \emph{and} $\P(A)=\P(\overline{B})$ in (iii) jointly amount to $\P(A)=\P(B)=0.5$ and can hence be interpreted as an intuitive reason for the non-attainability of the phi coefficient.
%Further, that $\phi$ takes into account this restriction is the reason why it is not attainable as the restriction is in in general not fulfilled (only in the case $\P(A)=\P(B)=0.5$). 

Even though we have established that the phi coefficient should not be used for measuring dependence, Proposition \ref{prop:characterizations_linear_dependence}  shows that it is useful for a different task, namely measuring closeness to equality of events. For example when evaluating binary classifiers, where $A$ is the event of interest and $B$ indicates that $A$ was predicted to happen, a perfect classifier fulfills $\{A=B\}$. This task is distinct from measuring dependence, where we take the marginal event probabilities as given, since for the classification task $\P(B)$ is not fixed, but needs to adopt to $\P(A)$. In the literature this crucial distinction between measures of dependence and measures of event equality has not been made. While the phi coefficient is indeed widely-used for the evaluation of binary classifiers \citep{matthews1975}, it is as widely-used as a measure of dependence, which should be avoided. At the same time the odds ratio, which is a proper dependence measure (as discussed in the next section) and not a measure of set equality, has been put forward as a tool for the evaluation of binary classifiers \citep{stephenson2000}. 

What has been recognized in the literature (e.g.\ \citet{bishop2007discrete}), is that the phi coefficient equals 1($-1$) under condition (iv), that is, if there are two zeros on the primary (secondary) diagonal of the contingency table, while other measures like Yule's $\myQ$ (in fact, all proper measures by (C) from Definition \ref{def:proper_measure} and (iii) in Proposition \ref{prop:characterizations}) require only one zero on the primary (secondary) diagonal to equal 1 ($-$1). However, the underlying reasons for this as made clear by Propositions \ref{prop:characterizations} and \ref{prop:characterizations_linear_dependence} have not been uncovered.

The items (vi) and (vii) in Proposition \ref{prop:characterizations_linear_dependence} are interesting because they establish the connection to conditions well-known from the theory of Pearson correlation and rank correlations, of which $\phi$ is a special case. (vii) reproduces the well-known condition that Pearson correlation is 1 ($-$1) if and only if the two variables are linear functions of each other, (vi) the well-known condition that grade correlation and Kendall's $\tau_b$ are 1 ($-$1) if and only if one variable is a monotonic function of the other \citep{Embrechts2002}. Thus, equality between the events (or one event and the complement of the other) corresponds to positive (negative) monotonic and linear (the two coincide in the binary case) dependence. Note that neither linear nor monotonic (nor arbitrary) functional relations are useful concepts of perfect dependence in the binary case (or more generally in the discrete case) as a functional relation is only possible under the restriction from (iii) on the marginal distributions, which is in general not fulfilled. This is also the reason why Chatterjee's $\xi$ is an improper measure (see Appendix \ref{subsec:furthermeasures}). 

\subsection{Contingency Coefficients}

Contingency coefficients aim at measuring dependence inherent in a general $k \times l$ contingency table. For two categorical random variables $X$ and $Y$, where $X$ takes values $x_1,...,x_k$ and $Y$ takes values $y_1,...,y_l$, Pearson's Mean Square Contingency ($\MSC$) coefficient \citep{Pearson1904} is just the population analogue of the test statistic of Pearson's chi-squared test of independence:
\begin{equation} \label{eq:MSC}
	\MSC(X,Y) := \sum_{i=1}^k \sum_{j=1}^l  \frac{ \left(\P(X=x_i,Y=y_j) - \P(X=x_i) \P(Y=y_j) \right)^2}{\P(X=x_i) \P(Y=y_j)}.
\end{equation}
Contingency coefficients arise by normalizing $\MSC(X,Y)$. As they are designed for nominal random variables, direction of dependence is not a sensible concept and they map to $[0,1]$, only trying to indicate strength of dependence. However, as discussed above, binary random variables can always be regarded as ordinal. Thus, in the $2 \times 2$ case, measures of dependence that indicate strength and direction of dependence can be used and the question arises why one should throw away the information on direction of dependence. Since they are widely-used in the binary case, we nevertheless discuss them. 
In Appendix \ref{subsec:contingency_coefficients}, we consider Cram\'ers V, Tschuprow's T, and Pearson's contingency coefficient and show that they are closely related to the phi coefficient in the binary setting (for example, the first two equal $|\phi(A,B)|$) and thus inherit its deficiencies, in particular not being proper dependence measures. The same holds true for distance correlation (which also equals $|\phi(A,B)|$) and Chatterjee's $\xi$ (which equals $\phi^2(A,B)$), dependence measures that have recently gained popularity in the statistical literature, see the discussion in Appendix \ref{subsec:furthermeasures}.

\section{Proper Dependence Measures} \label{sec:proper_measures}

\subsection{Cole's Coefficient}

Cole's Coefficient was proposed by \cite{cole1949}. An equivalent reformulation of it arises naturally from our discussion of dependence concepts: From the Fr\'echet--Hoeffding bounds of $\P(A \cap B)$ from Proposition \ref{prop:characterizations} arise the bounds on $\Cov(A,B)$ from \eqref{eq:covbounds}, which are attained only under perfect positive and negative dependence of $A$ and $B$ and whose absolute values are in general different from each other. Thus, it is natural to distinguish the cases of positive and negative dependence and normalize covariance with (the absolute values of) those bounds.
\begin{definition}[Cole's Coefficient] \label{def:Cole}
	Cole's $\myC$ is defined as
			$$\myC(A,B) = \begin{cases}
		\frac{\Cov(A,B)}{ \min(\P(A),\P(B)) - \P(A) \P(B)} , &\Cov(A,B) \geq 0\\
		\frac{\Cov(A,B)}{ -\max(0,\P(A) + \P(B) -1) + \P(A) \P(B)} , &\Cov(A,B) < 0
	\end{cases}.$$
\end{definition}

The idea of normalizing with the cases of perfect positive and negative dependence has later been used for rank correlations (\citet{VandenhendeLambert2003}, \citet{genest2007primer}) and for Pearson correlation and generalizations of it by \citet{fissler2023generalised}. Indeed, Cole's $\myC$ is a special case of attainable modifications of Kendall's $\tau$, Spearman's $\rho$ and Pearson correlation as well as of threshold correlation (see \citet{genest2007primer} for the former two and \citet{fissler2023generalised} for the latter two).

Cole's coefficient is normalized and attainable by construction and essentially inherits properties (B), (D) and (E) from the covariance such that it is a proper dependence measure.

\begin{proposition} \label{Cole_propriety}
	$\myC(A,B)$ is proper.
\end{proposition}

%$$\mathrm{C}(\mathds{1}_{A}, \mathds{1}_{B}) = \begin{cases}\frac{\P(A\cap B)\P(\overline{A}\cap \overline{B})-\P(\overline{A}\cap B)\P(A\cap \overline{B})}{\min(\P(A)\P(\overline{B}), \P(\overline{A})\P(B))}, \text{if} & \P(A\cap B)\P(\overline{A}\cap \overline{B})\ge \P(\overline{A}\cap B)\P(A\cap \overline{B})\\[0.5em]\frac{\P(A\cap B)\P(\overline{A}\cap \overline{B})-\P(\overline{A}\cap B)\P(A\cap \overline{B})}{\min(\P(A)\P(B), \P(\overline{A})\P(\overline{B}))}, \text{if} & \P(A\cap B)\P(\overline{A}\cap \overline{B})< \P(\overline{A}\cap B)\P(A\cap \overline{B})\\\end{cases}$$

\subsection{Yule's Q, the Odds Ratio and Relatives}

Besides the phi coefficient and Cole's coefficient, Yule's $\myQ$ can be viewed as yet another approach to normalizing covariance. It just replaces the minus in the alternative representation of $\Cov(A,B)$ from \eqref{eq:covariance_representation} with a plus and normalizes by the resulting term. 

\begin{definition}[Yule's $\myQ$] \label{def:Q_original}
	Yule's $\myQ$ is defined as
	\begin{equation} 
		\myQ(A, B)=\frac{\Cov(A,B)}{\P(A\cap B)\P(\overline{A}\cap \overline{B})+\P(\overline{A}\cap B)\P(A\cap\overline{B})}.
	\end{equation}
\end{definition}

Via \eqref{eq:covariance_representation}, Yule's $\myQ$ can be rewritten as
\begin{equation} \label{eq:Q_rewritten}
	\myQ(A, B)=\frac{\P(A\cap B)\P(\overline{A}\cap \overline{B})-\P(\overline{A}\cap B)\P(A\cap \overline{B})}{\P(A\cap B)\P(\overline{A}\cap \overline{B})+\P(\overline{A}\cap B)\P(A\cap\overline{B})}.
\end{equation}

Considering \eqref{eq:Q_rewritten}, it is clear that $\myQ(A,B)$ lies in $[-1,1]$ and fulfills the symmetry properties (E).  Its attainability directly follows by invoking part (iii) of Proposition \ref{prop:characterizations}. The independence property (B) is inherited from the covariance (since the normalization is nonzero). Monotonicity follows from the strictly monotonic relation between Yule's $\myQ$ and the odds ratio (Lemma \ref{lemma:Q_and_OR}) and the monotoniticy of the odds ratio (part (D) of Proposition \ref{prop:properties_OR}).  

\begin{proposition}
	$\myQ(A,B)$ is proper.
\end{proposition}

The normalization of Yule's $\myQ$ is rather ad hoc, in particular compared to the natural normalization of Cole's $\myC$. The normalization does not use the Fr\'{e}chet--Hoeffding bounds for covariance, but continuously changes with the joint probabilities. Thus, the measure is not as nicely interpretable in terms of covariance relative to its values under perfect positive or negative dependence. On the other hand, it avoids the need to use a case distinction in the normalization as Cole's $\myC$ does, which makes asymptotic theory and inference much easier (see Section \ref{sec:asymptotics}). One way to interpret Yule's $\myQ$ is by realizing that it is a special case of Goodman-Kruskal's $\gamma$ \citep{Goodman1954} and can therefore be interpreted in the same vein \citep{bishop2007discrete}.

A very natural way to motivate and interpret Yule's $\myQ$ is via its relation to the odds ratio. The odds ratio is different from the other measures of dependence discussed so far in that it does not build on the covariance and does not fall into the class of classical dependence measures mapping to $[-1,1]$ or $[0,1]$ as it maps to the extended nonnegative real numbers. Even though it thus falls a bit out of the framework of this paper, we nevertheless discuss it here shortly due to its popularity and to it being nicely interpretable and having nice properties. 

\begin{definition}[Odds Ratio] \label{def:odds_ratio}
	The odds ratio is defined as 
	$$\OR(A,B) := \frac{\P(A\cap B)\P(\overline{A}\cap \overline{B})}{\P(\overline{A}\cap B)\P(A\cap \overline{B})}$$
	if $\P(\overline{A}\cap B)\P(A\cap \overline{B}) \neq 0$ and  $\OR(A,B):=\infty$ if $\P(\overline{A}\cap B)\P(A\cap \overline{B})=0$.
\end{definition}

We discuss further details on the motivation behind the odds ratio, its interpretation and popularity in Appendix \ref{subsec:odds_ratio}.
Note that the odds ratio is closely related to the covariance in that it just replaces the subtraction in \eqref{eq:covariance_representation} by a division. Fittingly, it is also sometimes called the cross-product ratio, referring to the underlying contingency table. The odds ratio and Yule's $\myQ$ are closely related: They are strictly monotonic functions of each other.

\begin{lemma} \label{lemma:Q_and_OR}
	It holds that 
	$$\myQ(A,B) = \frac{\OR(A,B)-1}{\OR(A,B)+1} \quad \text{ and } \quad \OR(A,B) = \frac{1 + \myQ(A,B)}{1-\myQ(A,B)}.$$
\end{lemma}
Thus, Yule's $\myQ$ can be viewed as a transformation of the odds ratio to a new scale. Hence, it is unsurprising that the odds ratio is proper in an appropriate sense as well. Indeed, we show that it fulfills a modified set of axioms adopted to its codomain and its multiplicative rather than additive nature in Proposition \ref{prop:properties_OR} in the appendix.

Furthermore, Yule's $\myQ$ has been generalized in the sense of 
\begin{equation*} 
	\myQ_g(A, B)=\frac{ \left(\P(A\cap B)\P(\overline{A}\cap \overline{B})\right)^g-\left(\P(\overline{A}\cap B)\P(A\cap \overline{B})\right)^g}{\left(\P(A\cap B)\P(\overline{A}\cap \overline{B})\right)^g+\left(\P(\overline{A}\cap B)\P(A\cap\overline{B})\right)^g}, \qquad g \in (0,1]
\end{equation*} 
with special cases $\myQ:=\myQ_1$, $\myY:=\myQ_{0.5}$ \citep{yule1912} and $H:=\myQ_{0.75}$ \citep{digby1983approx}.
The generalized Yule's $\myQ_g$ also has a one-to-one relation to the odds ratio and can be shown to be proper along the same lines as Yule's $\myQ$. Further, $|\myQ_g(A,B)|$ is weakly increasing in $g$ and strictly increasing if $\Cov(A, B)\ne 0$, which can be seen by computing the derivative with respect to $g$. 

We discuss a further property of $\myQ$, $\myQ_g$ and the odds ratio, which has been interpreted as a form of invariance to the marginal distributions, in Appendix \ref{subsec:invariance_marginals}.

Another interesting measure is tetrachoric correlation. It is constructed via the assumption that the $2 \times 2$ contingency table arises by dichotomization of an underlying bivariate normal distribution, of which it is the Pearson correlation coefficient. Even though its construction seems a bit artificial, it is quite popular compared to the other proper measures and we discuss it in Appendix \ref{subsec:tetrachoric} and show that it is proper as well. 
It is related to  $\myQ_g$ in that versions of the latter have been used to approximate the former, which has no closed-form solution \citep{digby1983approx}.

%From the measures discussed in this subsection, we focus on Yule's Q for the rest of the paper, bearing in mind that it is a representative of a larger class of closely related measures including the odds ratio and Yule's Y.

\subsection{Comparison of $\phi$, $\myC$ and $\myQ$} 
\label{subsec:comparison}

It is possible to directly compare the discussed dependence measures in terms of the values they take since they only depend on the marginal event probabilities $\P(A)$ and $\P(B)$ and the joint probability $\P(A \cap B)$ or the value of one of the competing dependence measures, respectively. In Figure \ref{fig:colecomparison07} we plot Yule's $\myQ$ and the phi coefficient for a fixed value of Cole's $\myC$ of 0.7 over all combinations of event probabilities $\P(A)$ and $\P(B)$.

\begin{figure}[tb]
	\centering
	\includegraphics[width=1\linewidth]{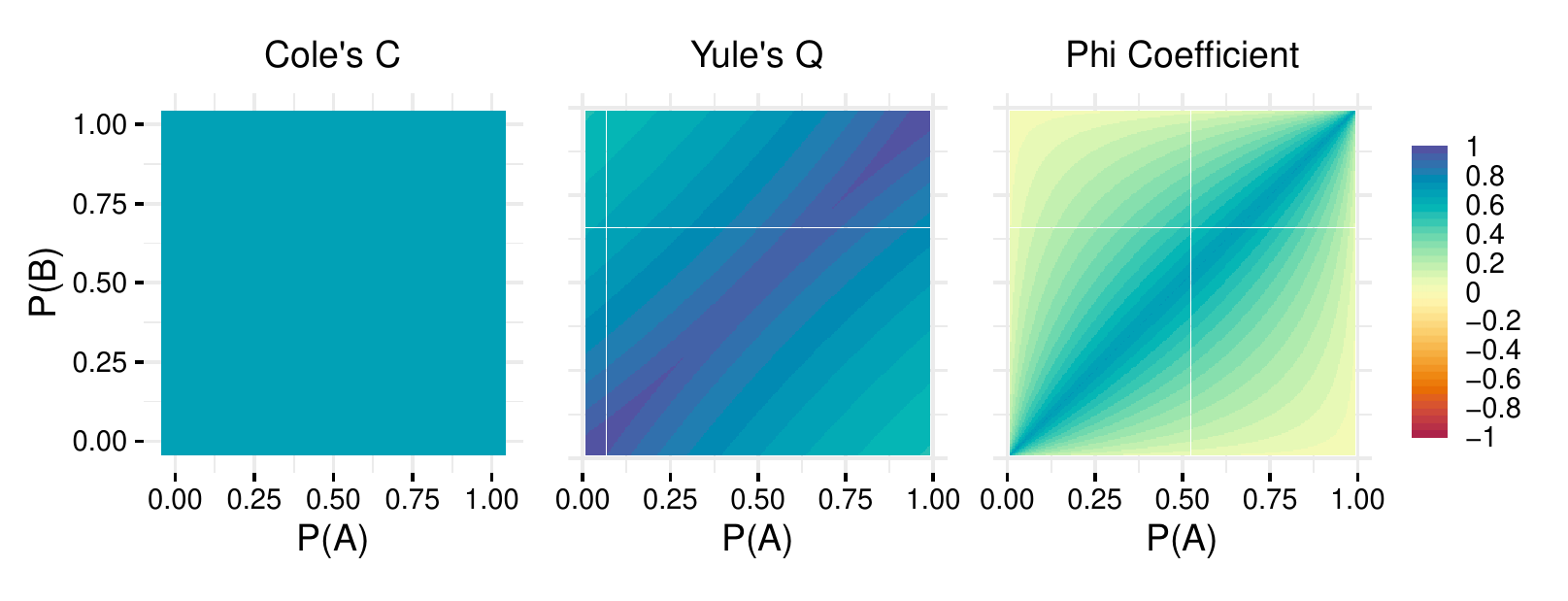}
	\caption{This figure depicts Cole's $\myC$ with a fixed value of 0.7 and the corresponding values of Yule's $\myQ$ and the phi coefficient for all combinations of marginal event probabilities $\P(A)$ and $\P(B)$.}
	\label{fig:colecomparison07}
\end{figure}

Firstly, the distinction between the two proper measures and the phi coefficient is clearly visible. The serious attainability problems of the phi coefficient (see also Figure \ref{fig:phibounds}) can lead to it having very small values despite the quite strong positive dependence. The attainability problems become more serious the further away we move from the diagonal, that is from the condition $\P(A)=\P(B)$ (also see the discussion around Proposition \ref{prop:characterizations_linear_dependence}). 

Secondly, even though the two proper measures are much closer to each other, they of course also take different values in general. In particular, for equal and at the same time very small or large event propabilities, Yule's $\myQ$ takes values close to 1 instead of close to 0.7.

For different values of Cole's coefficient qualitatively the same picture arises. We show the cases where it takes the values 0.3 and $-0.7$ in Figures \ref{fig:colecomparison03} and \ref{fig:colecomparison-07} in the Appendix. The latter figure in comparison with Figure \ref{fig:colecomparison07} also illustrates the symmetry relation $\delta(A,\overline{B})=-\delta(A,B)$ formulated in the second part of property (E), which all the coefficients possess.

\section{Statistical Inference} 
\label{sec:asymptotics} 

%The basis for the asymptotics for all three coefficients is a central limit theorem for $\mathbf{W}_i := (X_i, Y_i, X_i Y_i)^\top$, $i \in \mathbb{N}$ (see Lemma \ref{lemma:CLTjoint}), which follows from Assumption \ref{ass:iid} and by the covariance matrix  
%\begin{equation} 
%	\label{eq:variance_matrix_iid}
%	\Omega := \Var \left( \mathbf{W}_i  \right) =
%	\begin{pmatrix}
%		p (1-p) & r - p q & r (1-p) \\
%		r - p q & q (1-q) & r (1-q) \\
%		r (1-p) & r (1-q) & r (1-r)
%	\end{pmatrix}
%\end{equation}
%being positive definite and finite, which is ensured by Assumption \ref{ass:BoundedAway} as well. Assumption \ref{ass:PositiveRelativeFrequencies} is just an in-sample non-degeneracy condition, which makes sure that the plug-in estimators below exist.

\subsection{Estimators and Asymptotic Distributions} 
\label{subsec:asymptotic_distributions}

We consider a sequence of binary random variables $(X_i, Y_i)$ that are distributed as $(X,Y) = (\mathds{1}_A, \mathds{1}_B)$ for all $i \in \mathbb{N}$ and use the shorthand notations $p := \Pb(A)$, $q := \Pb(B)$, $r := \Pb(A \cap B)$, $\sigma := \Cov(A,B) = r - p q$, $\mp := \min(p, q) - pq$ and $\mn :=pq - \max(0,p+q-1)$. 
For a given sample of size $n \in \mathbb{N}$, we introduce the natural estimators $\p := \frac{1}{n} \sum_{i=1}^n X_i$, $\q := \frac{1}{n} \sum_{i=1}^n Y_i$ and $\r := \frac{1}{n} \sum_{i=1}^n X_i Y_i$ as well as $\hsigma := \r - \p \q$, $\hmp := \min(\p, \q) - \p \q$ and $\hmn := \p \q - \max(0, \p + \q- 1)$.
Assume that $0 < \p,\q < 1$
% \begin{align}
%     \label{eq:PositiveRelativeFrequencies}
%     0 < \p,\q < 1
%     \qquad \text{ and } \qquad 
%     \max(\p + \q - 1, 0) < \r < \min(\p, \q)
% \end{align}
holds such that we can define the following plug-in estimators for $\myQ$, $\myC$, and $\phi$ based on Definitions \ref{def:phi_coefficient},  \ref{def:Cole} and \ref{def:Q_original}, respectively,
%We use the following plug-in estimators for $\myQ$, $\myC$, and $\phi$ that arise by replacing $p$, $q$ and $r$ with their estimators $\p$, $\q$ and $\r$ in their Definitions \ref{def:phi_coefficient},  \ref{def:Cole} and \ref{def:Q_original}, respectively,
\begin{align}
	\label{eqn:Qest}
	\myQn &= \frac{\hsigma}{\r(1-\p-\q+\r) + (\q-\r)(\p-\r)}, \\
	\label{eqn:Cest}
	\myCn &= \frac{\hsigma}{\hmp} \mathds{1}\{\hsigma \ge 0\} + \frac{\hsigma}{\hmn} \mathds{1}\{\hsigma < 0\},  \\
	\label{eqn:Phiest}
	\widehat{\phi}_n &= \frac{\hsigma}{\sqrt{\p(1-\p)\q(1-\q)}},
\end{align}
that are well-defined given the condition $0 < \p,\q < 1$, which holds asymptotically with probability one (see \eqref{eqn:WellDefinedProbOne} for details) by the following assumption.
%\todo{I checked that the well-definedness of all estimators follows under $0 < \p,\q < 1$. I think it should be okay to state it without proof.}
%\lukas{I agree}

% where we use the estimated normalizations which are imposed to be strictly positive by \eqref{eq:PositiveRelativeFrequencies}.\footnote{If $\p, \q \in \{0,1\}$, $\myQn$, $\myCn$ and $\widehat{\phi}_n$ are undefined. If  $\r \in \{\max(\p + \q - 1, 0), \min(\p, \q)\}$, we set $\myQn = \myCn  = \operatorname{sign}(\hsigma$), i.e, $\pm 1$, and $\widehat{\phi}_n$ equals the estimated counterparts of its Fr\'{e}chet--Hoeffding bounds in \eqref{eq:phibounds}.}

% $\hmp := \min(\p, \q) - \p \q$ and $\hmn := \p \q - \max(0, \p + \q- 1)$ that are imposed to be strictly positive by Assumption \ref{ass:PositiveRelativeFrequencies}.

%We impose the following assumptions to derive the asymptotic distributions of plug-in estimators for $\myQ$, $\myC$, and $\phi$ that arise by replacing $p$, $q$ and $r$ with their estimators $\p$, $\q$ and $\r$ in their formulas in Definitions \ref{def:phi_coefficient},  \ref{def:Cole} and \ref{def:Q_original}, respectively.

\begin{assump}
	\label{ass:BoundedAway}
	Let $0 < p, q < 1$ and $\max(0, p + q - 1) <  r < \min(p, q)$.
\end{assump}	

% \begin{assump}
% 	\label{ass:PositiveRelativeFrequencies}
% 	Let $0<\p,\q,\r < 1$.
% \end{assump}

Assumption \ref{ass:BoundedAway} is imposed to rule out boundary effects: Violating the conditions $0 < p, q < 1$ would imply that the corresponding estimators $\p, \q$ are almost surely zero or one.
The condition $\max(0, p + q - 1) =  r$ would imply perfect negative dependence and $r = \min(p, q)$ perfect positive dependence such that our estimators for $\myQ$ and $\myC$ (and $\phi$ in the attainable case) would be either plus or minus 1 almost surely, and hence their asymptotic distributions would be degenerate. 
%The finite sample condition in Assumption \ref{ass:PositiveRelativeFrequencies} is required to assure that the estimator $\myCn$ in \eqref{eqn:Cest} below is well defined.
% are almost surely constant (plus or minus one for the former two and equal to the Fr\'{e}chet--Hoeffding lower or upper bound for the latter), and so their asymptotic distributions are degenerate.

The basis for the asymptotics for the three dependence coefficients in \eqref{eqn:Qest}--\eqref{eqn:Phiest} is a central limit theorem (CLT) for $\frac{1}{n} \sum_{i=1}^n \mathbf{W}_i$ with $\mathbf{W}_i := (X_i, Y_i, X_i Y_i)^\top$, $i \in \mathbb{N}$, in Lemma \ref{lemma:CLTjoint}
that can be derived under either one of the following two conditions:
\begin{assump}
	\label{ass:iid}
	Let $(X_i, Y_i)_{i \in \mathbb{N}}$ be a sequence of independent and identically distributed (iid) random variables and let $\Omega := \Var \left( \mathbf{W}_i  \right)$.
\end{assump}

\begin{assump}
	\label{ass:Dependence}
	Let $(X_i, Y_i)_{i \in \mathbb{N}}$ be a stationary ergodic adapted mixingale of size $-1$ and let the long-run covariance matrix  $\Omega := \lim_{n\to \infty} \Var \left( n^{-1/2} \sum_{i=1}^n  \mathbf{W}_i  \right)$ be strictly positive definite.
\end{assump}

%\begin{assump}
%	\label{ass:Multicolinearity}
%	For the sequence $\mathbf{W}_i := (X_i, Y_i, X_i Y_i)^\top$, $i \in \mathbb{N}$, let the ``long-run'' covariance matrix  $\Omega := \Var \left( n^{-1/2} \sum_{i=1}^n  \mathbf{W}_i  \right)$ % be strictly positive definite with 
%	have Eigenvalues bounded away from zero.
%\end{assump}

% \marc{It seems a bit strange/inconsistent to me that we write down the formulas for all the matrices that occur in the asymptotic variance matrices in the following propositions, but not the formula(s) for $\Omega$. In the iid case, we state it in a proof in the Appendix and in the time series case, we only provide the HAC estimator, but not the theoretical formula. I think we should at least provide the latter in the Appendix as well and refer to both formulas for $\Omega$ to the Appendix. We could include the formulas in Lemma \ref{lemma:CLTjoint}. What do you think, Timo?}

In the iid case of Assumption \ref{ass:iid}, $\Omega = \Var \left( \mathbf{W}_i  \right)$ is positive definite through Assumption \ref{ass:BoundedAway} that assures that the components of $\mathbf{W}_i$ are not perfectly dependent; see the proof of Lemma \ref{lemma:CLTjoint} for details.
Instead, an equivalent condition has to be imposed on the long run variance in the time series case of Assumption \ref{ass:Dependence}.
See \citet[Section 5.3]{White2001} for details on mixingale processes.
While time series asymptotics usually come at the cost of more restrictive moment conditions, these are of no concern here due to the boundedness of the binary random variables. 
We further discuss the time series case in Appendix \ref{subsec:time_series_asymptotics}.

%\color{red}
%Discussion of Assumptions:
%\begin{itemize}
%	\item If $p,q,r \in \{0,1\}$, then the related estimators are almost surely 0 or 1.
%	\item  $r < \min(p,q)$  holds by definition, and $r=p$ (and $r=q$ equivalently) would imply perfect dependence as $A \subseteq B$ a.s.
%	\item $r > \max(0, p+q-1)$ is an equivalent upper bound if $p$ and/or $q$ are too large. Then, $r$ must also be large as can probably be seen most conveniently in a Venn diagram...
%	\item ! We probably do not have a problem with the case  $r = pq$, then $\Cov(A,B) = $. This should NOT be a problem for Yule's $Q$!
%	\item Discuss mixingale, see \citet[Section 5.3]{White2001} 
%\end{itemize}
%\color{black}

%We first consider the plug-in estimator for Yule's $\myQ$ based on Definition \ref{def:Q_original},
%\begin{align*}
%	\myQn &= \frac{\hsigma}{\r(1-\p-\q+\r) + (\q-\r)(\p-\r)}.
%\end{align*}

%\todo{Please check my arguments for the case of a degenerate finite sample case for the asymptotics around \eqref{eqn:WellDefinedProbOne} in the proof.}
%\lukas{I like that}

We start by giving the asymptotic distribution of $\myQn$, which is Gaussian.
\begin{proposition}
	\label{prop:Asymptotics-Q}
%	Given Assumptions \ref{ass:BoundedAway}--\ref{ass:iid}, it holds that
	Given Assumption \ref{ass:BoundedAway} and either Assumption \ref{ass:iid} or Assumption \ref{ass:Dependence}, it holds that
	\begin{align*}
		& \big(J_g \Omega J_g^\top \big)^{-1/2} \sqrt{n} \big( \myQn - \myQ \big) \stackrel{d}{\to} \mathcal{N} \big( 0, 1\big), \qquad \text{where} \\
		&J_g := \frac{2}{\big(p (q - 2 r) + r (1 -2 q + 2 r)\big)^2}
		\begin{pmatrix}
			(q - 1) r (q - r) \\
			(p - 1) r (p - r) \\
			- (p^2 q + p q (-1 + q - 2 r) + r^2)
		\end{pmatrix}^\top.
	\end{align*}
	%is non-zero for all $(p,q,r)$ satisfying Assumption \ref{ass:BoundedAway}.
\end{proposition}
% \lukas{Why do we mention non-zero here but not in the the next two propositions?}
% \marc{Good question. Any reason for that, Timo? I commented this sentence out here, but please check.}
% \timo{It would strictly speaking be necessary (in all theorems) to invert the matrix in the line above I think. But we can leave it out as the non-zero-ness is rather obvious for the first two components of $J_g$...}

As $\myQ$ is a continuously differentiable transformation of $p, q$ and $r$ (with Jacobian matrix $J_g$ at $(p,q,r)$), the proof of Proposition \ref{prop:Asymptotics-Q} is a straightforward application of the delta method, but is complicated by the need to show non-degeneracies in the Jacobian $J_g$.
%As $J_g$ is shown to be non-zero in the domain of interest, this implies that the resulting limiting normal distribution is non-degenerate.

In contrast, the asymptotic distribution of $\myCn$ (and its derivation) is much more involved.

\begin{proposition}
	\label{prop:Asymptotics-C}
	Let $Z \sim \mathcal{N}(0,1)$, $\mathbf{V} = (V_1,V_2,V_3,V_4)^\top \sim J_f \Omega^{1/2} \mathbf{Z}$, where $\mathbf{Z} \sim \mathcal{N}(0, I_3)$ and
	\begin{align*}
		\Lambda^+ &:= 
		\begin{pmatrix}
			\frac{1}{m^+}  & 
			- \frac{\sigma}{m^+  m^+}
		\end{pmatrix}, 
		\qquad 
		J_{h^+} :=
		\begin{pmatrix}
			-q & -p & 1 \\
			\mathds{1}(p < q) - q& \mathds{1}(q < p) - p & 0 
		\end{pmatrix}, \\
		\Lambda^- &:=  
		\begin{pmatrix}
			\frac{1}{m^-}  & 
			- \frac{\sigma}{m^-  m^-}
		\end{pmatrix}, 
		\qquad
		J_{h^-} :=
		\begin{pmatrix}
			-q  & -p & 1 \\
			q -	\mathds{1}(p+q>1)   & 	p -	\mathds{1}(p+q>1) & 0 
		\end{pmatrix}, \\
		\Delta &:= (-q, -p, 1)^\top,
		\qquad \qquad \;\,
		J_f :=
		\begin{pmatrix}
			1 & 0 & 0 \\
			0 & 1 & 0  \\
			q & p & 0 \\
			-q & -p & 1
		\end{pmatrix}.
	\end{align*} 	
	Then, given Assumption \ref{ass:BoundedAway} and either Assumption \ref{ass:iid} or Assumption \ref{ass:Dependence}, it holds that:
	%given Assumptions \ref{ass:BoundedAway}--\ref{ass:iid}, the following holds:
	\begin{itemize} 
		%		\item If $Cov(A,B) = 0$, then $\big(\Delta^\top \Omega \Delta \big)^{-1/2} \sqrt{n} \big(\widehat{C}_n - C \big) \tod |Z_+| - |Z_-|$, where $Z_+ \sim \mathcal{N} \big( 0, (m^+)^{-2} \big)$ and $Z_- \sim \mathcal{N} \big( 0, (m^-)^{-2} \big)$. 
		\item If $\sigma = 0$, then 
		$\big(\Delta^\top \Omega \Delta \big)^{-1/2} \sqrt{n} \big(\myCn - \myC \big) \tod \frac{Z \, \mathds{1}\{Z > 0\}}{m^+} +  \frac{Z \, \mathds{1}\{Z < 0\}}{m^-}$.
		
		\item If $\sigma > 0$, and
		\begin{itemize}
			\item 
			if $p \not= q$, then $\big( \Lambda^{+} J_{h^+} \Omega J_{h^+}^\top  (\Lambda^{+})^\top \big)^{-1/2} \sqrt{n}  \big( \myCn - \myC \big)
			\stackrel{d}{\to} \mathcal{N} \big( 0, 1 \big)$; and
			
			\item 
			if $p = q$, then 
			$\sqrt{n} \big(\myCn - \myC \big) \tod	\Lambda^+ 
			\begin{pmatrix}
				V_{4} \\ 
				\big[ V_{1} - V_{3} \big]  - \mathds{1} \{ V_{1} > V_{2} \} \big[ V_{1} - V_{2} \big]
			\end{pmatrix}$.
%			\begin{pmatrix}
%	V_{4} \\ 
%	\mathds{1} \{ V_{1} \le V_{2} \} \Big[ V_{1} - V_{3} \Big]  
%	+ \mathds{1} \{ V_{1} > V_{2} \} \Big[ V_{2} - V_{3} \Big]
%\end{pmatrix}$.		
		\end{itemize}

		\item If $\sigma < 0$, and
		\begin{itemize}
			\item 
			if $p \not= 1-q$, then 
			$\big( \Lambda^{-} J_{h^-} \Omega J_{h^-}^\top  (\Lambda^{-})^\top \big)^{-1/2} \sqrt{n} 
			\big( \myCn - \myC \big)
			\stackrel{d}{\to} \mathcal{N} \big( 0, 1 \big)$; and
			
			\item
			if $p = 1-q$, then
			$\sqrt{n} \big(\myCn - \myC \big) \tod	\Lambda^-  
			\begin{pmatrix}
				V_{4} \\ 	
				 V_{3} - \mathds{1} \{ V_{1} > - V_{2} \} \big[  V_{1} + V_{2} \big]
			\end{pmatrix}
%			\begin{pmatrix}
%				V_{4} \\ 
%				\mathds{1} \{ V_{1} \le - V_{2} \} V_{3}
%				+ \mathds{1} \{ V_{1} > - V_{2} \} \Big[ V_{3} - V_{1} - V_{2} \Big]
%			\end{pmatrix}
$.
		\end{itemize}
	\end{itemize} 	
\end{proposition}

Proposition \ref{prop:Asymptotics-C} shows that the asymptotic behavior of $\myCn$ is non-standard and exhibits discontinuities that depend on the unknown population quantities $\sigma$, $p$ and $q$.
The case distinction in whether $\sigma$ (or equivalently $\myC$) is zero, positive or negative arises due to the differing normalizations in the definition of $\myC$:
For the case of $\myC = 0$, the asymptotic distribution is given by two half-normals with different variances.
%This allows for a straightforward test for the null hypothesis $C = 0$.
For $\{ \myC > 0,\, p \not= q \}$ and $\{ \myC < 0,\, p \not= 1 - q \}$, we obtain a standard asymptotic normal distribution.
We however get non-standard asymptotic distributions for the cases $\{ \myC > 0,\, p = q \}$ and $\{ \myC < 0,\, p = 1 - q \}$ that arise due to the minimum and maximum operators in the normalizations $\mp$ and $\mn$.

\begin{figure}[tb]
	\centering
	\includegraphics[width=\linewidth]{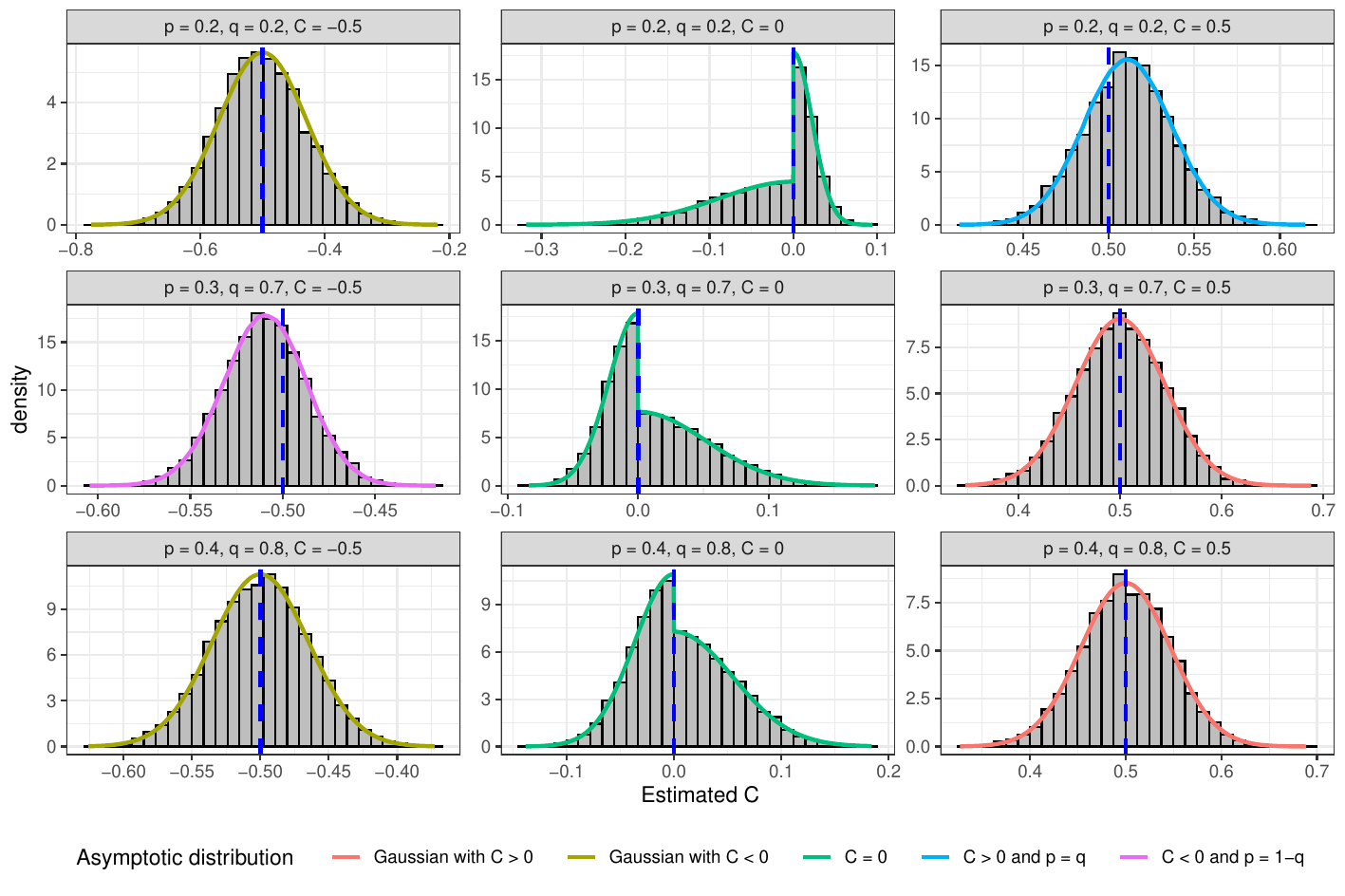}
	\caption{This figure illustrates the different asymptotic distributions of the estimated Cole coefficient $\myCn$ depending on three choices of the marginal probabilities $(p, q)$ coupled with $\myC \in \{-0.5, 0, 0.5\}$ in the different subplots.
	The histograms show the estimated values of $\myCn$ for a sample size of $n=2{,}000$ for $10{,}000$ simulation replications.
	The plots are augmented with the respective asymptotic densities and the vertical dashed blue lines indicate the true values of $\myC$.}
	\label{fig:C_AsyDist}
\end{figure}

Figure \ref{fig:C_AsyDist} illustrates the five different cases of the asymptotic distribution through simulations. 
Besides the standard asymptotic normality, we can clearly see the two half-normal distributions for $\myC=0$, where the difference of the variances depends on the specific setting governed by $p$ and $q$.
The complicated cases $\{ \myC > 0,\, p = q \}$ and $\{ \myC < 0,\, p = 1 - q \}$ arise in the upper right and middle left subplots. There we also find a bias in the estimates, which is reflected in the asymptotic distributions through the terms $\mathds{1} \{ V_{1} > V_{2} \} \big[ V_{1} - V_{2} \big]$ and $\mathds{1} \{ V_{1} > - V_{2} \} \big[  V_{1} + V_{2} \big]$, respectively.
For $\myC > 0$, this bias can be explained by the normalization $\hmp = \min(\p, \q) - \p \q > 0$, which is negatively biased through $\min(\p, \q)$, which tends to underestimate its population counterpart $\min(p, q)$ in the case $p = q$. The argument is similar for the case $\myC < 0$.

%The plug-in estimator for the $\phi$ coefficient from Definition \ref{def:phi_coefficient} is 
%\begin{align*}
%	\widehat{\phi}_n &= \frac{\hsigma}{\sqrt{\p(1-\p)\q(1-\q)}}.
%\end{align*}

For completeness, we also report the asymptotic distribution of $\widehat{\phi}_n$, whose proof is a straightforward application of the delta method.
\begin{proposition}
	\label{prop:Asymptotics-Phi}
	Given Assumption \ref{ass:BoundedAway} and either Assumption \ref{ass:iid} or Assumption \ref{ass:Dependence}, it holds that
	\begin{align*}
		& \big(J_l \Omega J_l^\top \big)^{-1/2} \sqrt{n} \big( \widehat{\phi}_n - \phi \big) \stackrel{d}{\to} \mathcal{N} \big( 0, 1\big), \qquad \text{where} \\
		&J_l :=  \frac{1}{\sqrt{p(1-p)q(1-q)}} \left( 		
			\begin{pmatrix}
				-q \\
				-p \\
				1
			\end{pmatrix}^\top 
			+ \frac{pq-r}{2p(1-p)q(1-q)}	
			\begin{pmatrix}
				(1-p)q(1-q) - pq(1-q)\\
				p(1-p)(1-q) - p(1-p)q \\
				0
			\end{pmatrix}^\top 
			\right).
	\end{align*}
\end{proposition}

\subsection{Feasible Inference, Tests and Confidence Intervals} 
\label{subsec:CIs}

%$\widehat{Q}_n$, $\widehat{C}_n$ and $\widehat{\phi}_n$ 

Feasible inference  requires consistent estimates of the asymptotic distributions. % (through their covariance matrices).
Denote by $\widehat{J}_{n, g}$, $\widehat{J}_{n, h^+}$, $\widehat{J}_{n, h^-}$, $\widehat{J}_{n,f}$, $\widehat{J}_{n,l}$, $\widehat{\Lambda}_n^+$, $\widehat{\Lambda}_n^-$,  $\hmn$, $\hmp$,  and $\widehat{\Delta}_n$ the estimated counterparts of the vectors and matrices defined in Propositions \ref{prop:Asymptotics-Q}, \ref{prop:Asymptotics-C} and \ref{prop:Asymptotics-Phi} that arise by simply replacing the true probabilities $p$, $q$ and $r$ by their estimators $\p$, $\q$ and $\r$. 
In the iid case, consistent estimation of $\Omega$ through a plug-in estimator using $\p, \q$ and $\r$ in \eqref{eq:variance_matrix_iid}  in Appendix \ref{app:proofs} is straight-forward.
%the consistent sample covariance is straightforward.
In the time series case, we employ the well-known ``HAC''-estimator of \citet{NeweyWest1987} whose consistency is shown in Proposition \ref{prop:CovConsistency}.
The continuous mapping theorem implies that the empirical covariance matrices in Propositions \ref{prop:Asymptotics-Q}, \ref{prop:Asymptotics-C} and \ref{prop:Asymptotics-Phi} are consistent estimators of their population counterparts, which ensures that the asymptotic distributions are unchanged when the population matrices are replaced with these estimators. 

Feasible inference in the form of testing and confidence intervals for $\myQ$ and $\phi$ can be carried out by standard methods based on the above asymptotic normality statements.
In contrast, the case distinctions in the asymptotic distribution of $\myCn$ in Proposition \ref{prop:Asymptotics-C} complicate the construction of tests and confidence intervals as the population quantities $\myC$, $p$ and $q$ are unknown.

For confidence intervals, we leverage the idea of building them through inverted hypothesis tests by running a sequence of tests (described below) with null hypotheses $\H_{c} = \{ \myC = c \}$ for all (or in practice, a fine sequence of) values $c \in (-1,1)$. 
The confidence interval at level $1 - \alpha \in (0,1)$ then consists of all values $c \in (-1,1)$, for which the hypothesis $\H_{c}$ is not rejected at level $\alpha$. 

Tests for $\H_0 = \{\myC=0\}$ (or for $\{\myC \le 0\}$ or $\{\myC \ge 0$\}) are straightforward using the respective limiting distribution for $\myC = 0$ in Proposition \ref{prop:Asymptotics-C}.\footnote{With the same reasons as given after ``second'' in the following paragraph, we write $\H_0 = \H_0 \cap \{\sigma = 0\}$ and include a two-sided test for $\{\sigma = 0\}$ and obtain the final $p$-value with a Bonferroni correction.}
When testing for $\H_{c}$ for any $c>0$,
%(and similarly for any $c<0$), 
we are however uncertain whether we are in the case $\{p = q\}$ or $\{p \not= q\}$ such that a valid testing approach must accommodate both possibilities. 
Hence, we write $\H_c = \H_= \cup H_{\not=}$ with $\H_{=} = \{ \myC=c, \, p=q \}$ and $\H_{\not=} = \{ \myC=c, \, p\not=q \}$, where $\H_{=}$ and $\H_{\not=}$ can be tested using the asymptotically valid $p$-values $p_=$ and $p_{\not=}$ that arise from the respective asymptotic distributions in Proposition \ref{prop:Asymptotics-C}.
Thus, testing for $\H_c$ through the union of these two hypotheses becomes an intersection-union (IU) test, for which $\max(p_{=}, p_{\not=})$ is an asymptotically valid $p$-value \citep{Berger1982}. %for which no multiplicity adjustment (as in Bonferroni corrections) is necessary 

This testing procedure is %in practice
however overly conservative for two reasons:
First, given the predominant case $\{p \not= q\}$, the $p$-value $p_{=}$ is often non-informative and dominates the IU-construction $\max(p_{=}, p_{\not=})$.
We resolve this by writing $\H_= = \H_= \cap \H_{p,q}$ with $\H_{p,q} = \{p=q\}$, which we test for by the joint CLT for $(\p, \q)$ in Lemma \ref{lemma:CLTjoint} in the Appendix with $p$-value $p_{p,q}$.
Second, the estimated asymptotic variances often differ substantially for positive and negative values of $c$ (see e.g., the middle column of Figure \ref{fig:C_AsyDist}), resulting in cases where for a large positive value of $\myC$, we reject many positive values $0 < c < \myC$, but cannot reject values $c \le 0$ due to a higher estimated variance in the negative case.
We solve this by writing $\H_c = \H_c \cap {\H}_{\sigma}$ with ${\H}_{\sigma} = \{ \sigma \ge 0 \}$, which is a necessary condition for $c > 0$.
A test for ${\H}_{\sigma}$ with $p$-value $p_\sigma$ is easily constructed based on the CLT for $\widehat{\sigma}_n$ provided in Lemma \ref{lemma:CLT_sigma} in the Appendix. 
Both these refinements are union-intersection (UI) tests, for which Bonferroni corrections are necessary.
Hence, our final $p$-value is given by 
$2 \min \{ \max[ 2 \min (p_=, p_{p,q}), p_{\not=}], p_\sigma \}$ corresponding to the hypothesis
$\H_c = \{ [ (\H_= \cap H_{p,q}) \cup \H_{\not=}] \cap \H_\sigma \}$.

For any $c < 0$, we use an equivalent construction with $\H_{=} = \{ \myC=c, \, p=1-q \}$ and $\H_{\not=} = \{ \myC=c, \, p\not=1-q \}$, $\H_{p,q} = \{p=1-q\}$ and $\H_\sigma = \{ \sigma \le 0\}$.

In Appendix \ref{sec:Simulations}, we illustrate the good finite sample properties for the standard confidence intervals for $\myQ$ and $\phi$ and for the more refined construction for $\myC$.
For the latter, we also demonstrate the superiority of our approach based on the two auxiliary UI-tests for $\H_{p,q}$ and ${\H}_{\sigma}$ in terms of length of the resulting confidence intervals.

\subsection{Fisher Transformations} 
\label{sec:FisherTransformations}

As the sampling distributions of the estimated measures are bounded, they are in finite samples often heavily influenced by the theoretical upper and lower bounds when the true measures are close (relative to the sample size) to these bounds.
This leads to skewed distributions, for which asymptotic normality (and the more complicated distributions for $\myCn$) are bad approximations and hence deteriorates the finite sample performance of the associated tests and confidence intervals.
To combat this problem, we adopt the classical Fisher transformation \citep{fisher1915}, originally proposed for Pearson correlation, 
\begin{equation} 
	\label{eq:Fisher_transformation}
	Z(\delta) := \mathrm{arctanh}(\delta) = \frac 1 2 \log \left( \frac{1+\delta}{1-\delta} \right)
\end{equation}
for a generic (population or estimated) dependence measure $\delta \in \{\myQ, \myC, \phi, \myQn, \myCn, \widehat{\phi}_n \}$.
The $\mathrm{arctanh}(\cdot)$ in \eqref{eq:Fisher_transformation} bijectively maps the interval $(-1,1)$ to $\mathbb{R}$ and hence provides a better approximation when the dependence measure is close to the boundary.

We provide additional details on the Fisher transformation together with the limiting distributions of  $Z(\myQn)$, $Z(\myCn)$  and $Z(\widehat{\phi}_n)$ in Propositions \ref{prop:Asymptotics-Z_Q}, \ref{prop:Asymptotics-Z_C} and \ref{prop:Asymptotics-Z_phi} in Appendix \ref{subsec:Fisher_transformations} by employing the delta method.
These distributions can be used to construct tests and confidence intervals analogously to Section \ref{subsec:CIs}.

\section{Case Study on Drug Use} \label{sec:application}

%\todo{Lukas: Data Availability check!}

We use data from the 2018 wave of the Monitoring the Future (MTF) study on drug use among the contemporary American youth \citep[DS1 Core Data]{miech2019}.\footnote{We use data from 2018 because subsequent waves have sizeably fewer observations.} We are interested in interdependence between consumption of different drugs and consider the events of ever having consumed certain drugs. The sample size is $14{,}502$, but for certain drugs there is a large number of non-responses.\footnote{The dataset is an aggregation of six different questionnaires students get randomly assigned to. Some drugs do not appear in all six questionnaires.} %\footnote{The survey has an ordinal scale in that the question ``On how many occasions (if any) have you ...?'' allows for the answers 0, 1-2, 3-5, 6-9, 10-19, 20-39 and 40+. In order to fit our setting, we pooled the latter six bins.}
The different drugs, the abbreviations we use for them in subsequent plots, their marginal relative consumption frequencies and the respective (univariate) sample sizes are listed in Table \ref{tab:marginals}.

% \begin{table}[tb]
% 	\centering
% 	\begin{threeparttable}
% 		\begin{tabular}{ccc}
% 			\toprule\toprule
% 			\textbf{drug} & \textbf{relative frequency}  & \textbf{sample size} \\
% 			\midrule
% 			alcohol (ALC) & $0.60$ &$13{,}219$ \\
% 			marijuana (MAR) & $0.44$ & $13{,}395$\\
% 			cigarettes (CIG) & $0.24$ & $13{,}680$\\
% 			amphetamines (AMP) & $0.09$ &$13{,}458$ \\
% 			tranquilizers (TRQ) & $0.07$ & $13{,}408$ \\
% 			narcotics (NAR) & $0.06$ & $13{,}305$\\
% 			LSD (LSD) & $0.05$ & $11{,}210$ \\
% 			psychedelics$^*$ (PSY) & $0.04$ & $11{,}202$ \\
% 			sedatives (SED) & $0.04$ & $13{,}476$\\
% 			MDMA (MDM) & $0.04$ & \;\,$6{,}736$\\
% 			cocaine (COK) & $0.04$ &$13{,}361$  \\
% 			crack (CRA) & $0.01$ &$12{,}920$ \\
% 			methamphetamines (MET) & $0.01$ & \;\,$4{,}446$\\
% 			heroin (HER) & $0.01$ & $13{,}444$\\
% 			\bottomrule
% 		\end{tabular}
% 		\begin{tablenotes}
% 			\item[*] without LSD
% 		\end{tablenotes}
% 	\end{threeparttable}
% 	\caption{This table lists all drugs that are considered in this case study together with their abbreviations, their relative frequencies of consumption and their respective number of replies.} 
% 	\label{tab:marginals}
% \end{table}

\begin{table}[tb]
	\centering
    \small 
	\begin{threeparttable}
		\begin{tabular}{cccc}
			\toprule
			\textbf{drug} & \textbf{abbreviation} & \textbf{relative frequency}  & \textbf{sample size} \\
			\midrule
			alcohol & ALC & $0.60$ &$13{,}219$ \\
			marijuana & MAR & $0.44$ & $13{,}395$\\
			cigarettes & CIG & $0.24$ & $13{,}680$\\
			amphetamines & AMP & $0.09$ &$13{,}458$ \\
			tranquilizers & TRQ & $0.07$ & $13{,}408$ \\
			narcotics & NAR & $0.06$ & $13{,}305$\\
			LSD & LSD & $0.05$ & $11{,}210$ \\
			psychedelics$^*$ & PSY & $0.04$ & $11{,}202$ \\
			sedatives & SED & $0.04$ & $13{,}476$\\
			MDMA & MDM & $0.04$ & \;\,$6{,}736$\\
			cocaine & COK & $0.04$ &$13{,}361$  \\
			crack & CRA & $0.01$ &$12{,}920$ \\
			methamphetamines & MET & $0.01$ & \;\,$4{,}446$\\
			heroin & HER & $0.01$ & $13{,}444$\\
			\bottomrule
		\end{tabular}
		\begin{tablenotes}
			\item[*] without LSD
		\end{tablenotes}
	\end{threeparttable}
	\caption{This table lists all drugs that are considered in this case study together with their abbreviations, their relative frequencies of consumption and their respective numbers of replies.} 
	\label{tab:marginals}
\end{table}

%\todo{I edited the table and put the abbreviation as a separate column. What do you think? The old table is in comments...}

Figure \ref{fig:drugs_correlations} shows contour plots of estimated correlation matrices %\todo{is correlation matrices a good word here?}
between all pairs of those events for Cole's $\myC$, Yule's $\myQ$ and the phi coefficient. All three measures indicate a positive relationship between the consumption of all the drugs considered. However, most values of the phi coefficient are very close to 0, which would usually be interpreted as evidence for very weak dependence, while the proper dependence measures tend to have very high values, indicating strong dependence. Thus, the dependence is structurally underrated by the phi coefficient due to its severe attainability issues. If the phi coefficient or one of the equally popular contingency coefficients was used in such a context, it could thus lead to possibly severe misperceptions with respect to the nature of polydrug use.

\begin{figure}[tb]
	\centering
	\includegraphics[width=\textwidth]{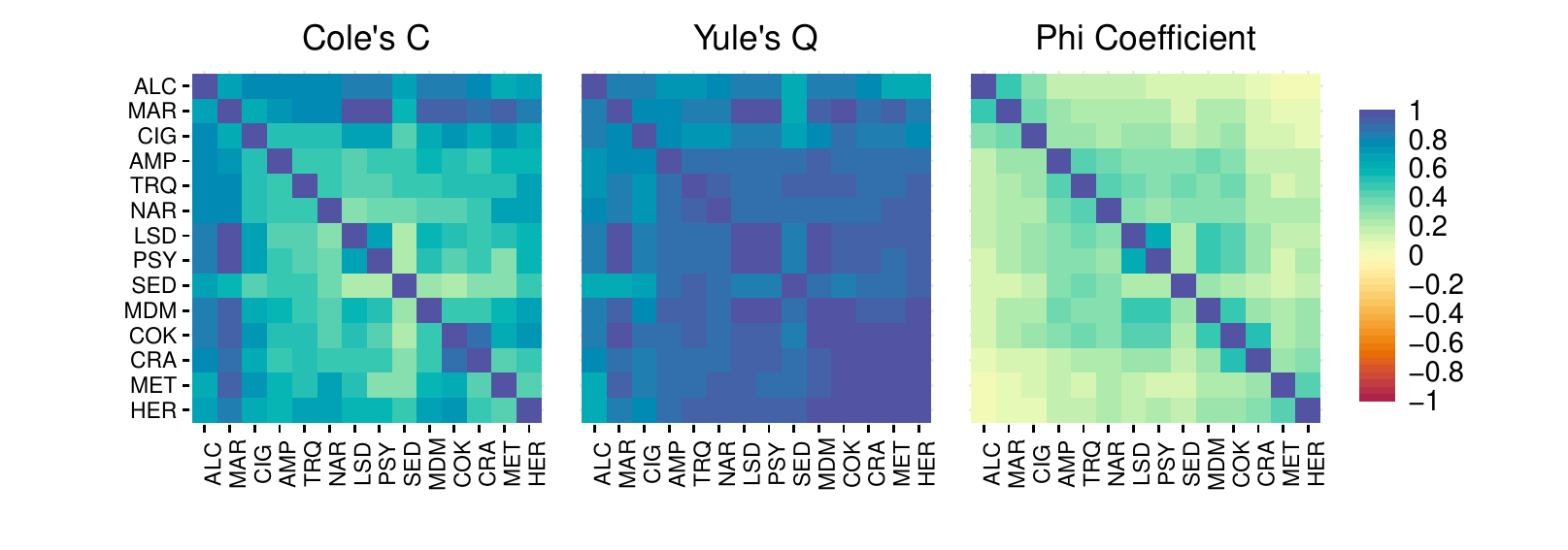}
	\caption{This figure depicts the correlations between the events of having consumed certain drugs. The pair of drugs to which each colored square refers is indicated by the description in the rows and the columns (for the list of abbreviations, see Table \ref{tab:marginals}).}
	\label{fig:drugs_correlations}
\end{figure}

What is more, the phi coefficient might give rise to very different policy advice with respect to prevention of hard drug usage among adolescents. While it especially accentuates the dependence between drugs of similar ``hardness'' such as alcohol and marijuana or meth and heroin, the proper correlations also yield sizeable results for dependence between drugs of very different ``hardness'' such as alcohol and meth. The reason for this behaviour are the very different marginal distributions of those drug use variables, which lead to a much narrower interval of values which $\phi$ can attain (see again Figures \ref{fig:phibounds} and \ref{fig:colecomparison07}). However, thus concluding that the gateway drug effect (alternatively, escalation hypothesis, progression hypothesis, or stepping-stone theory), describing the process of starting with soft drugs like alcohol and subsequently progressing to harder drugs, is negligible, would result in misleading policy advice. Instead, this effect has been extensively researched since \citet{kandel1975stages} and is nowadays widely accepted in the literature \citep{kandel2015gateway}.\footnote{The question to which degree the observed sequential pattern of drug use is causal is still up for debate. Though, the existence of a sizeable association is consensus, but not mirrored by the phi coefficient.}

While the two proper correlation matrices all indicate a very strong dependence overall, there are nevertheless clearly visible differences between them. Those differences can be very well understood by looking again at Figure \ref{fig:colecomparison07}. For example, Yule's $\myQ$ takes considerably higher values than Cole's $\myC$ for all drugs except for the three softest. This is due to the marginal event probabilities being rather small and similar, corresponding to a situation as in the lower left corner of the corresponding panels of Figure \ref{fig:colecomparison07}. Thus, this data example also illustrates that even though the proper dependence measures are usually close to each other (see again Figure \ref{fig:colecomparison07}), there may be substantial differences between them.

To get an idea of sampling uncertainty, we first focus on the dependence between marijuana and all other drugs (the second row or column of the correlation matrices above) and plot $\myCn$, $\widehat{\phi}_n$ and $\myQn$ and 90\% confidence intervals based on the Fisher transformation in Figure \ref{fig:case_study_CIs_mari}. Even though the consumption events for the harder drugs are rather rare events (that is, the marginal and even more so the joint consumption probabilities can be quite small), the confidence intervals are quite narrow due to the large sample size. 
As expected, the intervals for $\myC$ are slightly wider than the intervals for $\myQ$ due to their more complicated and possibly conservative construction method.
Notice that the narrow confidence intervals for $\phi$ are not an advantage of this measure but are caused by the suboptimal normalization of $\phi$ that forces $\phi$ to be much smaller in absolute value due to its non-attainability. 
%\todo{I decided to not comment on the Fisher transform as this is only plotted in Figure 14.}
% This deficiency is, as further explained in Supplement \ref{sec:Simulations}, also the reason that the Fisher transformation is relatively ineffective for $\phi$.

In general, the intervals tend to get wider the smaller the consumption probability for the drug gets as this increases the asymptotic variance. Another important factor is the effective sample size for a pair of drugs, which amounts to the joint number of respondents for those drugs as listed in Table \ref{tab:samplesizes_marijuana} in the appendix. It is lowest for methamphetamines, which explains the wider intervals in this case. In Figure \ref{fig:case_study_CIs_meth} in the appendix we also show the analogous graph for the combination of methamphetamines with all other drugs, where the sampling uncertainty is higher due to the consumption of this drug being a rare event and the smaller sample size when pairing this drug with others (see Table \ref{tab:samplesizes_meth}).

\begin{figure}[tb]
	\centering
	\includegraphics[width=1\linewidth]{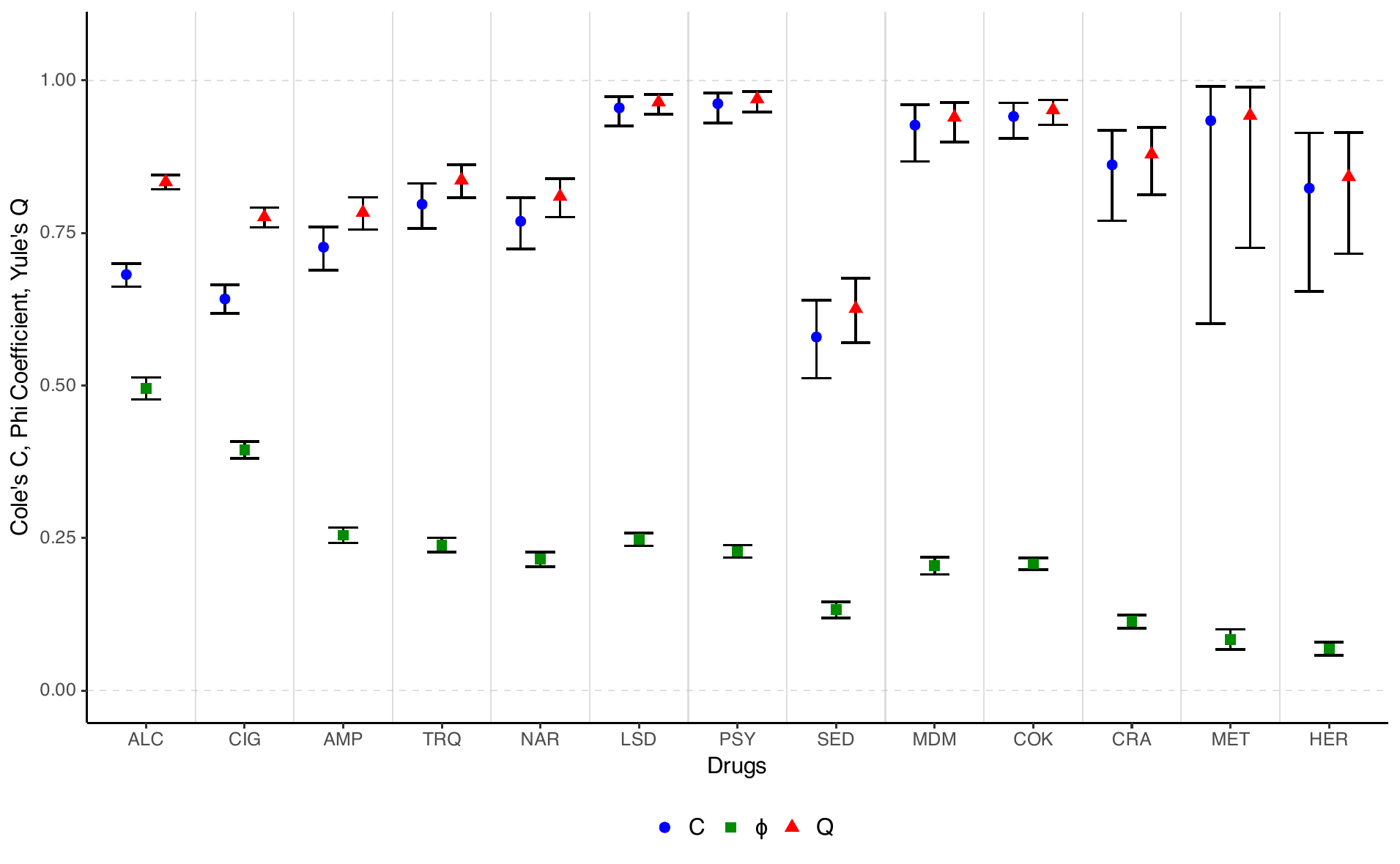}
	\caption{Point estimates $\myCn$, $\widehat{\phi}_n$ and $\myQn$ for marijuana compared with all other drugs and corresponding 90\% confidence intervals based on the Fisher transformation.}
	\label{fig:case_study_CIs_mari}
\end{figure}

\section{Conclusion} \label{sec:conclusion}

This paper deals with dependence measures for events or binary random variables. We introduce dependence concepts for events and the notion of a proper dependence measure. We then discuss theoretical properties of, statistical inference for, and relations between different measures. Important recommendations for statistical practice arise: The most widely-used measures, the phi coefficient and the closely related contingency coefficients, should be avoided for this purpose. They lack the crucial property of attainability %(of the values $-1$ and 1 that are usually associated with perfect (linear) dependence) 
and can hence severely understate strength of dependence, leading to possibly misleading conclusions. Instead, proper dependence measures should be used. The largely unknown Cole coefficient $\myC$ is a very natural and nicely interpretable representative of this class. Even though its statistical inference is challenging due to the case distinction in its definition, we construct tests and confidence intervals based on the asymptotic distribution that we derive. Yule's $\myQ$ is an attractive alternative, avoiding the case distinction in the normalization, and making possible the use of classical inferential procedures. We further discuss a generalized Yule's $\myQ$, the also closely related and popular odds ratio and tetrachoric correlation and establish their propriety.

Generalizing the notion of a proper dependence measure to arbitrary random variables is an important task for future research. While in the case of continuous random variables several sets of desirable properties have been put forward since \citet{Renyi1959} as discussed in the introduction, firstly the extension to the discrete case and secondly the agreement on a set of axioms need to be tackled. We see our work in the binary case as a foundation for this endeavor. A subsequent and equally important task is the search for or the development of measures fulfilling those properties. While in the continuous case rank correlations such as Kendall's $\tau$ and Spearman's $\rho$ fulfill many desirable properties, they are not attainable in the discrete case as well. %and thus, at least all classical dependence measures are improper.
%\todo{I don't understand the "at least" in the last sentence.}
%\lukas{I think Marc means that there may be dependence measures that are proper, but those aren't classical}

%\marc{We could also mention the development of inference for dependence measures involving a case distinction here, where our treatment of the Cole coefficient would serve as an important fundament.} 

\singlespacing
\bibliographystyle{apalike}
\bibliography{library_correlations}

\begin{thebibliography}{}

\bibitem[Balakrishnan and Lai, 2009]{Balakrishnan2009}
Balakrishnan, N. and Lai, C.-D. (2009).
\newblock {\em Continuous Bivariate Distributions}.
\newblock Springer, Dordrecht, second edition.

\bibitem[Berger, 1982]{Berger1982}
Berger, R.~L. (1982).
\newblock Multiparameter {{Hypothesis Testing}} and {{Acceptance Sampling}}.
\newblock {\em Technometrics}, 24(4):295--300.

\bibitem[Bishop et~al., 2007]{bishop2007discrete}
Bishop, Y.~M., Fienberg, S.~E., and Holland, P.~W. (2007).
\newblock {\em Discrete {{Multivariate Analysis Theory}} and {{Practice}}}.
\newblock Springer, New York, NY.

\bibitem[Boas, 1909]{boas1909}
Boas, F. (1909).
\newblock Determination of the {{Coefficient}} of {{Correlation}}.
\newblock {\em Science}, 29(751):823--824.

\bibitem[Bonett and Price, 2007]{bonett2007}
Bonett, D.~G. and Price, R.~M. (2007).
\newblock Statistical inference for generalized {Y}ule coefficients in
  2$\times$2 contingency tables.
\newblock {\em Sociological Methods \& Research}, 35(3):429--446.

\bibitem[Chatterjee, 2021]{chatterjee2021cor}
Chatterjee, S. (2021).
\newblock A new coefficient of correlation.
\newblock {\em Journal of the American Statistical Association},
  116(536):2009--2022.

\bibitem[Cole, 1949]{cole1949}
Cole, L.~C. (1949).
\newblock The {{Measurement}} of {{Interspecific Associaton}}.
\newblock {\em Ecology}, 30(4):411--424.

\bibitem[Cornfield, 1951]{cornfield_method_1951}
Cornfield, J. (1951).
\newblock A method of estimating comparative rates from clinical data.
  {A}pplications to cancer of the lung, breast, and cervix.
\newblock {\em JNCI: Journal of the National Cancer Institute},
  11(6):1269--1275.

\bibitem[Cramér, 1945]{Cramer1946}
Cramér, H. (1945).
\newblock {\em Mathematical Methods of Statistics}.
\newblock Almqvist \& Wiksells, Uppsala, Sweden.

\bibitem[Digby, 1983]{digby1983approx}
Digby, P. G.~N. (1983).
\newblock Approximating the {{Tetrachoric Correlation Coefficient}}.
\newblock {\em Biometrics}, 39(3):753.

\bibitem[Edwards, 1963]{edwards1963}
Edwards, A.~W. (1963).
\newblock The measure of association in a 2$\times$2 table.
\newblock {\em Journal of the Royal Statistical Society Series A: Statistics in
  Society}, 126(1):109--114.

\bibitem[Ekstr{\"o}m, 2011]{Ekstrom2011}
Ekstr{\"o}m, J. (2011).
\newblock The {{Phi-coefficient}}, the {{Tetrachoric Correlation Coefficient}},
  and the {{Pearson-Yule Debate}}.
\newblock {\em UCLA: Department of Statistics Papers}.

\bibitem[Embrechts et~al., 2002]{Embrechts2002}
Embrechts, P., McNeil, A., and Straumann, D. (2002).
\newblock Correlation and dependence in risk management: Properties and
  pitfalls.
\newblock In Dempster, M., editor, {\em Risk Management: Value at Risk and
  Beyond}, pages 176--223. Cambridge University Press, Cambridge, MA.

\bibitem[Falk and {Bar-Hillel}, 1983]{falk1983}
Falk, R. and {Bar-Hillel}, M. (1983).
\newblock Probabilistic {{Dependence Between Events}}.
\newblock {\em The Two-Year College Mathematics Journal}, 14(3):240--247.

\bibitem[Fisher, 1915]{fisher1915}
Fisher, R.~A. (1915).
\newblock Frequency {{Distribution}} of the {{Values}} of the {{Correlation
  Coefficient}} in {{Samples}} from an {{Indefinitely Large Population}}.
\newblock {\em Biometrika}, 10(4):507--521.

\bibitem[Fissler and Pohle, 2023]{fissler2023generalised}
Fissler, T. and Pohle, M.-O. (2023).
\newblock Generalised {{Covariances}} and {{Correlations}}.

\bibitem[Fr{\'e}chet, 1951]{frechet1951tableaux}
Fr{\'e}chet, M. (1951).
\newblock Sur les tableaux de corr{\'e}lation dont les marges sont donn{\'e}es.
\newblock {\em Annales de l'Universit{\'e} Lyon A (3)}, 14:53--77.

\bibitem[Geenens, 2020]{geenens2020}
Geenens, G. (2020).
\newblock Copula modeling for discrete random vectors.
\newblock {\em Dependence Modeling}, 8(1):417--440.

\bibitem[Geenens, 2023]{geenens2023}
Geenens, G. (2023).
\newblock Towards a universal representation of statistical dependence.
\newblock {\em arXiv preprint arXiv:2302.08151}.

\bibitem[Genest and Ne{\v s}lehov{\'a}, 2007]{genest2007primer}
Genest, C. and Ne{\v s}lehov{\'a}, J. (2007).
\newblock A {{Primer}} on {{Copulas}} for {{Count Data}}.
\newblock {\em ASTIN Bulletin}, 37(2):475--515.

\bibitem[Goodman and Kruskal, 1954]{Goodman1954}
Goodman, L.~A. and Kruskal, W.~H. (1954).
\newblock Measures of {{Association}} for {{Cross Classifications}}.
\newblock {\em Journal of the American Statistical Association},
  49(268):732--764.

\bibitem[Greenhouse, 1982]{greenhouse1982cornfield}
Greenhouse, S.~W. (1982).
\newblock Jerome {Cornfield's} contributions to epidemiology.
\newblock {\em Biometrics}, 38:33--45.

\bibitem[Hoeffding, 1940]{hoeffding1940}
Hoeffding, W. (1940).
\newblock {Masstabinvariante Korrelationstheorie}.
\newblock {\em Schriften des Mathematischen Instituts und Instituts fur
  Angewandte Mathematik der Universitaet Berlin}, 5:181--233.

\bibitem[Kandel, 1975]{kandel1975stages}
Kandel, D. (1975).
\newblock Stages in {{Adolescent Involvement}} in {{Drug Use}}.
\newblock {\em Science}, 190(4217):912--914.

\bibitem[Kandel and Kandel, 2015]{kandel2015gateway}
Kandel, D. and Kandel, E. (2015).
\newblock The {{Gateway Hypothesis}} of substance abuse: Developmental,
  biological and societal perspectives.
\newblock {\em Acta Paediatrica}, 104(2):130--137.

\bibitem[Lehmann, 1966]{lehmann1966}
Lehmann, E.~L. (1966).
\newblock Some {{Concepts}} of {{Dependence}}.
\newblock {\em Annals of Mathematical Statistics}, 37(5):1137--1153.

\bibitem[Liebetrau, 1983]{liebetrau1983measures}
Liebetrau, A.~M. (1983).
\newblock {\em Measures of Association}.
\newblock Sage, Newbury Park, CA.

\bibitem[Mari and Kotz, 2004]{Mari2001}
Mari, D.~D. and Kotz, S. (2004).
\newblock {\em Correlation and Dependence}.
\newblock Imperial College Press, London, UK, reprinted edition.

\bibitem[Matthews, 1975]{matthews1975}
Matthews, B.~W. (1975).
\newblock Comparison of the predicted and observed secondary structure of {T4}
  phage lysozyme.
\newblock {\em Biochimica et Biophysica Acta (BBA)-Protein Structure},
  405(2):442--451.

\bibitem[McLeish, 1974]{McLeish1974}
McLeish, D.~L. (1974).
\newblock Dependent central limit theorems and invariance principles.
\newblock {\em Annals of Probability}, 2(4):620--628.

\bibitem[Miech et~al., 2019]{miech2019}
Miech, R.~A., Johnston, L.~D., Bachman, J.~G., O'Malley, P.~M., and
  Schulenberg, J.~E. (2019).
\newblock Monitoring the future: A continuing study of {American} youth
  (12th-grade survey), 2018.

\bibitem[Mosteller, 1968]{mosteller1968}
Mosteller, F. (1968).
\newblock Association and estimation in contingency tables.
\newblock {\em Journal of the American Statistical Association}, 63(321):1--28.

\bibitem[Ne{\v s}lehov{\'a}, 2007]{neslehova2007rank}
Ne{\v s}lehov{\'a}, J. (2007).
\newblock On rank correlation measures for non-continuous random variables.
\newblock {\em Journal of Multivariate Analysis}, 98(3):544--567.

\bibitem[Newey and West, 1987]{NeweyWest1987}
Newey, W.~K. and West, K.~D. (1987).
\newblock A {{Simple}}, {{Positive Semi-Definite}}, {{Heteroskedasticity}} and
  {{Autocorrelation Consistent Covariance Matrix}}.
\newblock {\em Econometrica}, 55(3):703--708.

\bibitem[Pearson, 1900]{pearson1900}
Pearson, K. (1900).
\newblock I. {{Mathematical}} contributions to the theory of evolution.
  ---{{VII}}. {{On}} the correlation of characters not quantitatively
  measurable.
\newblock {\em Philosophical Transactions of the Royal Society of London.
  Series A, Containing Papers of a Mathematical or Physical Character},
  195(262-273):1--47.

\bibitem[Pearson, 1904]{Pearson1904}
Pearson, K. (1904).
\newblock {Mathematical Contributions to the Theory of Evolution. XIII. On the
  Theory of Contingency and its Relation to Association and Normal
  Correlation}.
\newblock Drapers' Company Research Memoirs. Biometric Series 1. Dulau and Co.,
  London, UK.

\bibitem[Plackett, 1965]{plackett1965}
Plackett, R.~L. (1965).
\newblock A class of bivariate distributions.
\newblock {\em Journal of the American Statistical Association},
  60(310):516--522.

\bibitem[R{\'e}nyi, 1959]{Renyi1959}
R{\'e}nyi, A. (1959).
\newblock On measures of dependence.
\newblock {\em Acta Mathematica Academiae Scientiarum Hungarica},
  10(3):441--451.

\bibitem[Scarsini, 1984]{scarsini1984measures}
Scarsini, M. (1984).
\newblock On measures of concordance.
\newblock {\em Stochastica}, 8(3):201--218.

\bibitem[Schweizer and Wolff, 1981]{Schweizer1981}
Schweizer, B. and Wolff, E.~F. (1981).
\newblock On {{Nonparametric Measures}} of {{Dependence}} for {{Random
  Variables}}.
\newblock {\em Annals of Statistics}, 9(4):879--885.

\bibitem[Shannon, 1948]{shannon1948mathematical}
Shannon, C.~E. (1948).
\newblock A mathematical theory of communication.
\newblock {\em The Bell System Technical Journal}, 27(3):379--423.

\bibitem[Stephenson, 2000]{stephenson2000}
Stephenson, D.~B. (2000).
\newblock Use of the ``{{Odds Ratio}}'' for {{Diagnosing Forecast Skill}}.
\newblock {\em Weather and Forecasting}, 15(2):221--232.

\bibitem[Sz{\'e}kely et~al., 2007]{szekely2007measuring}
Sz{\'e}kely, G.~J., Rizzo, M.~L., and Bakirov, N.~K. (2007).
\newblock Measuring and testing dependence by correlation of distances.
\newblock {\em Annals of Statistics}, 35(6):2769--2794.

\bibitem[Tchen, 1980]{tchen1980inequalities}
Tchen, A.~H. (1980).
\newblock Inequalities for {{Distributions}} with {{Given Marginals}}.
\newblock {\em Annals of Probability}, 8(4):814--827.

\bibitem[Tj{\o}stheim et~al., 2022]{Tjostheim2022}
Tj{\o}stheim, D., Otneim, H., and St{\o}ve, B. (2022).
\newblock Statistical {{Dependence}}: {{Beyond Pearson}}'s {$\rho$}.
\newblock {\em Statistical Science}, 37(1):90--109.

\bibitem[Tschuprow, 1925]{Tschuprow1925}
Tschuprow, A.~A. (1925).
\newblock {\em Grundbegriffe und Grundprobleme der Korrelationstheorie}.
\newblock Teubner, Berlin, Germany.

\bibitem[{van der Vaart}, 1998]{VanderVaart2000}
{van der Vaart}, A.~W. (1998).
\newblock {\em Asymptotic Statistics}.
\newblock Cambridge University Press.

\bibitem[Vandenhende and Lambert, 2003]{VandenhendeLambert2003}
Vandenhende, F. and Lambert, P. (2003).
\newblock Improved rank-based dependence measures for categorical data.
\newblock {\em Statistics \& Probability Letters}, 63(2):157--163.

\bibitem[Warrens, 2008]{warrens2008association}
Warrens, M.~J. (2008).
\newblock On {{Association Coefficients}} for 2\texttimes 2 {{Tables}} and
  {{Properties that Do not Depend}} on the {{Marginal Distributions}}.
\newblock {\em Psychometrika}, 73(4):777--789.

\bibitem[Warrens, 2019]{warrens2019similarity}
Warrens, M.~J. (2019).
\newblock Similarity {{Measures}} for 2 \texttimes{} 2 tables.
\newblock {\em Journal of Intelligent \& Fuzzy Systems}, 36(4):3005--3018.

\bibitem[Wermuth, 2025]{wermuth2025nominal}
Wermuth, J.-L. (2025).
\newblock Proper correlation coefficients for nominal random variables.
\newblock {\em arXiv Preprint: 2505.00785}.

\bibitem[White, 2001]{White2001}
White, H. (2001).
\newblock {\em Asymptotic Theory for Econometricians}.
\newblock Academic Press, San Diego.

\bibitem[Yanagimoto and Okamoto, 1969]{yanagimoto1969partial}
Yanagimoto, T. and Okamoto, M. (1969).
\newblock Partial orderings of permutations and monotonicity of a rank
  correlation statistic.
\newblock {\em Annals of the Institute of Statistical Mathematics},
  21(1):489--506.

\bibitem[Yule, 1900]{yule1900}
Yule, G.~U. (1900).
\newblock {{VII}}. {{On}} the {{Association}} of {{Attributes}} in
  {{Statistics}}: With {{Illustrations}} from the {{Material}} of the
  {{Childhood Society}}, \&c.
\newblock {\em Philosophical Transactions of the Royal Society of London.
  Series A, Containing Papers of a Mathematical or Physical Character},
  194(252-261):257--319.

\bibitem[Yule, 1912]{yule1912}
Yule, G.~U. (1912).
\newblock On the {{Methods}} of {{Measuring Association Between Two
  Attributes}}.
\newblock {\em Journal of the Royal Statistical Society}, 75(6):579--652.

\end{thebibliography}

\newpage

\appendix
\appendixpage

\section{Details on Further Dependence Measures}
\label{sec:FurtherMeasures}

\subsection{Contingency Coefficients}
\label{subsec:contingency_coefficients}

Several contingency coefficients have been defined based on Pearson's Mean Square Contingency coefficient $\MSC(X,Y)$ from \eqref{eq:MSC} (see \cite{liebetrau1983measures} for an overview), including Cram\'ers $\myV$, Tschuprow's $\myT$, and Pearson's contingency coefficient $\myPC$
\begin{align*}
	\myV(X,Y) &= \sqrt{ \frac{\MSC(X,Y)}{\min(k-1,l-1)}}, \qquad 
	\myT(X,Y) = \sqrt{ \frac{\MSC(X,Y)}{\sqrt{(k-1)(l-1)}}}, \qquad \text{ and } \\
	\myPC(X,Y) &= \sqrt{ \frac{\MSC(X,Y)}{ 1 + \MSC(X,Y) }}.
\end{align*}
In the binary case, $\MSC(A,B):= \MSC( \mathds{1}_A, \mathds{1}_B )$ becomes
\begin{align} \label{eq:mean_square_contingency}
	\begin{split}
		\MSC(A,B) = \; &\frac{ \left( \P(A \cap B) - \P(A) \P(B) \right)^2}{\P(A) \P(B)} + \frac{ \left( \P(A \cap \overline{B}) - \P(A) \P(\overline{B}) \right)^2}{\P(A) \P(\overline{B})}\\ 
		&\quad+ \frac{ \left( \P(\overline{A} \cap B) - \P(\overline{A}) \P(B) \right)^2}{\P(\overline{A}) \P(B)}+ \frac{ \left( \P(\overline{A} \cap \overline{B}) - \P(\overline{A}) \P(\overline{B}) \right)^2}{\P(\overline{A}) \P(\overline{B})}.
	\end{split}
\end{align}
Further, in the binary case the mean square contingency coefficient just equals the squared phi coefficient. The result seems to be known in the literature, but a proof is hard to find, so we provide one.
\begin{lemma} 
	\label{lemma:MSC_phi}
	It holds that $\MSC(A,B)=\phi^2(A,B)$.
\end{lemma}
We write $\myV(A,B):=\myV(\mathds{1}_A, \mathds{1}_B )$ and analogously $\myT(A,B)$ and $\myPC(A,B)$. Here, Cram\'ers $\myV$ and Tschuprow's $\myT$ coincide, $\myV(A,B) = \myT(A,B)=\sqrt{\MSC(A,B)}$. 
Thus, the contingency coefficients are just functions of the phi coefficient.

\begin{corollary}
	It holds that
	$$\myV(A,B) = | \phi(A,B) |$$
	and  
	$$ \myPC(A,B) = \sqrt{ \frac{ \phi(A,B)^2 }{ 1 + \phi(A,B)^2 } }.$$
\end{corollary}

Thus, in the binary case, Cram\'ers $\myV$ and Tschuprow's $\myT$ just report the same information on strength of dependence as the phi coefficient, but remove the sign, that is, the information on direction of dependence. They of course inherit the attainability problems of $\phi(A,B)$ and  are thus no proper dependence measures (not even in terms of a modified set of axioms adopted to measures of strength of dependence mapping to $[0,1]$). 
Pearson's contingency coefficient $\myPC(A,B)$ exhibits even more shortcomings, for example its maximum value is $\frac{1}{\sqrt{2}}$ instead of 1.

\subsection{Odds Ratio} 
\label{subsec:odds_ratio}

The origin of the odds ratio seems to be unknown, but \citet{cornfield_method_1951} has played an important role in popularizing the measure in epidemiology \citep{greenhouse1982cornfield}. 
It is very popular in medicine. % \citep{Bland1468}.
Indeed, it surmounts all the classical dependence measures discussed here in terms of Google $(14{,}200{,}000)$ and Google Scholar $(1{,}590{,}000)$ citations (7th November 2025, compare Table \ref{tab:google_hits}), presumably due to the sheer size of the medical literature.
The quotient between a probability (here: $\P(A|B)$ or $\P(B|A)$) and its counterprobability is called the odds. Dividing the odds for $A$ conditional on $B$ and the odds for $A$ conditional on $\overline{B}$  by each other yields the odds ratio,
\begin{equation} \label{eq:odds_ratio}
	\OR(A,B):=\frac{\frac{\P(A|B)}{1-\P(A|B)}}{\frac{\P(A|\overline{B})}{1-\P(A|\overline{B})}}=\frac{\frac{\P(B|A)}{1-\P(B|A)}}{\frac{\P(B|\overline{A})}{1-\P(B|\overline{A})}}=\frac{\P(A\cap B)\P(\overline{A}\cap \overline{B})}{\P(\overline{A}\cap B)\P(A\cap \overline{B})}. 
\end{equation}
Thus, it measures dependence between $A$ and $B$ by considering how the odds of $A$ increase given $B$. At first, it seems that the odds ratio does not fit into the framework of this paper, where we are interested in mutual dependence between events, that is, $A$ and $B$ should play an interchangeable role, which is guaranteed by the symmetry of a measure. Instead, the odds ratio considers $A$ conditional on $B$. However, the odds ratio is indeed symmetric as it can be rewritten as the second and third equalities in \eqref{eq:odds_ratio}. In our Definition \ref{def:odds_ratio}, we also take into account that the terms from \eqref{eq:odds_ratio} would be undefined in case of perfect positive dependence.

The odds ratio can be regarded as proper as it fulfills a modified set of axioms, that is, suitable adaptations of the properties from \ref{def:proper_measure} to its scale and multiplicative nature. 
\begin{proposition} \label{prop:properties_OR}
	\begin{enumerate}[(A)] 
		\item \emph{Normalization}: $0 \le \OR(A,B) \le \infty$.
		\item \emph{Independence}: $\OR(A,B)=1$ if and only if $A$ and $B$ are independent.
		\item \emph{Attainability}: $\OR(A,B)=\infty \ (0)$ if and only if $A$ and $B$ are perfectly positively (negatively) dependent.
		\item \emph{Monotonicity}: Let $\P(A)=\P(A^*)$ and $\P(B)=\P(B^*)$. $\OR(A,B) \ge (\le) \ \OR(A^*,B^*)$ if and only if $A$ and $B$ are stronger positively (negatively) dependent than $A^*$ and $B^*$.
		\item \emph{Symmetry}: $\OR(A,B)=\OR(
		B,A)$ and $\OR(A,\overline{B})=\frac{1}{\OR(A,B)}$.
	\end{enumerate}
\end{proposition}

\subsection{Invariance to the Marginals of $\myQ$, $\myQ_g$ and the Odds Ratio} 
\label{subsec:invariance_marginals}

$\myQ$, $\OR$ and $\myQ_g$ have an interesting property, which is discussed already in \cite{yule1912} and later in \citet{edwards1963}, \citet{plackett1965} and \citet{mosteller1968} for empirical contingency tables containing absolute frequencies and plays a central role in an approach to copula modeling and measuring dependence suggested in \citet{geenens2020} and \citet{geenens2023}: When the columns and rows of a $2 \times 2$ contingency table (see Table \ref{tab:contingency_table} again) are multiplied by constants $c_1, c_2$ and $r_1,r_2$ and then the numbers are normalized by an appropriate normalising constant $\nu$ so that the joint probabilities sum up to 1 again, those measures do not change. Considering the odds ratio, it is easy to see that those constants cancel out:
$$\frac{c_1 r_1 \nu \P(A\cap B) c_2 r_2 \nu \P(\overline{A}\cap \overline{B})}{c_1 r_2 \nu \P(\overline{A}\cap B) c_2 r_1 \nu \P(A\cap \overline{B})} = \OR(A,B).$$
By Lemma \ref{lemma:Q_and_OR}, this directly carries over to Yule's $\myQ$. This property is seen as a desirable property by the papers cited above and often described as invariance of those measures to the marginal event probabilities because the measures do not change when the marginal event probabilities change and the respective joint and conditional probabilities change accordingly. On the one hand, this is indeed a very nice property as it allows to really disentangle the marginal distributions of the events and their dependence. On the other hand, this property is incompatible with the property that after distinguishing positive and negative dependence the dependence measure is proportional to $\Cov(A,B)$ and a linear function of $\P(A \cap B)$. This is a property that only Cole's $\myC$ possesses, whereas the normalization of $\myQ$ changes with the joint probabilities. We do not take a strong stance regarding those properties. Both measures fulfill the more fundamental properties leading to propriety and both are very natural dependence measures.

\subsection{Tetrachoric Correlation} \label{subsec:tetrachoric}

Tetrachoric correlation is constructed via the assumption that the $2 \times 2$ contingency table arises by dichotomization of an underlying bivariate normal distribution. The joint probability $\P(A \cap B)$ has to match the area under the bivariate density below quantiles of the two marginals determined by the event probabilities $\P(A)$ and $\P(B)$. As the means and variances play no role here, the correlation coefficient of this underlying bivariate normal is uniquely determined by the joint probability. We adopt the definition from \cite{Ekstrom2011}.

\begin{definition}[Tetrachoric Correlation]
	For interior cases $(\max(\P(A)+\P(B)-1,0)<\P(A\cap B)<\min(\P(A), \P(B)))$, tetrachoric correlation $\myTC(A,B)$ is implicitly defined via the equation
	$$\P(A\cap B)=\int_{-\infty}^{\Phi^{-1}(P(A))} \int_{-\infty}^{\Phi^{-1} (P(B))} f_{\myTC(A,B)} (x, y) dxdy,$$
	%	$$\P(A\cap B)=\int_{\Phi^{-1}(1-\P(A))}^{\infty}\int_{\Phi^{-1}(1-\P(B))}^{\infty}\phi_{T(A,B)} (x, y) dxdy,$$
	where $f_\rho$ denotes the density of a bivariate normal distribution with zero means, unit variances and correlation coefficient $\rho$.
	Otherwise, we have $\myTC(A,B)=-1$ for $\P(A\cap B)=\max(\P(A)+\P(B)-1,0)$ and $\myTC(A,B)=1$ for $\P(A, B)=\min(\P(A), \P(B))$.
\end{definition}

Even though tetrachoric correlation does not seem to be a very natural measure (except for cases, where the binary variables $\mathds{1}_A$ and $\mathds{1}_B$ are really or at least approximately generated by an underlying normal distribution), it is proper. This may at first be surprising, but is essentially due to the fact that under normality the Pearson correlation coefficient has desirable properties, in particular the attainability and independence property.

\begin{proposition} \label{prop:tetrachoric_propriety}
	$\myTC(A,B)$ is proper.
\end{proposition}

\subsection{Further Measures}\label{subsec:furthermeasures}
In this subsection, we collect further measures that have originally not been designed for the binary case. Still, they can be used and their popularity in the contemporaneous statistical literature justifies a separate treatment.

The first of such measures is distance correlation \citep{szekely2007measuring}. Its numerator, the distance covariance, is defined as \begin{align*}
    \mathcal{V}(X,Y):=\frac{1}{c_pc_q}\int_{\mathbb{R}^{p+q}}\frac{|f_{X,Y}(t,s)-f_X(t)f_Y(s)|^2}{|t|_p^{1+p}|s|_q^{1+q}}dtds,
\end{align*}
where $f_{X,Y}$ and $f_X,f_Y$ are the joint and marginal characteristic functions, respectively. Additionally, $|\cdot|_p$ denotes the euclidean norm in $\mathbb{R}^p$ and we define $c_p:=\pi^{(1+p)/2}/\Gamma((1+p)/2)$ with $\Gamma(\cdot)$ denoting the complete gamma function.

Thus, resembling Pearson correlation, distance correlation is the nonnegative number $\mathcal{R}(X,Y)$ defined by 
\begin{align*}
    \mathcal{R}^2(X,Y):=\begin{cases}
        \frac{\mathcal{V}^2(X,Y)}{\sqrt{\mathcal{V}^2(X,X)\mathcal{V}^2(Y,Y)}}, &\mathcal{V}^2(X,X)\mathcal{V}^2(Y,Y)>0\\
        0,&\mathcal{V}^2(X,X)\mathcal{V}^2(Y,Y)=0
    \end{cases}
\end{align*}
Simple calculations show that in the binary case, distance correlation equals the absolute value of the phi coefficient, $\mathcal{R}(\mathds{1}_A,\mathds{1}_B)=|\phi(\mathds{1}_A,\mathds{1}_B)|$. This equality already holds at the empirical level (see Definition 4 in \cite{szekely2007measuring} for an estimator). Thus, as Cramér's V and Tschuprow's T, distance correlation does not indicate the direction of dependence, but inherits all the other shortcomings of the phi coefficient. As a consequence, it is an improper measure.

A second measure that is very popular and directed at functional relationships has been popularized by \cite{chatterjee2021cor}. It is defined as \begin{align*}
    \xi(X,Y):=\frac{\int \Var(\E[\mathds{1}_{\{Y\ge t\}}|X])d\mu(t)}{\int \Var(\mathds{1}_{\{Y\ge t\}})d\mu(t)},
\end{align*}
where $\mu$ is the law of $Y$, and can be symmetrized by considering $\max\{\xi(X,Y),\xi(Y,X)\}$. However, its estimator relies on randomized breaking of ties and is therefore stochastic even for a given sample. Since the binary case is the case in which the most ties occur, this can lead to a lot of variability in the estimator. In addition to that, its population version collapses to the square of the phi coefficient, thus sharing its non-attainability: $\xi(X,Y)=\phi^2(X,Y)$.

Lastly, the mutual information \citep{shannon1948mathematical} can also be used as a dependence measure if a suitable normalization is applied. The most popular one has been named the uncertainty coefficient and reads as
	\begin{align*}
		U(X,Y):=2\left(\frac{H(X)+H(Y)-H(X,Y)}{H(X)+H(Y)}\right),
	\end{align*}
	where \begin{align*}
		H(X):=-\sum_{i=1}^{k}p_{i\cdot}\ln(p_{i\cdot}), \ H(Y):=-\sum_{j=1}^{l}p_{\cdot j}\ln(p_{\cdot j}), \ \  H(X,Y):=-\sum_{i=1}^{k}\sum_{j=1}^{l}p_{ij}\ln(p_{ij})
	\end{align*}
and $p_{i\cdot}, p_{\cdot j}, p_{ij}$ denote the marginal and joint probabilities of the respective values that the (discrete) random variables $X$ and $Y$ take (see the setup around equation (\ref{eq:MSC})). As shown in \cite{wermuth2025nominal}, it attains 1 if and only if the contingency table only has one non-zero entry in each row and each column and thus suffers from the same non-attainability problem as $\phi$.

\section{Inference Using the Fisher Transformation}  
\label{subsec:Fisher_transformations}

As already mentioned in Section \ref{sec:FisherTransformations}, the sampling distributions of $\myQn$, $\myCn$ and $\widehat{\phi}_n$ are heavily influenced by their theoretical Fr\'{e}chet--Hoeffding bounds when the true measures are close to these bounds.
This leads to skewed empirical distributions for which the asymptotic normal approximations (and similarly for $\myCn$) are bad.
This can strongly deteriorate the finite sample performance of associated tests and confidence intervals as can be seen from our simulation results in Appendix \ref{sec:Simulations}.

In the case of Pearson correlation, where the analogous problem occurs, the classical tool to combat this problem is Fisher's transformation \citep{fisher1915},
\begin{align}
	\label{eq:Fisher_transformationAppendix}
	Z(\delta) := \mathrm{arctanh}(\delta) = \frac 1 2 \log \left( \frac{1+\delta}{1-\delta} \right)
\end{align}
for a generic dependence measure $\delta \in \{\myQ, \myC, \phi, \myQn, \myCn, \widehat{\phi}_n \}$.
This transformation maps the bounded dependence measures to an unbounded scale while leaving values in the center (approximately the interval $[-0.5, 0.5]$) almost unchanged.
On this unbounded scale, asymptotic normality is usually a much better finite sample approximation, especially when the dependence measure is close to the boundary. 
We therefore apply the Fisher transformation to construct more reliable tests and confidence intervals for $\myQ$ (as brought up in \cite{bonett2007}), $\myC$ and $\phi$. 
Even though the limiting distribution is not Gaussian for $\myC$, the same problem applies here (unless $\myC=0$).

Starting with the limiting distributions of the Fisher transformation as given below in Propositions \ref{prop:Asymptotics-Z_Q}--\ref{prop:Asymptotics-Z_phi}, we obtain non-normal limiting distributions for $\myQn$, $\myCn$ and $\widehat{\phi}_n$ by applying the inverse transformation $Z^{-1}(z) = \mathrm{tanh}(z) = (e^{2z}-1)/(e^{2z}+1)$ of \eqref{eq:Fisher_transformationAppendix}.
These usually lead to much better approximations of the sampling distributions close to the boundaries as illustrated in Figure \ref{fig:qsimulation}.
For the true $\myQ = 0.95$ with $p=q=0.5$, we plot the standard (red) and Fisher transformed (blue) limiting distributions, their implied $90\%$ confidence intervals in dashed lines, and a histogram depicting $100{,}000$ estimates $\myQn$ for a sample size of $n=100$ from simulated data.
We see that the finite sample histogram is much better approximated by the blue Fisher transformed limiting distribution.

\begin{figure}[bt] 
	\centering
	\includegraphics[width=0.7\linewidth]{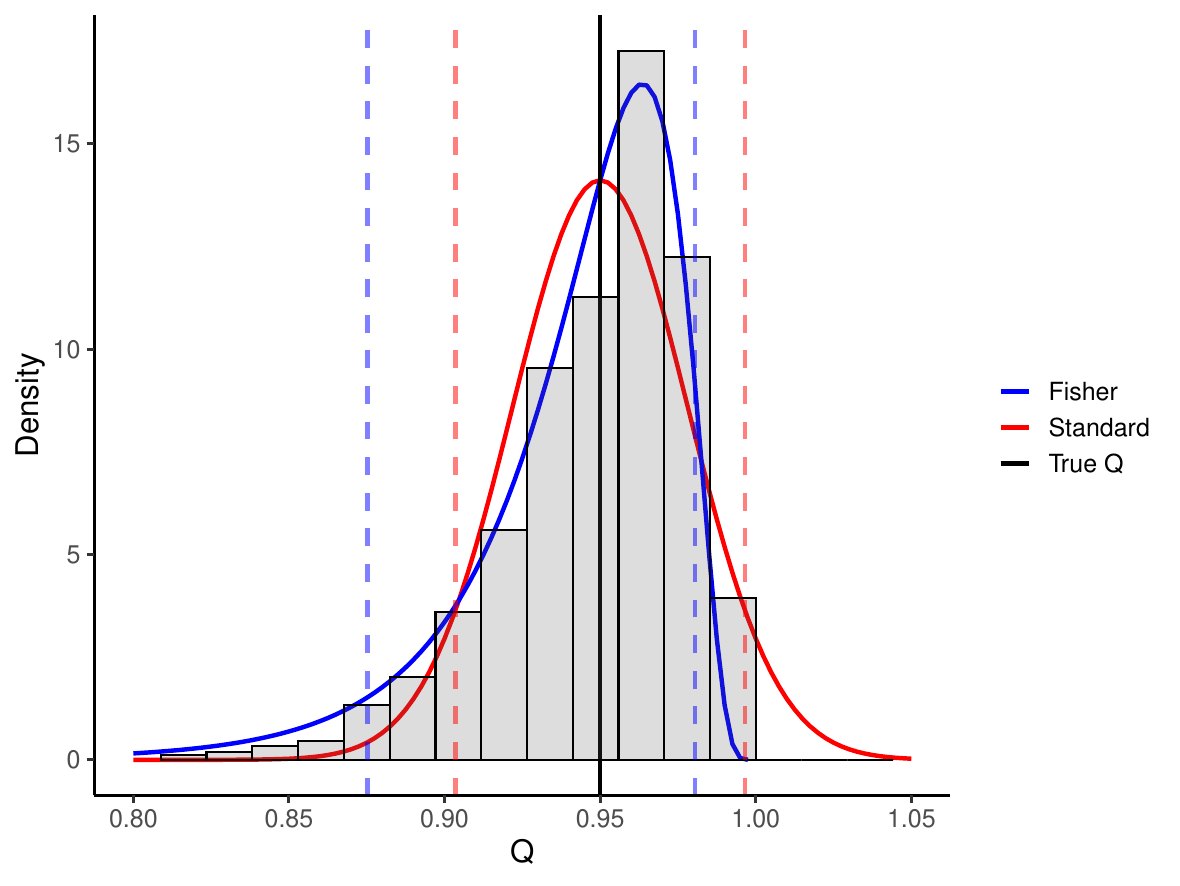}
	\caption{Histogram of $\myQn$ from $100{,}000$ simulated samples of size $n=100$ with a true value of $\myQ=0.95$, marginal probabilities $p=q=0.5$ and the standard asymptotically normal and the (backtransformed) Fisher transformed limiting distributions together with corresponding $90\%$ confidence intervals in dashed lines.}
	\label{fig:qsimulation}
\end{figure}

Consequently, the Fisher transformation leads to more reliable tests and confidence intervals close to the boundaries and almost identical results away from them. 
To execute a test for $\mathcal{H}_0: \delta=d$ one just needs to Fisher transform the hypothesized value and test $\mathcal{H}_0: Z(\delta)=Z(d)$ using the respective asymptotic distribution of the Fisher transformation (again using Bonferroni corrections as described above in the case of $\myC$). 

We construct confidence intervals as follows. For $\myQ$ and $\phi$, we obtain a confidence interval for the Fisher transform based on its asymptotic normality and then apply the inverse Fisher transformation to these bounds.
For $\myC$, we modify the algorithm described in Section \ref{subsec:CIs} by using (inverted) tests based on the asymptotic distribution of the Fisher transforms from Proposition \ref{prop:Asymptotics-Z_C} instead of Proposition \ref{prop:Asymptotics-C}. 
We illustrate the improved performance of the confidence intervals based on the Fisher transform compared to the ones based on the usual limiting distributions in simulations in Appendix \ref{sec:Simulations}.
The importance of the Fisher transformation is also illustrated when comparing the confidence intervals based on it to standard confidence intervals in the introductory example, where they differ substantially as the sampling distribution is close to the upper bound 1, see Table \ref{tab:smallpox_CIs} in the Appendix, and in the case study in Section \ref{sec:application}, where they differ in those cases and are very close to each other otherwise, see Figures \ref{fig:case_study_CIs_mari_comparison} and \ref{fig:case_study_CIs_meth_comparison} in the Appendix.

%\todo{MARC: Propositione/Propositions am Ende?} 
We continue by formally providing the limiting distributions of the Fisher transformations of $\myQ$, $\myC$ and $\phi$ introduced in Section \ref{sec:FisherTransformations}, which can be obtained from the asymptotic distributions from Subsection \ref{subsec:asymptotic_distributions} together with the delta method.
For Yule's $\myQ$, we instead derive the limit distribution directly as in this case the transformation is a simpler expression than Yule's $\myQ$ itself, namely $Z(\myQ)  = \frac 1 2 \log ( \OR )$ with the plug-in estimator 
$$Z(\myQn) = \frac{1}{2} \log \big( \widehat{\OR}_n \big)= \frac{1}{2}  \log \big( \r(1-\p-\q+\r) \big) - \frac{1}{2} \log \big( (\q-\r)(\p-\r) \big).$$ 
\begin{proposition}
	\label{prop:Asymptotics-Z_Q}
	%Given assumptions \ref{ass:BoundedAway}--\ref{ass:iid}, it holds that
	Given Assumption \ref{ass:BoundedAway} and either Assumption \ref{ass:iid} or Assumption \ref{ass:Dependence}, it holds that
	\begin{align*}
		& \big(J_h \Omega J_h^\top \big)^{-1/2} \sqrt{n} \big( Z(\myQn) - Z(\myQ) \big) \stackrel{d}{\to} \mathcal{N} \big( 0, 1\big), \qquad \text{where} \\
		&J_h :=  -\frac{1}{2(p-r)(q-r)} 
		\begin{pmatrix}
			q - r \\
			p - r \\
			2r - p - q
		\end{pmatrix}^\top 
		-\frac{1}{2(1-p-q+r)} 
		\begin{pmatrix}
			1 \\
			1 \\
			-r^{-1} (1-p-q+2r)
		\end{pmatrix}^\top.
	\end{align*}
\end{proposition}

\begin{proof}[Proof of Proposition~\ref{prop:Asymptotics-Z_Q}]
	Note that we can write
	\begin{align*}
		Z(\myQn)  = h \big( \p, \q, \r \big), \quad \text{with} \quad
		h(\pa, \qa, \ra) := \frac 1 2 \left( \log \left( \ra (1-\pa-\qa+\ra) \right) - \log \left( (\qa-\ra)(\pa-\ra) \right) \right).
	\end{align*}
	The function $h$ is continuously differentiable at $(p,q,r)$ with Jacobian matrix
	\begin{align*}
		J_h(\pa, \qa, \ra) = - \frac 1 2 \left( \frac{1}{(\pa-\ra)(\qa-\ra)} 
		\begin{pmatrix}
			\qa - \ra \\
			\pa - \ra \\
			2\ra - \pa - \qa
		\end{pmatrix}^\top 
		+\frac{1}{1-\pa-\qa+\ra} 
		\begin{pmatrix}
			1 \\
			1 \\
			-\ra^{-1} (1-\pa-\qa+2\ra)
		\end{pmatrix}^\top 
		\right).
	\end{align*}
%	Thus, $h$ is differentiable at $(p, q, r)$ as all the partial derivatives are continuous at that point by Assumption \ref{ass:BoundedAway}.
	Furthermore, its first row simplifies to $- \frac{1-\qa}{2 (\pa-\ra)(1-\pa-\qa+\ra)}$, which is nonzero by Assumption \ref{ass:BoundedAway}.
	Hence, by the delta method \cite[Theorem 3.1]{VanderVaart2000} and  Lemma \ref{lemma:CLTjoint} the claim follows.
\end{proof}

\begin{proposition}
	\label{prop:Asymptotics-Z_C}
	With the quantities defined in Proposition \ref{prop:Asymptotics-C}, defining $\gamma_\myC:=1/(1-\myC^2)$ and given that Assumption \ref{ass:BoundedAway} and either Assumption \ref{ass:iid} or Assumption \ref{ass:Dependence} are in place, it holds that:
	\begin{itemize} 
		\item If $\sigma = 0$, then 
		$\big(\Delta^\top \Omega \Delta \big)^{-1/2} \sqrt{n} \big(Z(\myCn) - Z(\myC) \big) \tod \frac{\gamma_\myC Z \, \mathds{1}\{Z > 0\}}{m^+} +  \frac{\gamma_\myC Z \, \mathds{1}\{Z < 0\}}{m^-}$.
		
		\item If $\sigma > 0$, and
		\begin{itemize}
			\item 
			if $p \not= q$, then $\big( \Lambda^{+} J_{h^+} \Omega J_{h^+}^\top  (\Lambda^{+})^\top \big)^{-1/2} \gamma_\myC^{-1} \sqrt{n}   \big( Z(\myCn) - Z(\myC)\big)
			\stackrel{d}{\to} \mathcal{N} \big( 0, 1 \big)$; and
			
			\item 
			if $p = q$, then 
			$\sqrt{n}  \big( Z(\myCn) - Z(\myC) \big) \tod	\gamma_\myC \Lambda^+  
			\begin{pmatrix}
				V_{4} \\ 
				\big[ V_{1} - V_{3} \big]  - \mathds{1} \{ V_{1} > V_{2} \} \big[ V_{1} - V_{2} \big]
			\end{pmatrix}$.
			\end{itemize}

			\item If $\sigma < 0$, and
			\begin{itemize}
				\item 
				if $p \not= 1-q$, then 
				$\big( \Lambda^{-} J_{h^-} \Omega J_{h^-}^\top  (\Lambda^{-})^\top \big)^{-1/2} \gamma_\myC^{-1} \sqrt{n} 
			 	\big( Z(\myCn) - Z(\myC) \big)
				\stackrel{d}{\to} \mathcal{N} \big( 0, 1 \big)$; and
				
				\item
				if $p = 1-q$, then
				$\sqrt{n}  \big( Z(\myCn) - Z(\myC) \big) \tod	\gamma_\myC \Lambda^- 
				\begin{pmatrix}
					V_{4} \\ 	
					V_{3} - \mathds{1} \{ V_{1} > - V_{2} \} \big[  V_{1} + V_{2} \big]
				\end{pmatrix}
				$.
			\end{itemize}
		\end{itemize} 	
	\end{proposition}

\begin{proof}[Proof of Proposition \ref{prop:Asymptotics-Z_C}]
	Noting that the Fisher transformation $Z(\delta)$ from \eqref{eq:Fisher_transformation} is a differentiable function with derivative $Z^\prime(\delta) = 1/(1-\delta^2)$ for any $|\delta| < 1 $, we can apply the delta method to each of the limiting random variables from Proposition \ref{prop:Asymptotics-C} to arrive at the claim.
\end{proof}

\begin{proposition}
	\label{prop:Asymptotics-Z_phi}
	With $J_l$ defined in Proposition \ref{prop:Asymptotics-Phi}, defining $\gamma_\phi :=1/(1-\phi^2)$ and given that Assumption \ref{ass:BoundedAway} and either Assumption \ref{ass:iid} or Assumption \ref{ass:Dependence} are in place, it holds that
	$$\big(J_l \Omega J_l^\top \big)^{-1/2} \gamma_\phi^{-1} \sqrt{n} \big( Z(\widehat{\phi}_n) - Z(\phi) \big) \stackrel{d}{\to} \mathcal{N} \big( 0, 1\big).$$
\end{proposition}

\begin{proof}[Proof of Proposition \ref{prop:Asymptotics-Z_phi}]
	In the same way as in the proof of Proposition \ref{prop:Asymptotics-Z_C}, we apply the delta method to the limit distribution from Proposition \ref{prop:Asymptotics-Phi} to arrive at the claim.
\end{proof}

As discussed for all the other quantities showing up in the asymptotic variances in Section \ref{sec:asymptotics}, $J_h$, $\gamma_{\myC}$ and $\gamma_\phi$ can as well be estimated consistently by their empirical plug-in counterparts $\widehat{J}_h$, $\widehat{\gamma}_{\myC}$ and $\widehat{\gamma}_\phi$.

\section{Time Series Asymptotics}
\label{subsec:time_series_asymptotics}

All our asymptotic results hold under independence as well as under classical time series dependencies, namely by imposing through Assumption \ref{ass:Dependence} that $\mathbf{W}_i := (X_i, Y_i, X_i Y_i)^\top$, $i \in \mathbb{N}$, follows a stationary ergodic adapted mixingale, going back to \cite{McLeish1974}.
Mixingales behave asymptotically like martingale difference processes, analogous to mixing processes that behave asymptotically like independent processes \citep{White2001}.
Generalizations to non-stationarity would be feasible at the cost of a more involved notation as e.g., even the population values $p_n := \frac{1}{n} \sum_{i=1}^n \P(X_i=1)$ would depend on the sample size $n$.
In these cases, a CLT different from Lemma \ref{lemma:CLTjoint} would have to be invoked in the proofs.
While allowing for more general  dependence and non-stationarity conditions in CLTs usually comes at the cost of more restrictive moment conditions, this is not the case here as all moments exist due to the boundedness of our binary random variables.

To incorporate a possible time series dependence of our random variables, we employ a Heteroskedasticity and Autocorrelation Consistent (HAC) estimator for the long-run covariance matrix $\Omega$ \citep{NeweyWest1987}.
For a bandwidth sequence of integers $m_n \to \infty$ with $m_n = o(n^{1/4})$ and a triangular array of weights $\omega_{nj}$ with $\omega_{nj} \le K < \infty$ for all $n \in \mathbb{N}$, $j \le m_n$, and $\omega_{nj} \to 1$ for all $j=1,\dots m_n$ as $n \to \infty$, let
\begin{align}
	\label{eqn:HACEstimator}
	\widehat{\Omega}_n = \frac{1}{n} \sum_{i=1}^n  \moW_{n,i}\moW_{n,i}^\top   +  
	\frac{1}{n} \sum_{j=1}^{m_n} \omega_{n j}  \sum_{i=j+1}^{n} \left( \moW_{n,i} \moW_{n,i-j}^\top  + \moW_{n,i-j} \moW_{n,i}^\top \right),
\end{align} 	
where $\moW_{n,i} := \mW_{i} - (\p, \q, \r)^\top$.
The following proposition shows that this estimator is consistent.

\begin{proposition}
	\label{prop:CovConsistency}
%	It holds that $\widehat{J}_g \toP J_g$, $\widehat{\Lambda}_n^+ \toP \Lambda^+$,  
%	$\widehat{\Lambda}_n^- \toP \Lambda^-$, 
%	$\widehat{\Delta}_n \toP \Delta$, 
%	$\widehat{J}_{n, h^+} \toP J_{n, h^+}$,
%	$\widehat{J}_{n, h^-} \toP J_{n, h^-}$,
%	$\widehat{J}_{n,f} \toP J_{n,f}$, and
	Given that $\moW_{n,i}$ is $\alpha$-mixing of size $-2r/(r-2)$ for some $r>2$, we have $\widehat{\Omega}_n \toP \Omega$.
	%	Furthermore, the weak convergences from Proposition \ref{prop:Asymptotics-C} also hold with estimated quantities! \todo{Is that really correct? How is that framed correctly?}
\end{proposition}

\begin{proof}[Proof of Proposition \ref{prop:CovConsistency}]
	%	As $\big( \p, \q, \r \big) \toP \big( p, q, r \big)$ by a law of large numbers, the continuous mapping theorem directly implies that
	%	$\widehat{J}_g \toP J_g$,
	%	$\widehat{\Lambda}_n^+ \toP \Lambda^+$,  
	%	$\widehat{\Lambda}_n^- \toP \Lambda^-$, 
	%	$\widehat{\Delta}_n \toP \Delta$, 
	%	$\widehat{J}_{n, h^+} \toP J_{n, h^+}$,
	%	$\widehat{J}_{n, h^-} \toP J_{n, h^-}$, and
	%	$\widehat{J}_{n,f} \toP J_{n,f}$.
	%	Notice for this that the indicator function in the definition of the Jacobian function has a single discontinuity point that has zero mass in the joint asymptotic distribution of $(\p, \q)$.
	The claim $\widehat{\Omega}_n \toP \Omega$ follows from \citet[Theorem 6.20]{White2001} as the triangular array $\moW_{n,i}$ is assumed to be $\alpha$-mixing,  the $\moW_{n,i}$ are bounded and hence any moments are bounded, and by construction, $\mathbb{E}[\moW_{n,i}] = \mathbb{E}[\mW_{i}] - \mathbb{E}[\p, \q, \r] = (p, q, r) -  (p, q, r)  = 0$.
\end{proof}

The estimation and consistency of all other matrices involved in the asymptotic distributions is unchanged compared to the iid case.

\section{Simulations}
\label{sec:Simulations} 

\sloppy
We analyze the finite sample performance of the confidence intervals described in Section \ref{subsec:CIs} based on the standard asymptotics in Propositions \ref{prop:Asymptotics-Q}--\ref{prop:Asymptotics-Phi} as well as the Fisher transformations described in Section \ref{sec:FisherTransformations}, Appendix \ref{subsec:Fisher_transformations} and Propositions \ref{prop:Asymptotics-Z_Q}--\ref{prop:Asymptotics-Z_phi}.
We simulate the pairs of binary random variables $(X_i, Y_i)$ independently for different $i=1,\dots,n$ and use $n \in \{100; 500; 2{,}000\}$.
The joint behavior of $(X_i, Y_i)$  is characterized by the probabilities $p = \P(X_i=1)$, $q = \P(Y_i=1)$ and $r =  \P(X_i=1, \, Y_i=1)$.
In order to cover the different settings in Proposition \ref{prop:Asymptotics-C}, we use $(p,q) \in \big\{(0.05, 0.5), \, (0.1,0.9), \, (0.2, 0.2), \, (0.3, 0.7), \, (0.4, 0.8) \big\}$.
%\begin{align*}
%	(p,q) \in \big\{(0.05, 0.85), (0.2, 0.2), (0.3, 0.7), (0.4, 0.8) \big\}
%\end{align*}
 %$p,q$ and settings equal and different to the special cases $p=q$ and $p=1-q$ from Proposition \ref{prop:Asymptotics-C}.
%\textcolor{red}{
For each of these combinations of $p$ and $q$, we choose $39$ equally spaced values for $r$ between $\max(p+q-1,0) + 1/20$ and $\min(p,q) - 1/20$.
In the following, we only report results for the settings where the empirical boundary-conditions $0 < \p,\q < 1$ and $\max(\p + \q - 1, 0) < \r < \min(\p, \q)$ are satisfied in at least $50\%$ of the simulation runs.

%For each of these combinations of $p$ and $q$, we choose 41 equally spaced values for the joint probability $r$ in the interval $\big(r_{\min}, r_{\max} \big) = \big( \max(0, p+q-1), \min(p,q) \big)$.}
% We consider samples of size $n \in \{100, 500, 2000\}$ 
% and construct the confidence intervals as described in Sections \ref{subsec:CIs}  and \ref{sec:FisherTransformations}.

%We calculate the confidence intervals for $\myQ$ using the standard method based on asymptotic normality, that is, taking the confidence interval as the estimate plus/minus a Gaussian quantile times the estimated asymptotic standard deviation.
%For $\myC$, we base our confidence intervals at level $1-\alpha$ on inverted test statistics, that is, we test the null hypotheses $\mathbb{H}_c: C = c$ for all $c \in [-1,1]$, where we can use the appropriate asymptotic distribution from Proposition \ref{prop:Asymptotics-C} under this respective null hypothesis for each $c$. 
%The confidence interval is then spanned by the set $\big\{ c \in [-1,1] \mid \mathbb{H}_c \text{ is not rejected at level } \alpha \big\}$.
%\textcolor{red}{Notice that this (so far) only works if $p \not= q$ and $p \not= 1-q$! Suitable adaptions should be developed for these cases.}

\begin{figure}[p]
	\centering
	\includegraphics[width=\linewidth]{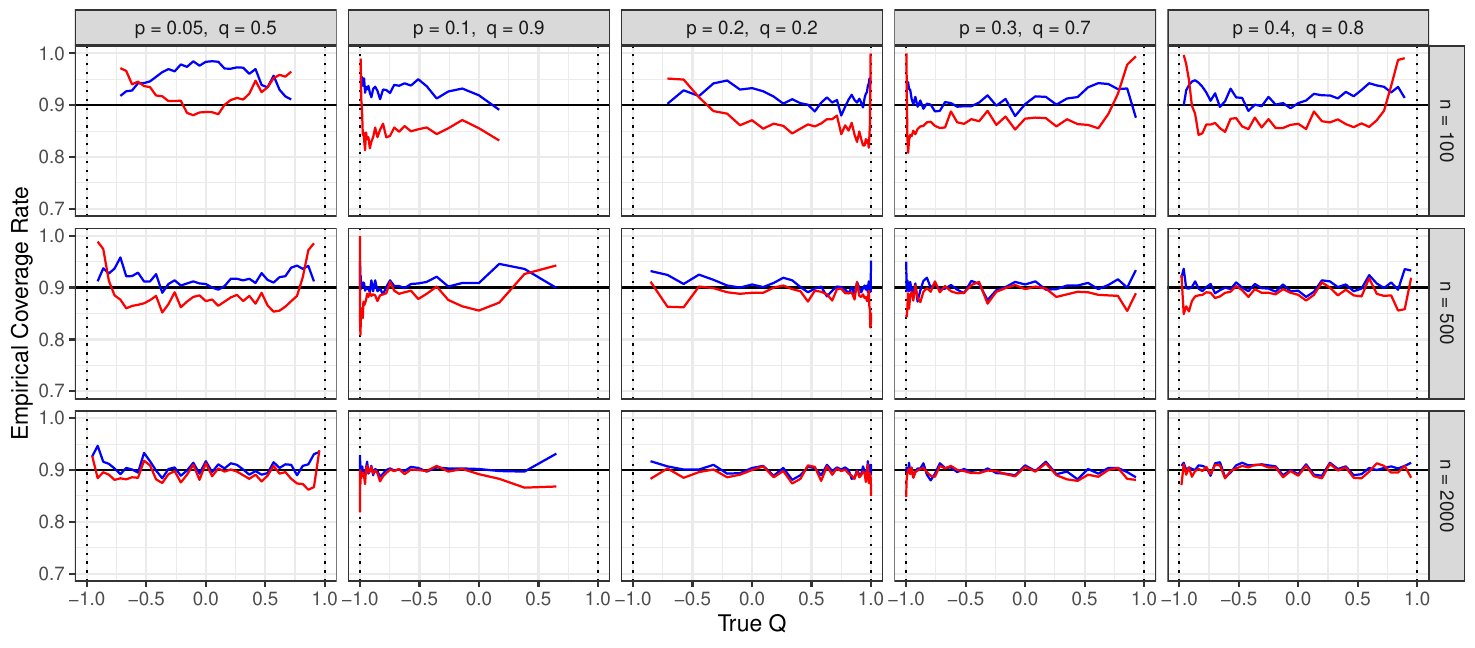}
	\includegraphics[width=\linewidth]{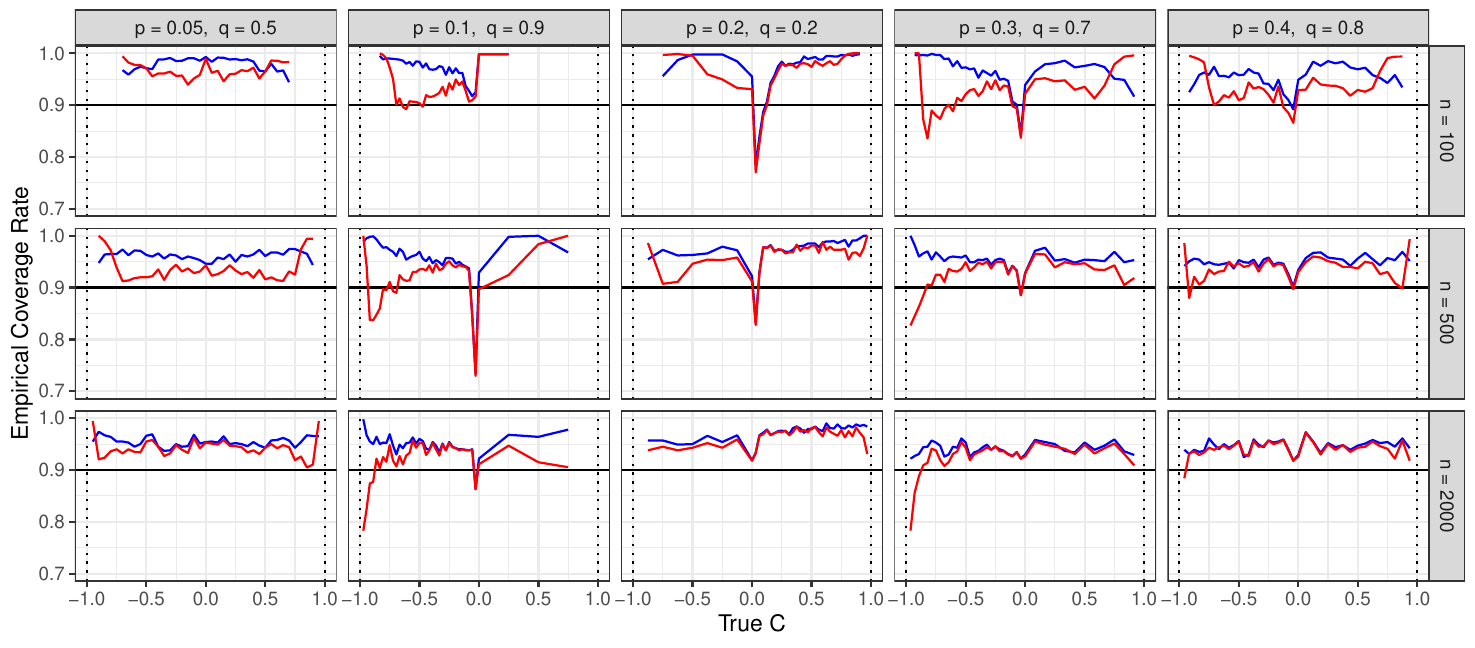}
	\includegraphics[width=\linewidth]{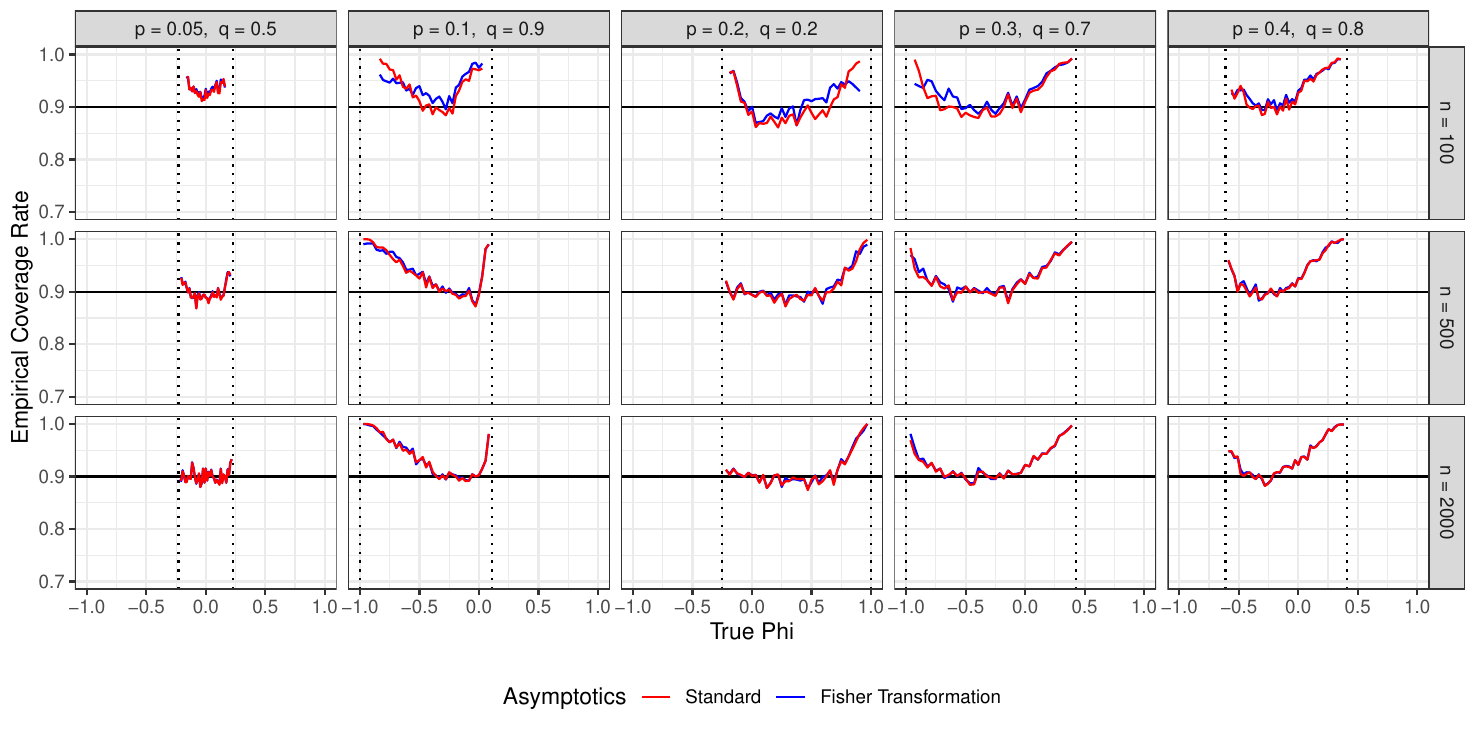}
	\caption{Empirical confidence interval coverage rates with nominal level $1-\alpha = 0.9$ for $M=1{,}000$ simulation replications for the measures $\myQ$ (top), $\myC$ (middle) and $\phi$ (bottom) for the sample sizes $n \in \{100;500;2{,}000\}$ and the probabilities $p$ and $q$ given in the column headings.
    The $x$-axis displays the true values for $\{\myQ, \myC, \phi\}$ that arise for the sequence of joint probabilities $r$.
    The respective Fr\'{e}chet--Hoeffding bounds are displayed as vertical dashed lines.
	Standard confidence intervals are given in red and confidence intervals based on the Fisher transformation in blue.}
	\label{fig:CI_CoverageRates}
\end{figure}

% \lukas{Maybe use $\phi$ in the R plot? I managed to do this via unicode: "$\backslash$u03D5", see tables at the end}
% \marc{Would be nice, but if it causes effort, we can also leave it as it is.}

Figure \ref{fig:CI_CoverageRates} shows the coverage rates of the confidence intervals that are constructed based on the standard and the Fisher-transformed asymptotic distributions for $\myQ$, $\myC$ and $\phi$ in the respective panels.
For $\myQ$, we obtain very accurate coverage rates that converge to the nominal level of $90\%$. 
While the standard approach displays inaccuracies close to the boundaries, the Fisher transformation expectedly improves the accuracy in these regions.

For $\myC$, we observe mostly conservative coverage rates that are caused by the multiple testing corrections in the construction of the confidence intervals.
The Fisher transformation again improves the accuracy close to the boundaries.
The notable downside peak in the coverage rates close to the value of $\myC = 0$ can be explained by the piecewise normalization of $\myC$.
In particular, for a true $\sigma > 0$, the last two terms in \eqref{eqn:C+4terms} involving the quantity $\mathds{1}\{ \hsigma < 0 \}$ in the proof of Proposition \ref{prop:Asymptotics-C} are shown to be $o_\P(1)$ and hence do not contribute to the asymptotic distribution, but deliver non-negligible contributions to the finite sample distributions.

For $\phi$, we again see the drastic effect of the non-attainability through the theoretical boundaries depicted with the vertical dashed lines.
While the coverage rates of the standard method are very accurate in the center, we observe an over-coverage closer to the boundaries that cannot be improved with the Fisher transformation.
As the Fisher transformation leaves the central part $[-0.5, \, 0.5]$ almost unchanged, it is ineffective when the theoretical boundaries of $\phi$ are far from $-1$ and $1$, once again illustrating a fundamental deficiency of the phi coefficient.

\begin{figure}[p]
	\centering
	\includegraphics[width=\linewidth]{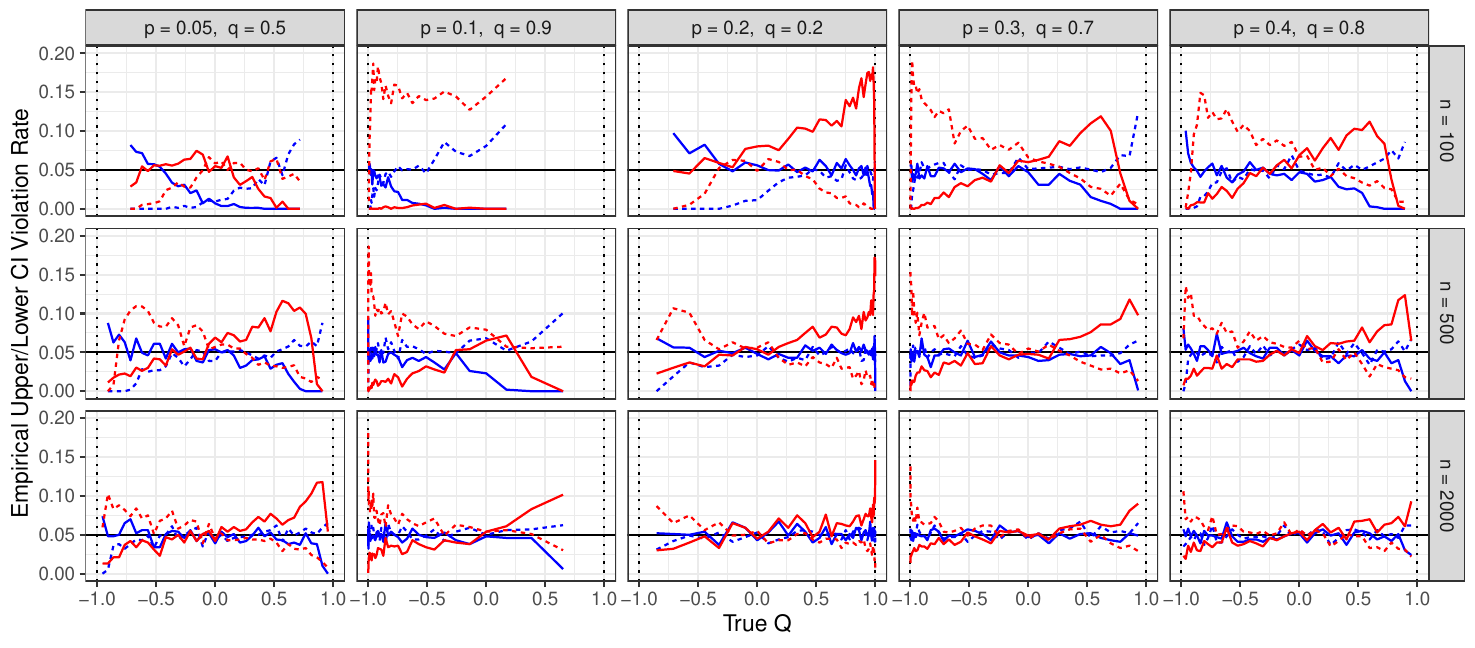}
	\includegraphics[width=\linewidth]{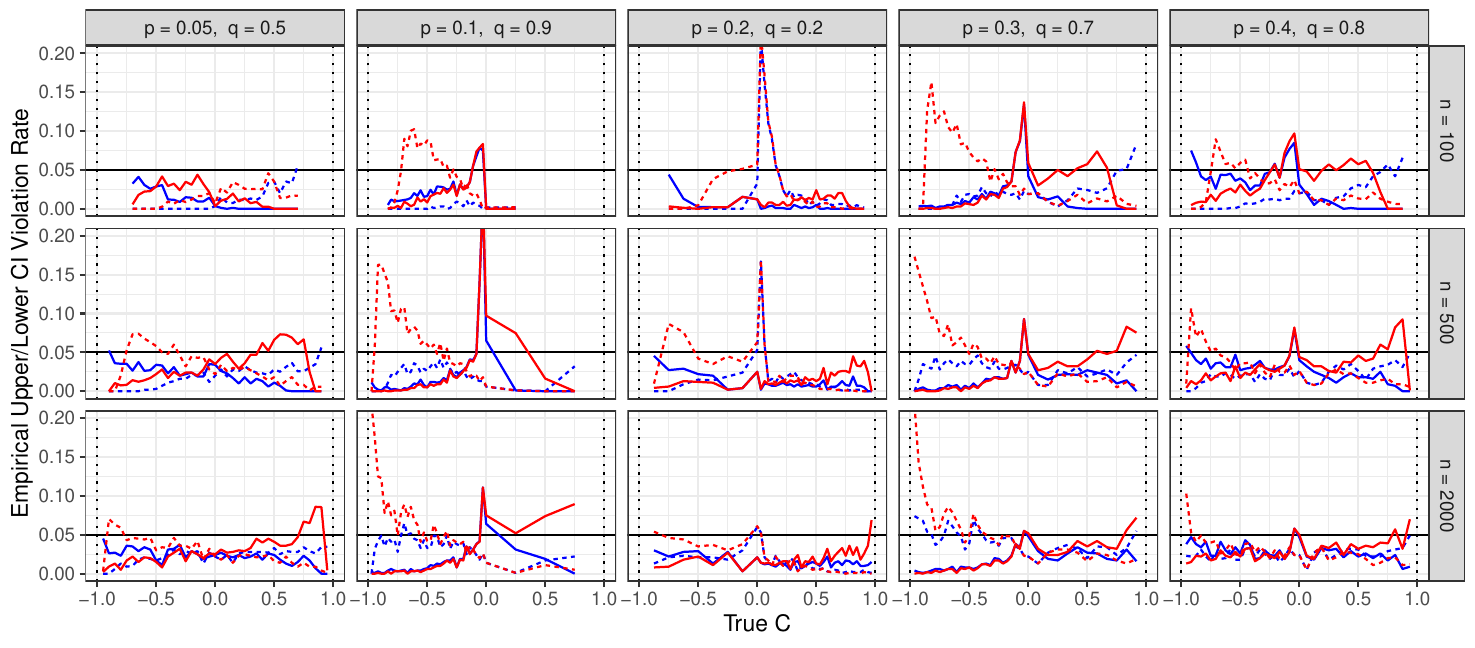}
	\includegraphics[width=\linewidth]{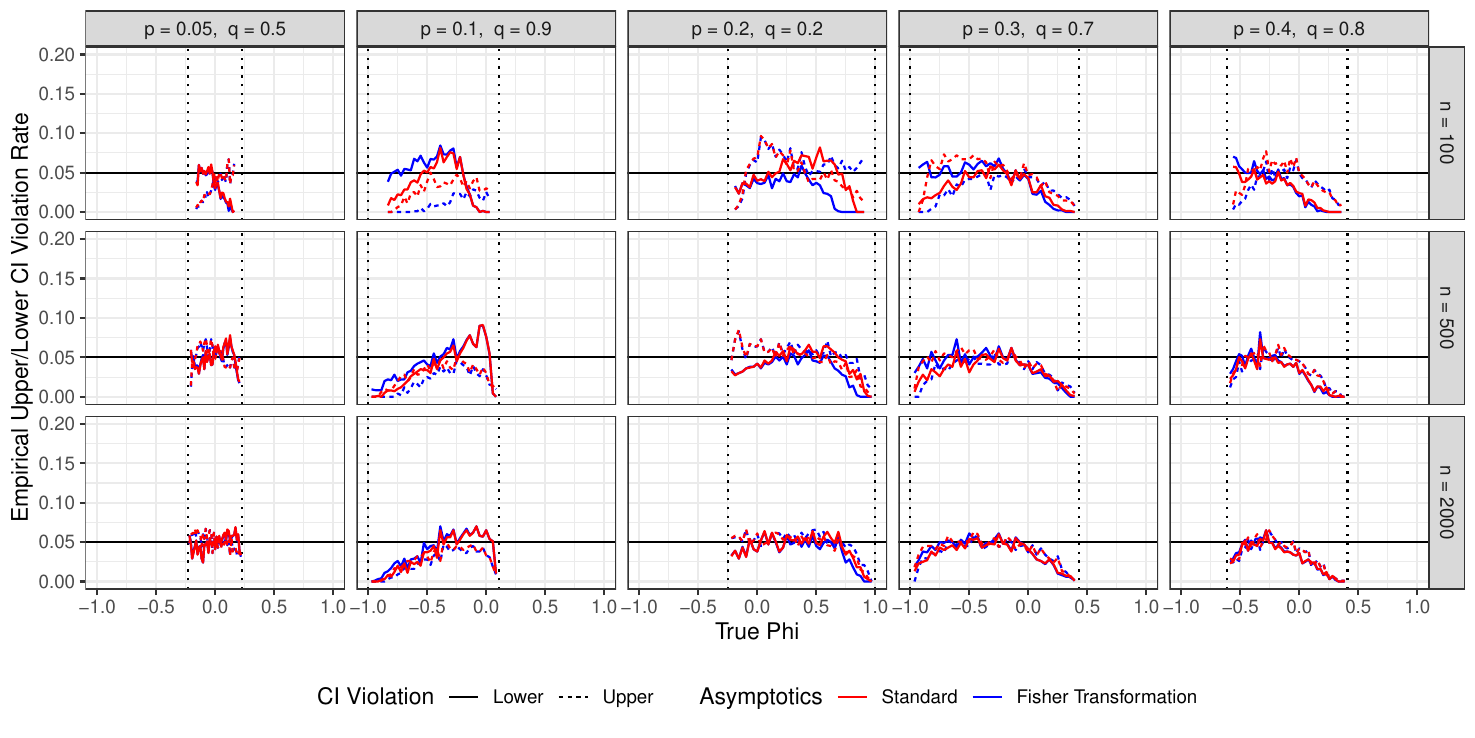}
 	\caption{Empirical violation rates of confidence intervals with nominal level $1-\alpha = 0.9$.
  	The ``lower'' (``upper'') violation rate reports the frequency of true values being below (above) the estimated confidence intervals.
    We show results for $M=1{,}000$ simulation replications for the measures $\myQ$ (top), $\myC$ (middle) and $\phi$ (bottom) for the sample sizes $n \in \{100;500;2{,}000\}$ and the probabilities $p$ and $q$ given in the column headings.
    The $x$-axis displays the true values for $\{\myQ, \myC, \phi\}$ that arise for the sequence of joint probabilities $r$.
    The respective Fr\'{e}chet--Hoeffding bounds are displayed as vertical dashed lines.
	Standard confidence intervals are given in red and confidence intervals based on the Fisher transformation in blue.}
	% \caption{Empirical violation rates of confidence intervals with nominal level $1-\alpha = 0.95$.
	% 	The ``lower'' (``upper'') violation rate reports the frequency of true values being below (above) the estimated confidence intervals.
	% 	We show results for $M=2000$ simulation replications for the measures $\myQ$ (top), $\myC$ (middle) and $\phi$ (bottom) for the sample sizes $n \in \{100,500,2000\}$ and the settings $(p,q) \in \big\{(0.05, 0.5), (0.2, 0.2), (0.3, 0.7), (0.4, 0.8) \big\}$ together with a full range of true attainable values (highlighted by the vertical dashed lines) for the true $\{\myC, \myQ, \phi\}$ on the $x$-axis.
	% Standard confidence intervals are given in green and confidence intervals based on the Fisher transformation in red.}
	\label{fig:CI_ViolationRates}
\end{figure}

Figure \ref{fig:CI_ViolationRates} refines the analysis by plotting the lower (upper) ``violation rates'', which are defined as the frequency of population dependence measures being below (above) the estimated confidence intervals.
For $\myQ$ in the upper panel, we especially see the asymmetry of the standard asymptotics that can be explained by approximating the skewed empirical distribution with the symmetric Gaussian limit; see in particular Figure \ref{fig:qsimulation}.
The Fisher transformation fixes this asymmetry and delivers confidence intervals that have approximately the same mass to either side.  
While similar effects can be observed for $\myC$, the Fisher transformation is again relatively ineffective for $\phi$.

\begin{table}[tb]
    \footnotesize
    \begin{tabular}{lclcl}
		\toprule
		Method Name && Hypothesis Combination && Valid $p$-values \\
		\midrule 
		Full &&  $\H_c = \{ [ (\H_= \cap \H_{p,q}) \cup \H_{\not=}] \cap \H_\sigma \}$ & &  $2 \min \{ \max[ 2 \min (p_=, p_{p,q}), p_{\not=}], p_\sigma \}$   \\     
  		No $\sigma$-Test &&  $\H_c = [ (\H_= \cap \H_{p,q}) \cup \H_{\not=}]$ & &  $\max[ 2 \min (p_=, p_{p,q}), p_{\not=}]$   \\     
    	No $p,q$-Test &&  $\H_c = \{ [ \H_= \cup \H_{\not=}] \cap \H_\sigma \}$ & &  $2 \min \{ \max[p_=, p_{\not=}], p_\sigma \}$   \\     
      	Basic &&  $\H_c =  [ \H_= \cup \H_{\not=}]$ & &  $\max[ p_=, p_{\not=}]$   \\  
       \midrule
       \multicolumn{5}{l}{Elementary Hypotheses:} \\
       For $c > 0$: &&
       \multicolumn{3}{l}{
        $\H_= = \{C=c, p=q\}, \; 
        \H_{\not=} = \{C=c, p \not= q\}, \; 
        \H_{p,q} = \{p=q\}, \;
        \H_\sigma = \{\sigma \ge 0\}$} \\
        For $c < 0$: &&
       \multicolumn{3}{l}{
        $\H_= = \{C=c, p=1-q\}, \; 
        \H_{\not=} = \{C=c, p \not= 1-q\}, \; 
        \H_{p,q} = \{p=1-q\}, \;
        \H_\sigma = \{\sigma \le 0\}$} \\
        \bottomrule
	\end{tabular}
    \caption{Overview of the different theoretically valid combination methods of elementary hypotheses and corresponding $p$-values used for the construction of confidence intervals for $\myC$ in Figure \ref{fig:CI_C_Details}.}
    \label{tab:CICConstructions}
\end{table}

\begin{figure}[tb]
	\centering
	\includegraphics[width=\linewidth]{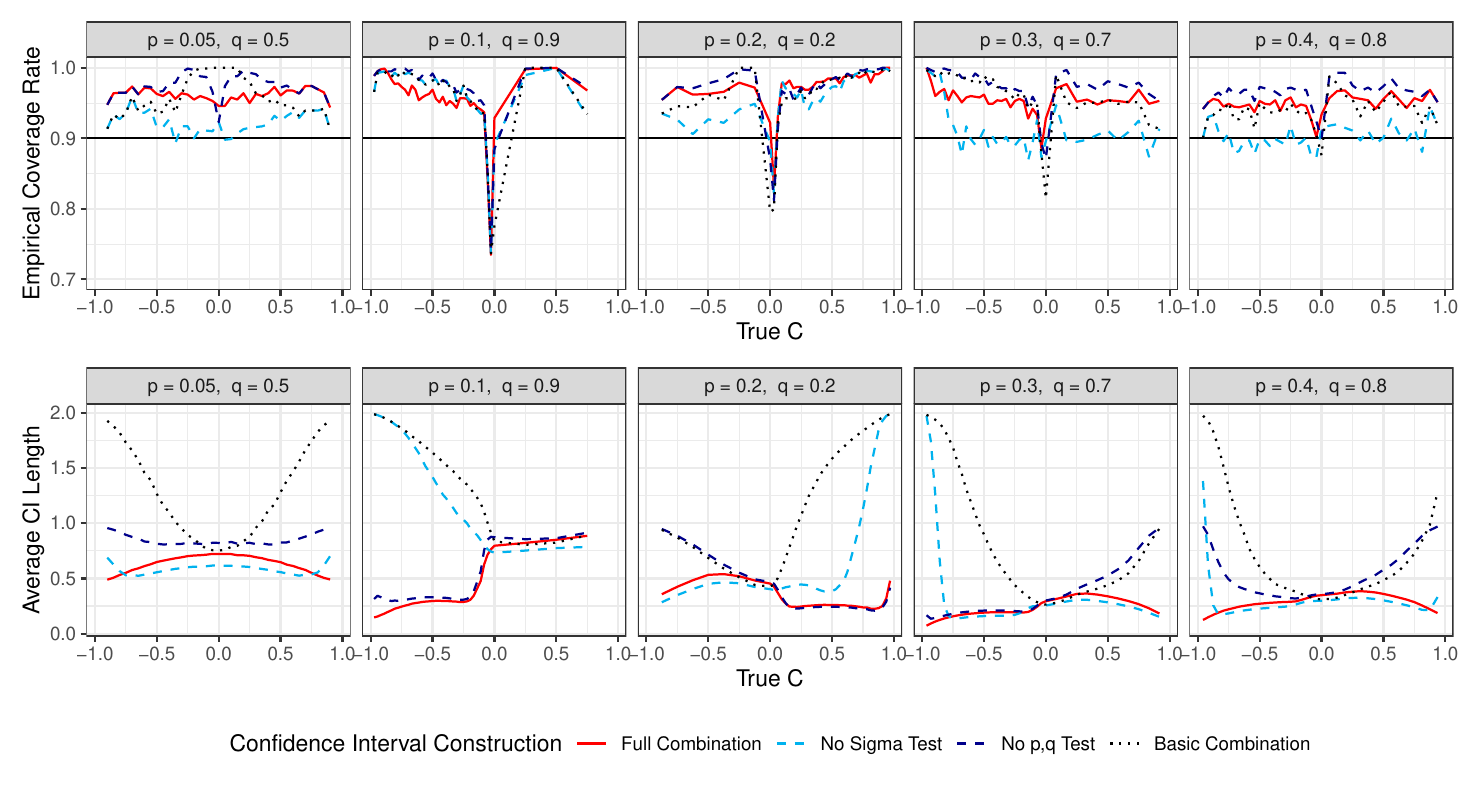}
	\caption{Empirical coverage rates (upper panel) and average confidence interval length (lower panel) with level $1-\alpha = 0.9$ based on the different construction principles given in Table \ref{tab:CICConstructions} for $\myC$ using the Fisher transformation.
    The $x$-axes display the true values of $\myC$ that arise for the sequence of joint probabilities $r$.}
	\label{fig:CI_C_Details}
\end{figure}

We continue to illustrate the benefits of the combination of elementary hypotheses in the construction of tests for $\H_c = \{C=c\}$ for any $c \not= 0$ as discussed in Section \ref{subsec:CIs}.
For this, we consider simplified methods that omit one/both tests for $\H_{p,q}$ and $\H_\sigma$ respectively as illustrated in detail in Table \ref{tab:CICConstructions}.
Figure \ref{fig:CI_C_Details} displays the empirical coverage rates and average lengths of confidence intervals for $\myC$ based on these different construction principles for a fixed sample size of $n=500$ and by using the Fisher transformation.
We see that all confidence intervals have satisfactory coverage rates that are often conservative due to the multiple testing corrections.
However, the lengths of the different methods differ substantially, whereas the ``full'' method we propose achieves overall very narrow intervals.
While it is often not the most narrow one as its inherent multiple testing corrections result in conservative coverage rates, it is in all settings close to the most narrow method.
Which method is the most narrow one however depends on practically unknown population quantities, and hence cannot be known in practice.
In contrast, omitting the respective tests for $\H_{p,q}$ and $\H_\sigma$ results in drastic increases of the length in certain situations.
Overall, the ``full'' method is relatively narrow (powerful) in all settings and trades some length in well-behaved settings in order to insure against very unfortunate cases.

\section{Proofs} 
\label{app:proofs}

Here, we give the proofs for all ``conceptual'' results of Sections \ref{sec:dependence_concepts}--\ref{sec:proper_measures} and Appendix \ref{sec:FurtherMeasures} in Section \ref{sec:ProofsConcept}, and continue to present the proofs for the ``asymptotic'' results  of Section \ref{sec:asymptotics} in Section 
\ref{sec:ProofsAsymptotics}.
Section \ref{sec:AddLemmas} presents some additional lemmas together with their proofs.

% \subsection{Proofs for Theory Results}
\subsection{Proofs for the Results in Sections \ref{sec:dependence_concepts}--\ref{sec:proper_measures} and Appendix \ref{sec:FurtherMeasures}}
\label{sec:ProofsConcept}

\begin{proof}[Proof of Proposition \ref{prop:characterizations}]
	For the first part of (i) let $\P(A) \leq \P(B)$ w.l.o.g. Then it holds that: $A$ and $B$ are perfectly positively dependent $\iff \P(A \setminus B)=0 \iff \P(A \cap B) = \P(A) = \min(\P(A),\P(B))$.\\
	For the second part of (i) note that according to Lemma \ref{lemma:complement_dependence} (iii) $A$ and $B$ are perfectly negatively dependent if and only if $A$ and $\overline{B}$ are perfectly positively dependent, which is equivalent to $\P(A \cap \overline{B}) = \min(\P(A),\P(\overline{B}))$, which is in turn equivalent to $\P(A \cap B) = \P(A) - \P(A \cap   \overline{B}) = \P(A) - \min(\P(A),\P(\overline{B}))= \max(0,\P(A) + \P(B) - 1)$.\\
	The upper bound in the third part of (i) follows by $\{A \cap B\} \subseteq A$ and $\{A \cap B\} \subseteq B$ and the monotonicity of probability. The lower bound follows by using the upper bound  for the inequality in $P(A \cap B) = \P(A) - \P(A \cap   \overline{B}) \geq \P(A) - \min(\P(A),\P(\overline{B}))= \max(0,\P(A) + \P(B) - 1)$.
	(ii) and (iii) follow directly from Definition \ref{def:perfect_dependence}.
\end{proof}

\begin{proof}[Proof of Lemma \ref{lemma:complement_dependence}]
	For (i), it holds that $\P(A)=\P(A \cap B) + \P(A \cap \overline{B})$, which implies that positive dependence between $A$ and $B$ is equivalent to $\P(A) - \P(A \cap \overline{B}) \ge \ \P(A)\P(B)$, which is in turn equivalent to $\P(A\cap \overline{B}) \le \ \P(A)\P(\overline{B})$, that is, to negative dependence between $A$ and $\overline{B}$.

    (ii) follows by using $\P(A \cap B) + \P(A \cap \overline{B}) = \P(A)=\P(A^*) = \P(A^* \cap B^*) + \P(A^* \cap \overline{B^*})$ and (iii) is immediate from Definition \ref{def:perfect_dependence}.
\end{proof}

\begin{proof}[Proof of Proposition \ref{cor:pos_neg_dependence}]
	$A$ and $B$ being positively (negatively) dependent is equivalent to them being stronger positively (negatively) dependent than two independent events $A^*$ and $B^*$ with the same marginal probabilities. By (B) it holds that $\delta(A^*,B^*)=0$. Thus, by (D), $\delta(A,B) \geq (\leq)\ 0$ is equivalent to positive (negative) dependence of $A$ and $B$.
\end{proof}

\begin{proof}[Proof of Proposition \ref{prop:characterizations_linear_dependence}]
	The equivalence of (i) and (ii) is established by checking under which conditions the bounds from \eqref{eq:phibounds} become 1 and $-1$.\\
	(ii) immediately implies the second condition in (iii) and the first condition follows directly from Definition \ref{def:perfect_dependence}, remembering that the symmetric difference can be written as $A \Delta B = (A \setminus B) \cup (B \setminus A)$.\\
	(iii) implies (ii): Consider the case of positive dependence. If $\P( A \setminus B) = 0$ or $\P( B \setminus A) = 0$ holds, then under $\P(A)=\P(B)$ both conditions must hold, which implies that $\P (A \Delta B) = \P( (A \setminus B) \cup (B \setminus A)) \leq \P( A \setminus B) + \P(B \setminus A) =0$; analogously for negative dependence.\\
	(iv) is a simple reformulation of (ii).
	(v) is a reformulation of (ii) in terms of indicators, noting that $\mathds{1}_{\overline{B}} = 1-\mathds{1}_B$.\\
	(vii) is immediate from (v).\\
	(vi) is immediate from (v).\\
	(v) follows from (vi) as binary variables have to be equal if they are increasing functions of each other and for a decreasing functional relationship one variable has to equal the negative of the other plus 1.
\end{proof}

To prove Proposition \ref{Cole_propriety}, we need the following lemma.
\begin{lemma} \label{lemma:FHbounds_relation}
	Denote the Fr\'{e}chet--Hoeffding upper and lower bound of $\Cov(A,B)$ from \eqref{eq:covbounds} by
	$$\myFH^+ (A,B) = \min(\P(A),\P(B)) - \P(A) \P(B)$$
	and
	$$\myFH^- (A,B) = \max(0,\P(A) + \P(B) -1) - \P(A) \P(B).$$
	It holds that 
	$$\myFH^+ (A,\overline{B}) = - \myFH^- (A,B).$$
\end{lemma}

\begin{proof}[Proof of Lemma \ref{lemma:FHbounds_relation}]
	By part (i) of Proposition  \ref{prop:characterizations} and Lemma \ref{lemma:complement_dependence} (iii) it holds that $\myFH^+ (A,\overline{B})=\Cov(A,\overline{B})$ and $\myFH^- (A,B) = \Cov(A,B)$ if $A$ and $B$ are perfectly negatively dependent. By the second part of property (E) of $\Cov(A,B)$ (as established in Proposition \ref{prop:covariance_properties}) the claim follows.
\end{proof}

\begin{proof}[Proof of Proposition \ref{Cole_propriety}]
	(B): As the normalization terms are positive it holds that $\myC(A,B)=0$ if and only if $\Cov(A,B)=0$ and since $\Cov(A,B)$ fulfills property (B) by Proposition \ref{prop:covariance_properties}, $\myC(A,B)$ does so as well.
	
	(D): For fixed marginal probabilities $\myC(A,B)$ is a strictly monotonic function of $\Cov(A,B)$ and thus fulfills property (D) because $\Cov(A,B)$ does.
	
	(E): The normalization terms and $\Cov(A,B)$ are symmetric and thus $\myC(A,B)$ is as well, $\myC(A,B) = C(B,A)$. For the second symmetry property, we make a case distinction: For $\Cov(A,B) = 0$, the analogous property of covariance implies $0=\Cov(A,\overline{B})=-\Cov(A,B)$, which establishes the claim (as the normalization is nonzero). For $\Cov(A,B) > 0$, we have by the analogous property of $\Cov(A,B)$ and by Lemma \ref{lemma:FHbounds_relation} that $\myC(A,\overline{B}) = \frac{\Cov(A,\overline{B})}{\myFH^+ (A,\overline{B})}=  \frac{\Cov(A,B)}{\myFH^- (A,B)} = - \myC(A,B)$. The proof works analogously for $\Cov(A,B) < 0$.
\end{proof}

\begin{proof}[Proof of Lemma \ref{lemma:MSC_phi}]
	$\MSC(A,B)$ from \eqref{eq:mean_square_contingency} can be rewritten as 
	$$\MSC(A,B) = \phi(A,B)^2 \P(\overline{A}) \P(\overline{B}) + \phi(A,\overline{B})^2 \P(\overline{A}) \P(B) + \phi(\overline{A},B)^2 \P(A) \P(\overline{B}) + \phi(\overline{A},\overline{B})^2 \P(A) \P(B),$$
	which by the symmetry properties of the phi coefficient (Proposition \ref{prop:phi_properties}), $\phi(A,B) = \phi(\overline{A},\overline{B}) = - \phi(\overline{A},B) = - \phi(A,\overline{B})$, becomes
	$$\MSC(A,B) = \phi(A,B)^2 \left( \P(\overline{A}) \P(\overline{B}) + \P(\overline{A}) \P(B) + \P(A) \P(\overline{B}) +  \P(A) \P(B) \right)= \phi(A,B)^2.$$
\end{proof}

\begin{proof}[Proof of Proposition \ref{prop:properties_OR}]
	(A), (B) and (E) follow directly from Definition \ref{def:odds_ratio}, for (C) we additionally invoke part (iii) of Proposition \ref{prop:characterizations}. For (D) consider again Definition \ref{def:odds_ratio} and let the event probabilities $\P(A)$ and $\P(B)$ be fixed. Note that then $\P(\overline{A} \cap \overline{B})$ increases (decreases) and $\P(\overline{A} \cap B)$ and $\P(A \cap \overline{B})$ decrease (increase) if and only if $\P(A \cap B)$ increases (decreases). Thus, the numerator increases (decreases)  if and only if $\P(A \cap B)$ increases (decreases), which proves the claim. 
\end{proof}

%\begin{proof}[Proof of Lemma \ref{lem:Q_a}]
%	
%	Denote $w:=\P(A\cap B), x:=\P(\overline{A}\cap \overline{B}), y:=\P(\overline{A}\cap B), z:=\P(A\cap \overline{B})$.
%	
%	It holds that $$\frac{dQ_a(A, B)}{da}=-\dfrac{2\left(wx\right)^a\left(yz\right)^a\left(\ln\left(yz\right)-\ln\left(wx\right)\right)}{\left(\left(yz\right)^a+\left(wx\right)^a\right)^2}$$
%	
%	and thus $$sgn\left(\frac{dQ_a(A, B)}{da}\right)=sgn\left(\Q_a(A,B)\right)=sgn\left(\Cov(A,B)\right)$$
%	
%	with $sgn(x):=\mathds{1}_{\{x>0\}}-\mathds{1}_{\{x<0\}}$.
%\end{proof}

\begin{proof}[Proof of Proposition \ref{prop:tetrachoric_propriety}]
	Property (A) is clear as the Pearson correlation coefficient lies between -1 and 1. (B) has been established in \citet[Proposition 6]{Ekstrom2011} already. (C) follows directly from the definition of $T(A,B)$. Theorem 5 in \citet{Ekstrom2011} implies that $T(A,B)$ and $\phi(A,B)$ are strictly monotonic functions of each other and as (D) holds for the phi coefficient (Proposition \ref{prop:phi_properties}), it thus also holds for tetrachoric correlation. (E) follows from symmetry properties of the bivariate normal distribution.
\end{proof}

\subsection{Proofs for the Asymptotic Results in Section \ref{sec:asymptotics}}
\label{sec:ProofsAsymptotics}
%\color{red}
%As mentioned in Section \ref{sec:asymptotics}, we prove the theorems formulated under assumptions \ref{ass:BoundedAway}--\ref{ass:iid} in this Section under the more general assumptions \ref{ass:BoundedAway}, \ref{ass:PositiveRelativeFrequencies}, \ref{ass:Dependence} and \ref{ass:Multicolinearity} discussed in Subsection \ref{subsec:time_series_asymptotics}. Note that assumptions \ref{ass:BoundedAway} and \ref{ass:iid} imply Assumption \ref{ass:Dependence} and Assumption \ref{ass:BoundedAway} implies Assumption \ref{ass:Multicolinearity}.
%\color{black}

\begin{proof}[Proof of Proposition~\ref{prop:Asymptotics-Q}]
    As $\p$ and $\q$ are consistent estimators for $p$ and $q$ and as Assumption \ref{ass:BoundedAway} is in place, we have that $\mathbb{P}\big( \{\p \in \{0,1\}\} \cup  \{\q \in \{0,1\}\} \big) \to 0$ such that $\myQn$ is well-defined asymptotically with probability one.   
    Formally, for any set $D \in \F$, we have that
    \begin{align}
        \begin{aligned}
        \label{eqn:WellDefinedProbOne}
        \mathbb{P} \big( \{\sqrt{n} \big( \myQn - \myQ \big)  \in D \} \big) 
        &= 
        \mathbb{P} \big( \{\sqrt{n} \big( \myQn - \myQ \big)  \in D \} \cap  \{0 < \p, \q < 1\}  \big) \\
        &\qquad + \mathbb{P} \big( \{\sqrt{n} \big( \myQn - \myQ \big)  \in D \} \cap  \{ \p \in \{0, 1\} \text{ or } \q \in \{0, 1\} \} \big)  \\
        &= \mathbb{P} \big( \{\sqrt{n} \big( \myQn - \myQ \big) \in D \} \cap  \{0 < \p, \q < 1\}  \big) + o_\mathbb{P}(1),
        \end{aligned}
    \end{align}
    such that the $o_\mathbb{P}(1)$-part does not contribute to the asymptotic distribution of $\sqrt{n} \big( \myQn - \myQ \big)$ and we continue to work under the condition that $\{0 < \p, \q < 1\}$.

    % As $\p$ and $\q$ and $\r$ are consistent estimators for $p, q$ and $r$ and as Assumption \ref{ass:BoundedAway} is in place, there must exist an $\tilde{n} \in \mathbb{N}$ large enough, such that for all $n \in  \mathbb{N}$, $n \ge \tilde{n}$, the conditions in \eqref{eq:PositiveRelativeFrequencies} also hold with probability one for the empirical counterparts to Assumption \ref{ass:BoundedAway}. Hence, the estimator $\myQn$ in \eqref{eqn:Qest} is well defined for all $n \in \mathbb{N}$ large enough.
    
	We can write
	\begin{align}
		\label{eqn:Q_g}
		\myQn = g \big( \p, \q, \r \big), \quad \text{with} \quad
		g(\pa, \qa, \ra) := \frac{\ra - \pa \qa}{\ra (1-\pa-\qa+\ra) + (\qa-\ra)(\pa-\ra)},
	\end{align}
	which is differentiable at $(p, q, r)$ as $(q - r)(p - r) > 0$ and $r (1 - p - q + r) > 0$ by Assumption \ref{ass:BoundedAway}.
	The function $g$ has the (non-zero) Jacobian matrix of continuous partial derivatives %$J_g$ (when replacing $p$ by $p$) given in Proposition \ref{prop:Asymptotics-Q}.
	\begin{align}
		\label{eqn:QJacobian}
		J_g(\pa, \qa, \ra) = \frac{2}{\big(\pa (\qa - 2 \ra) + \ra (1 - 2\qa + 2\ra)\big)^2}
		\begin{pmatrix}
			(\qa - 1) \ra (\qa - \ra) \\
			(\pa - 1) \ra (\pa - \ra) \\
			- (\pa^2 \qa + \pa \qa (-1 + \qa - 2 \ra) + \ra^2)
		\end{pmatrix}^\top.
	\end{align}
	As the function $g$ is almost surely continuous at $(p, q, r)$ as $r < q$ and $r < p$, the continuous mapping theorem implies that $\myQn \toP \myQ$.
	
	We continue to show asymptotic normality for $\myQn$ by using Lemma \ref{lemma:CLTjoint} and applying the delta method. 
	For this, we use Lemma \ref{lem:denominator} (reported at the end of this section) to show that the denominator of the fraction in \eqref{eqn:QJacobian} is non-zero for all $(p, q, r)$ satisfying Assumption \ref{ass:BoundedAway}.

	Hence, by the delta method \cite[Theorem 3.1]{VanderVaart2000}, we can conclude that
		\begin{align*}
			\sqrt{n} \big( \myQn - \myQ \big) 
			&= \sqrt{n} \big( g(\p, \q, \r)  - g(p, q, r) \big) 
			=  J_g(p, q, r)^\top \sqrt{n} 
			\begin{pmatrix}
				\p - p \\ \q - q \\ \r - r
			\end{pmatrix} + o_\mathbb{P}(1).
		\end{align*}
	As $\sqrt{n} \big( \p - p, \; \q - q, \; \r - r \big)^\top \tod  \mathcal{N}(0, \Omega)$ by Lemma \ref{lemma:CLTjoint} under either Assumption \ref{ass:iid} or Assumption \ref{ass:Dependence}, the result follows with the help of Theorem 2.7 in \citet{VanderVaart2000}.
\end{proof}

\begin{proof}[Proof of Proposition \ref{prop:Asymptotics-C}]
	% We use the notation $\hsigma = \r - \p \q$ for the estimated covariance. 
    We start to note that Assumption \ref{ass:BoundedAway} implies that $\mp, \mn > 0$.
    With the same arguments as in the beginning of the proof of Proposition~\ref{prop:Asymptotics-Q}, we get that $\hmp, \hmn > 0$ with asymptotic probability one.
    With the same arguments as in \eqref{eqn:WellDefinedProbOne}, we continue to work under the condition that $\hmp, \hmn > 0$ in order to derive the distributional limit of $\sqrt{n} \big(\myCn - \myC \big)$.

	Straightforward calculations yield that
	\begin{align}
		\label{eqn:Cgeneral4terms}
		\sqrt{n} \big(\myCn - \myC \big)
		&=  \sqrt{n}  \mathds{1}\{\hsigma > 0\}  \frac{\hsigma}{\hmp}
		+ \sqrt{n} \mathds{1}\{\hsigma < 0\} \frac{ \hsigma}{\hmn} 
		- \sqrt{n}  \mathds{1}\{\sigma > 0\}  \frac{\sigma}{\mp} 
		- \sqrt{n}  \mathds{1}\{\sigma < 0\} \frac{\sigma}{\mn}.
	\end{align} 	
	We now distinguish the three cases that the true $\sigma$ is positive, negative, or zero:
	
	\vspace{0.3cm}
	\noindent
	$\bullet$
	\textbf{In the zero case}, let $\sigma = \Cov(A,B) = 0$ such that $\myC = 0$. We get from \eqref{eqn:Cgeneral4terms} that
	\begin{align*}
		\sqrt{n} \big(\myCn  - \myC \big)
		&= \sqrt{n}  \mathds{1}\{\hsigma > 0\}  \frac{ \hsigma}{\hmp}
		+ \sqrt{n}  \mathds{1}\{\hsigma < 0\} \frac{\hsigma}{\hmn} \\
		&=  \frac{1}{\hmp} \sqrt{n} \big( \hsigma - \sigma \big) \mathds{1}\{ \sqrt{n} (\hsigma - \sigma) > 0 \}
		+  \frac{1}{\hmn} \sqrt{n} \big( \hsigma - \sigma \big) \mathds{1}\{ \sqrt{n} (\hsigma - \sigma) < 0 \}.
	\end{align*}
	As $1/\hmp \toP 1/\mp$ and $1/\hmn \toP 1/\mn$ by the continuous mapping theorem, and  $\big( \Delta^\top \Omega \Delta \big)^{-1/2} \sqrt{n} (\hsigma - \sigma) \tod \mathcal{N}(0,1)$ with $\Delta = (-q, -p, 1)^\top$ by Lemma \ref{lemma:CLT_sigma},
	we get for some $Z \sim \mathcal{N}(0,1)$ that
	\begin{align*}
		\big(\Delta^\top \Omega \Delta \big)^{-1/2} \sqrt{n} \big(\myCn - \myC \big) \tod \frac{Z \, \mathds{1}\{Z > 0\}}{\mp} +  \frac{Z \, \mathds{1}\{Z < 0\}}{\mn}.
	\end{align*}
	
	\vspace{0.3cm}
	\noindent
	$\bullet$
	\textbf{In the positive case}, let $\sigma = \Cov(A,B) > 0$ such that $\myC > 0$. 
	Then,
	\begin{align}
		\begin{aligned} 
			\label{eqn:C+4terms}
			\sqrt{n} \big(\myCn - \myC \big)
			&=  \sqrt{n} \left( \frac{ \hsigma}{\hmp} \mathds{1}\{ \hsigma > 0 \} - \frac{\sigma}{\mp} \right)
			+ \sqrt{n}  \frac{\hsigma}{\hmn} \mathds{1}\{ \hsigma < 0 \} \\
			&= \sqrt{n}  \frac{\hsigma \mathds{1}\{ \hsigma > 0 \} - \sigma }{\hmp}   
			+ \sqrt{n}  \left( \frac{\sigma}{\hmp} - \frac{\sigma}{\mp} \right) +   \sqrt{n}   \frac{\hsigma}{\hmn} \mathds{1}\{\hsigma < 0 \} \\
			&= \sqrt{n}  \frac{\hsigma - \sigma}{\hmp} + \sqrt{n} \left( \frac{\sigma (\mp - \hmp )}{\mp \hmp} \right)  - \sqrt{n} \frac{\hsigma}{\hmp}  \mathds{1}\{ \hsigma \le 0 \} +   \sqrt{n}  \frac{\hsigma}{\hmn} \mathds{1}\{ \hsigma < 0 \}.
		\end{aligned}
	\end{align} 	
	We continue by showing that the first two terms on the right-hand side of \eqref{eqn:C+4terms} contribute to the asymptotic distribution, whereas the last two are $o_\mathbb{P}(1)$. 
	For the latter, we show $L_1$-convergence of the term $\sqrt{n} \mathds{1}\{ \hsigma \le 0 \}$ by using Chebyshev's inequality.
	Notice that $\mathbb{E}[\hsigma] = \frac{n-1}{n} \sigma$ such that
	\begin{align*}
		\mathbb{E} \big| \sqrt{n} \mathds{1}\{ \hsigma \le 0 \} \big|
		&= \mathbb{E} \left[ \sqrt{n} \mathds{1}\{ \sqrt{n} ( \hsigma -\sigma) \le  - \sqrt{n} \sigma \} \right] 
		= \sqrt{n} \, \mathbb{P} \big( \sqrt{n} ( \hsigma - \sigma) \le  - \sqrt{n} \sigma \big) \\ 
		&= \sqrt{n} \, \mathbb{P} \left( \sqrt{n} \left(  \hsigma - \frac{n-1}{n} \sigma \right) \le - \sqrt{n} \frac{n-1}{n} \sigma \right) \\
		&\le \sqrt{n} \, \mathbb{P} \left( \left| \sqrt{n} \left( \hsigma - \frac{n-1}{n} \sigma \right) \right| \ge  \sqrt{n} \frac{n-1}{n} \sigma \right) \\
		&\le \sqrt{n} \frac{\Var \left( \sqrt{n} \left( \hsigma - \frac{n-1}{n} \sigma \right) \right)}{n \frac{(n-1)^2}{n^2} \sigma^2} 
		= n^{-1/2} \frac{n^2}{(n-1)^2} \frac{\Var( \sqrt{n} \hsigma )}{\sigma^2} \to 0,
	\end{align*}
	as a CLT applies to $\sqrt{n} (\hsigma - \sigma)$, and hence its variance is bounded from above for large $n$.
	As $\hmp \toP \mp > 0$, $\hsigma \toP  \sigma$, and $L_1$-convergence implies convergence in probability, the continuous mapping theorem yields that $  \sqrt{n} \frac{\hsigma}{\hmp}  \mathds{1}\{ \hsigma \le 0 \}  \toP 0$.
	Equivalent arguments show that  $\sqrt{n}  \frac{ \hsigma}{\hmn} \mathds{1}\{ \hsigma < 0 \} \toP 0$.
	%\todo{I believe these arguments are correct. Alternatively, but much more complicated, one could apply Barry-Esseen type inequalities with extensions given in \citet{Jirak2016, PinelisMolzon2016}.}
	
	For the first two terms on the right-hand side of \eqref{eqn:C+4terms}, notice that $\hmp$ and $\mp$ are strictly positive. 
	Hence, we get that
	\begin{align*}
		\sqrt{n} \frac{ \hsigma - \sigma }{\hmp} + \sqrt{n} \left( \frac{\sigma ( \mp - \hmp)}{\mp \hmp} \right)
		&= \frac{1}{\hmp} \sqrt{n}  \big( \hsigma - \sigma \big) - \frac{\sigma}{\mp \hmp}  \sqrt{n}  \big( \hmp - \mp \big)   \\
		&= 
		\begin{pmatrix}
			\frac{1}{\hmp}  & 
			- \frac{\sigma}{\mp \hmp}
		\end{pmatrix} 
		\times
		\begin{pmatrix}
			\sqrt{n} \big( \hsigma - \sigma \big) \\
			\sqrt{n} \big( \hmp - \mp  \big) 
		\end{pmatrix}.
	\end{align*}
	We now derive the joint asymptotic distribution of $\sqrt{n} \big( \hsigma - \sigma, \; \hmp - \mp \big)^\top$, taking into account the two different cases $p = q$ and $p \not= q$.
	
	First, consider any $p \not= q$ that satisfies the assumptions of the theorem.
	Then, 
	\begin{align*}
		\begin{pmatrix}
			\sqrt{n} \big( \hsigma - \sigma \big) \\
			\sqrt{n} \big( \hmp - \mp  \big) 
		\end{pmatrix}
		= \sqrt{n} \big( h^+(\p, \q, \r) - h^+(p, q,r) \big),
		\quad \text{with} \quad 
		h^+(\pa, \qa, \ra) := 
		\begin{pmatrix}
			\ra - \pa \qa \\
			\min(\pa, \qa) - \pa \qa
		\end{pmatrix},
	\end{align*}
	which is continuously differentiable at $(p, q, r)$ as $p \not=q$ with Jacobian matrix at point $(p, q, r)$ 
	\begin{align*}
		J_{h^+} = 
		\begin{pmatrix}
			-q  & -p & 1 \\
			\mathds{1}(p < q) - q& \mathds{1}(q < p) - p & 0 
		\end{pmatrix}.
	\end{align*}
	
	As $\, \Omega^{-1/2} \sqrt{n} \big( \p - p, \; \q - q, \;  \r  - r \big)^\top \stackrel{d}{\to} \mathcal{N} \big(0, I_3 \big)$ holds by Lemma \ref{lemma:CLTjoint}, the delta method \cite[Theorem 3.1]{VanderVaart2000} yields that
	\begin{align*}
		\big( J_{h^+} \Omega J_{h^+}^\top \big)^{-1/2} 
		\begin{pmatrix}
			\sqrt{n} \big( \hsigma - \sigma \big) \\
			\sqrt{n} \big( \hmp - \mp  \big) 
		\end{pmatrix}
		\stackrel{d}{\to} \mathcal{N} \big( 0, I_2 \big).
	\end{align*}
	As 	$\begin{pmatrix} \frac{1}{\hmp}  &  - \frac{\sigma}{\mp \hmp} \end{pmatrix}  \toP	\begin{pmatrix} \frac{1}{\mp}  &  - \frac{\sigma}{\mp \mp} \end{pmatrix} = \Lambda^{+}$, the continuous mapping theorem implies that
	\begin{align*}
		\big( \Lambda^{+} J_{h^+} \Omega J_{h^+}^\top  (\Lambda^{+})^\top \big)^{-1/2} \sqrt{n}  \big( \myCn - \myC \big)
		\stackrel{d}{\to} \mathcal{N} \big( 0, 1 \big).
	\end{align*}
	
	Second, for the case $p = q$, we define the random vector
	$\mathbf{V}_n = (V_{1n},V_{2n},V_{3n},V_{4n})^\top := \sqrt{n} \big( \p - p, \,  \q - q, \, \p \q - p q, \,  \hsigma - \sigma \big)$ such that 
	$\mathbf{V}_n =  \sqrt{n} ( f(\p, \q, \r) - f(p, q, r))$ with 
	$f(\pa, \qa, \ra) = \big(\pa, \, \qa, \, \pa \qa, \, \ra - \pa \qa \big)$, which has Jacobian $J_f$ at point $(p,q,r)$, given in Proposition \ref{prop:Asymptotics-C}.
	The delta method then implies that $\mV_n \tod \mV := J_{f} \Omega^{1/2} \mZ$, where $\mZ \sim \mathcal{N}(0, I_3)$.
	%	\begin{align}
		%		\label{eqn:VnConvergence}
		%	 	\mathbf{V}_n \tod \mathbf{V} := J_{f} \Omega^{1/2} \mathbf{Z}, 
		%		\qquad  \mathbf{Z} \sim \mathcal{N}(0, I_3),
		%		\qquad J_{f}(p,q,r) :=
		%		\begin{pmatrix}
			%			1 & 0 & 0 \\
			%			0 & 1 & 0  \\
			%			q & p & 0 \\
			%			-q & -p & 1
			%		\end{pmatrix}.
		%	\end{align}
	Notice that the limiting distribution $\mV$ of $\mV_n$ is degenerate due to colinearities in $f$.
	
	As $p = q$, we further get that
	\begin{align*}
		\sqrt{n} \big( \hmp - \mp \big)  
		&= \mathds{1} \{\p \le \q \} \sqrt{n} \big( \p (1 - \q)  - p (1-q) \big)  + \mathds{1} \{ \p > \q \} \sqrt{n} \big( ( 1-  \p) \q  - (1-p) q \big)  \\
		&= \mathds{1} \{ \sqrt{n} (\p - p) \le \sqrt{n} (\q  - q) \} \Big[ \sqrt{n} (\p - p) - \sqrt{n} (\p \q   - p q)  \Big]  \\
		& + \mathds{1} \{ \sqrt{n} (\p - p) > \sqrt{n} ( \q  - q) \} \Big[ \sqrt{n} (\q - q) - \sqrt{n} ( \p  \q   - p q)  \Big]  \\
		& = \mathds{1} \{ V_{1n} \le V_{2n} \} \Big[ V_{1n} - V_{3n} \Big]  
		+ \mathds{1} \{ V_{1n} > V_{2n} \} \Big[ V_{2n} - V_{3n} \Big].
		%		&\tod  \mathds{1} \{ Z_{1} \le Z_{2} \} \Big[ Z_{1} - Z_{3} \Big]  
		%		+ \mathds{1} \{ Z_{1} > Z_{2} \} \Big[ Z_{2} - Z_{3} \Big],
	\end{align*}
	Thus, by the continuous mapping theorem (for convergence in distribution)
	\begin{align*}
		\begin{pmatrix}
			\sqrt{n} \big( \hsigma - \sigma \big) \\
			\sqrt{n} \big( \hmp - \mp \big) 
		\end{pmatrix}
		&=
		\begin{pmatrix}
			V_{4n} \\ 
			\mathds{1} \{ V_{1n} \le V_{2n} \} \Big[ V_{1n} - V_{3n} \Big]  
			+ \mathds{1} \{ V_{1n} > V_{2n} \} \Big[ V_{2n} - V_{3n} \Big]
		\end{pmatrix} \\
		&\tod
		\begin{pmatrix}
			V_{4} \\ 
			\mathds{1} \{ V_{1} \le V_{2} \} \Big[ V_{1} - V_{3} \Big]  
			+ \mathds{1} \{ V_{1} > V_{2} \} \Big[ V_{2} - V_{3} \Big]
		\end{pmatrix} 
	\end{align*}
	and hence,
	\begin{align*}
		\sqrt{n} \big(\myCn - \myC \big) &\tod	\Lambda^+
		\begin{pmatrix}
			V_{4} \\ 
			\mathds{1} \{ V_{1} \le V_{2} \} \Big[ V_{1} - V_{3} \Big]  
			+ \mathds{1} \{ V_{1} > V_{2} \} \Big[ V_{2} - V_{3} \Big]
		\end{pmatrix} \\
		&=	\Lambda^+
		\begin{pmatrix}
			V_{4} \\ 
			\big[ V_{1} - V_{3} \big]  - \mathds{1} \{ V_{1} > V_{2} \} \big[ V_{1} - V_{2} \big]
		\end{pmatrix}.
	\end{align*}

	\vspace{0.3cm}
	\noindent
	$\bullet$
	\textbf{In the negative case}, let $\sigma = \Cov(A,B) < 0$ such that $\myC < 0$. 
	With arguments as above,
	\begin{align*}
		\sqrt{n} \big(\myCn - \myC \big)
		= \sqrt{n}  \frac{\hsigma - \sigma}{\hmn} + \sqrt{n} \left( \frac{\sigma ( \mn - \hmn )}{\mn \hmn} \right)  - \sqrt{n} \frac{ \hsigma}{\hmn}  \mathds{1}\{\hsigma \ge 0 \} +   \sqrt{n} \frac{\hsigma}{\hmp} \mathds{1}\{ \hsigma > 0 \}, 
	\end{align*}
	and the last two terms again converge to zero in probability.
	Hence,
	\begin{align*}
		\sqrt{n} \big(\myCn - \myC \big)
		&= \sqrt{n}  \frac{\hsigma - \sigma}{\hmn} + \sqrt{n} \left( \frac{\sigma ( \mn - \hmn )}{\mn \hmn} \right) + o_\mathbb{P}(1) \\
		&= 
		\begin{pmatrix}
			\frac{1}{\hmn}  & 
			- \frac{\sigma}{\mn \hmn}
		\end{pmatrix} 
		\times
		\begin{pmatrix}
			\sqrt{n} \big( \hsigma - \sigma \big) \\
			\sqrt{n} \big( \hmn - \mn \big) 
		\end{pmatrix} + o_\mathbb{P}(1).
	\end{align*}
	
	We first consider the case $p \not= 1 - q$.
	Then,
	\begin{align*}
		\begin{pmatrix}
			\sqrt{n} \big( \hsigma - \sigma \big) \\
			\sqrt{n} \big( \hmn - \mn \big) 
		\end{pmatrix}
		&= \sqrt{n} \big( h^-(\p, \q, \r) - h^-(p, q, r) \big), \qquad \text{with} \\
		h^-(\pa,\qa,\ra) &= 
		\begin{pmatrix}
			\ra - \pa\qa \\
			\pa\qa - \max(0, \pa + \qa -1) 
		\end{pmatrix}
	\end{align*}
%	with $h^-(p,q,r) := 
%	\begin{pmatrix}
%		r - pq \\
%		pq - \max(0, p + q -1) 
%	\end{pmatrix}$,
	which is continuously differentiable at any point $p \not= 1 - q$ with Jacobian matrix $J_{h^-}$ (at point $(p,q,r)$).
	%	\begin{align*}
		%		J_{h^-}(p,q,r) 
		%%		&=
		%%		\begin{pmatrix}
			%%			-q  & -p & 1 \\
			%%			q \mathds{1}(p+q<1) + (q-1) \mathds{1}(p+q>1)   & 	p \mathds{1}(p+q<1) + (p-1)  \mathds{1}(p+q>1)  & 0 
			%%		\end{pmatrix} \\
		%		&=
		%		\begin{pmatrix}
			%			-q  & -p & 1 \\
			%			q -	\mathds{1}(p>1-q)   & 	p -	\mathds{1}(p>1-q) & 0 
			%		\end{pmatrix}.
		%	\end{align*}
	%
	%	Defining $\Lambda^- =  
	%	\begin{pmatrix}
		%		\frac{1}{m^-}  & 
		%		- \frac{\sigma}{m^-  m^-}
		%	\end{pmatrix}$, 
	As above, by the delta method and the continuous mapping theorem, we get that
	\begin{align*}
		\big( \Lambda^{-} J_{h^-} \Omega J_{h^-}^\top  (\Lambda^{-})^\top \big)^{-1/2} \sqrt{n} 
		\big( \myCn - \myC \big)
		\stackrel{d}{\to} \mathcal{N} \big( 0, 1 \big).
	\end{align*}
	
	Second, for the case $p = 1 - q$, we again use the random vector
	$\mV_n = \sqrt{n} \big( \p -p, \, \q - q, \, \p \q - p q, \,  \hsigma - \sigma \big)$ together with its (degenerate) limiting distribution $\mV = J_{f} \Omega^{1/2} \mZ$, where $\mZ \sim \mathcal{N}(0, I_3)$.
	We then get that
	\begin{align*}
		\sqrt{n} \big( \hmn - \mn \big)  
		&= \mathds{1} \{ \p \le 1- \q \}  \sqrt{n} \big( \p \q   - p q \big) \\
		&+  \mathds{1} \{ \p > 1- \q \}  \sqrt{n} \big( \p \q  - \p +(1- \q ) - p q + p - (1-q) \big)  \\
		&= \mathds{1} \{ \sqrt{n} (\p - p) \le - \sqrt{n} ( \q - q) \} \Big[  \sqrt{n} ( \p \q   - p q)  \Big]  \\
		& + \mathds{1} \{ \sqrt{n} (\p - p) > - \sqrt{n} ( \q - q) \} \Big[ \sqrt{n} ( \p \q   - p q) - \sqrt{n} (\p - p) - \sqrt{n} ( \q  - q)  \Big]  \\
		& = \mathds{1} \{ V_{1n} \le - V_{2n} \} V_{3n}
		+ \mathds{1} \{ V_{1n} > - V_{2n} \} \Big[ V_{3n} - V_{1n} - V_{2n} \Big].
	\end{align*}
	Thus, similar as above in the positive case, we get  that
	%	\begin{align*}
		%		\begin{pmatrix}
			%			\sqrt{n} \big( \hat \sigma - \sigma \big) \\
			%			\sqrt{n} \big( \hat m^- - m^- \big) 
			%		\end{pmatrix}
		%		&=
		%		\begin{pmatrix}
			%			V_{4n} \\ 
			%			\mathds{1} \{ V_{1n} \le - V_{2n} \} V_{3n}
			%			+ \mathds{1} \{ V_{1n} > - V_{2n} \} \Big[ V_{3n} - V_{1n} - V_{2n} \Big]
			%		\end{pmatrix} \\
		%		&\tod
		%		\begin{pmatrix}
			%			V_{4} \\ 
			%			\mathds{1} \{ V_{1} \le - V_{2} \} V_{3}
			%			+ \mathds{1} \{ V_{1} > - V_{2} \} \Big[ V_{3} - V_{1} - V_{2} \Big]
			%		\end{pmatrix} 
		%	\end{align*}
	%	and
	\begin{align*}
		\sqrt{n} \big(\myCn - \myC \big) &\tod	\Lambda^- 
		\begin{pmatrix}
			V_{4} \\ 
			\mathds{1} \{ V_{1} \le - V_{2} \} V_{3}
			+ \mathds{1} \{ V_{1} > - V_{2} \} \Big[ V_{3} - V_{1} - V_{2} \Big]
		\end{pmatrix} \\
		&= \Lambda^-  
		\begin{pmatrix}
			V_{4} \\ 	
			V_{3} - \mathds{1} \{ V_{1} > - V_{2} \} \big[  V_{1} + V_{2} \big]
		\end{pmatrix},
	\end{align*}
	which concludes this proof.
	
	%	\vspace{2cm}
	%	\color{purple}
	%	THE PURPLE PART IS INCORRECT!!!!! THE PROBLEM IS THAT THE INDICATOR FUNCTIONS DEPEND ON THE ESTIMATED QUANTITIES, AND ARE HENCE TO BE INCLUDED IN THE ASYMPTOTIC DISTRIBUTION OF THE LATTER TERMS
	%	
	%	If $p = 1-q$, then 
	%	\begin{align*}
		%		\begin{pmatrix}
			%			\sqrt{n} \big( \hat \sigma - \sigma \big) \\
			%			\sqrt{n} \big( \hat m^- - m^- \big) 
			%		\end{pmatrix}
		%		&= \mathds{1} \{\hat p \le 1- \hat q \} \sqrt{n}  
		%		\begin{pmatrix}
			%			\sqrt{n} \big( \hat \sigma - \sigma \big) \\
			%			\sqrt{n} \big( \hat p \hat q - p q \big) 
			%		\end{pmatrix} \\
		%		&+ \mathds{1} \{\hat p > 1- \hat q \} \sqrt{n}  
		%		\begin{pmatrix}
			%			\sqrt{n} \big( \hat \sigma - \sigma \big) \\
			%			\sqrt{n} \big( (\hat p \hat q  - \hat p - \hat q + 1) - (p q - p - q + 1) \big) 
			%		\end{pmatrix},
		%	\end{align*}
	%	where we can apply the Delta-method to the two terms individually with Jacobian matrices
	%	\begin{align*}
		%		J_{h_{1}^-}(p,q,r) &=
		%		\begin{pmatrix}
			%			-q  & -p & 1 \\
			%			q  & 	p & 0 
			%		\end{pmatrix} 
		%	 	\qquad \text{and} \qquad  
		%		J_{h_{2}^-}(p,q,r) =
		%		\begin{pmatrix}
			%			-q  & -p & 1 \\
			%			q-1 & p-1 & 0
			%		\end{pmatrix}.
		%	\end{align*}
	%	\color{black}	
	
\end{proof}

\begin{proof}[Proof of Proposition~\ref{prop:Asymptotics-Phi}]
    With the same arguments as in \eqref{eqn:WellDefinedProbOne} in the beginning of the proof of Proposition~\ref{prop:Asymptotics-Q}, we get that $\widehat{\phi}_n$ is well-defined with asymptotic probability one.
	Note that we can write
	\begin{align*}
		\widehat{\phi}_n = l \big( \p, \q, \r \big), \quad \text{with} \quad
		l(\pa,\qa,\ra) := \frac{\ra - \pa\qa}{\sqrt{\pa(1-\pa)\qa(1-\qa)}}.
	\end{align*}
	The function $l$ is differentiable at $(p, q, r)$ satisfying Assumption \ref{ass:BoundedAway} and
	has the Jacobian matrix 
	\begin{align*}
		J_l(\pa,\qa,\ra) = \frac{1}{\sqrt{\pa(1-\pa)\qa(1-\qa)}} \left( 		
		\begin{pmatrix}
			-\qa \\
			-\pa \\
			1
		\end{pmatrix}^\top 
		+ \frac{\pa\qa-\ra}{2\pa(1-\pa)\qa(1-\qa)}	
		\begin{pmatrix}
			(1-\pa)\qa(1-\qa) - \pa\qa(1-\qa)\\
			\pa(1-\pa)(1-\qa) - \pa(1-\pa)\qa \\
			0
		\end{pmatrix}^\top 
		\right).
	\end{align*}
	Furthermore, $J_l(p,q,r) \not= 0$ as its third dimension is obviously non-zero for all $p,q$ satisfying Assumption \ref{ass:BoundedAway}.
	Hence, by the delta method \cite[Theorem 3.1]{VanderVaart2000} and  Lemma \ref{lemma:CLTjoint} the claim follows as in the proof of Proposition~\ref{prop:Asymptotics-Q}.
%	\todo[inline]{Timo: 
%		I think we should illustrate that $J_l$ is non-zero at this point to argue that the asymptotic normal distribution is non-degenerate! 
%		Otherwise, we cannot compute the inverse $\big(J_l \Omega J_l^\top \big)^{-1/2}$ in Proposition \ref{prop:Asymptotics-Phi}!}
\end{proof}

\subsection{Additional Lemmas}
\label{sec:AddLemmas}

The following result provides a central limit theorem for the sequence $\mathbf{W}_i = (X_i, Y_i, X_i Y_i)^\top$ that is invoked in the proofs of Propositions \ref{prop:Asymptotics-Q}, \ref{prop:Asymptotics-C}, \ref{prop:Asymptotics-Phi} and \ref{prop:Asymptotics-Z_Q}.

\begin{lemma} 
	\label{lemma:CLTjoint}
	% Given Assumptions \ref{ass:BoundedAway}, \ref{ass:Dependence} and \ref{ass:Multicolinearity}, it holds that
	Let Assumption \ref{ass:BoundedAway} and either Assumption \ref{ass:iid} or Assumption \ref{ass:Dependence} hold.
	Then, given $\Omega_n = \Var \left( n^{-1/2} \sum_{i=1}^n \mathbf{W}_i \right)$ and $\Omega = \lim_{n\to \infty} \Omega_n$,  it holds that
	\begin{align*}
		\sqrt{n} \, \Omega_n^{-1/2}
		\begin{pmatrix}
			\p - p \\ \q - q \\ \r  - r
		\end{pmatrix}
		\stackrel{d}{\to}
		\mathcal{N} \big(0, I_3 \big),
		\qquad \text{and} \qquad
		\sqrt{n} \, \Omega^{-1/2}
		\begin{pmatrix}
			\p - p \\ \q - q \\ \r  - r
		\end{pmatrix}
		\stackrel{d}{\to}
		\mathcal{N} \big(0, I_3 \big).
		%		\qquad
		%		\Omega_n = \Var \left( n^{-1/2} \sum_{i=1}^n \mathbf{W}_i^\top \right).
	\end{align*}
\end{lemma}

\begin{proof}
	% Given the non-degeneracy conditions in Assumption~\ref{ass:Multicolinearity}, 
	Given the non-degeneracy condition in Assumption \ref{ass:Dependence}, for any $\lambda \in \mathbb{R}^3$ with $\Vert \lambda \Vert = 1$, it holds that
	$\bar \sigma_n^2 := \lambda^\top \Var \left( n^{-1/2} \sum_{i=1}^n  \mathbf{W}_i  \right) \lambda  \to \bar \sigma^2 > 0$.~
	Hence, using the dependence condition in Assumption~\ref{ass:Dependence}, we can apply Theorem~5.16 in \cite{White2001} together with the Cramér-Wold theorem, and the claim follows.
	
	The claim under Assumption \ref{ass:iid} follows as independence implies the mixingale condition and as under Assumption \ref{ass:BoundedAway}, the components of $\mathbf{W}_i$ are not collinear almost surely, such that the covariance matrix  
	\begin{equation} 
		\label{eq:variance_matrix_iid}
		\Omega = \Var \left( \mathbf{W}_i  \right) =
		\begin{pmatrix}
			p (1-p) & r - p q & r (1-p) \\
			r - p q & q (1-q) & r (1-q) \\
			r (1-p) & r (1-q) & r (1-r)
		\end{pmatrix}
	\end{equation}
	is strictly positive definite.	
\end{proof}

For the test for the sign of $\sigma$ involved in the construction of confidence intervals for $\myC$ as explained in Subsection \ref{subsec:CIs} and in the proof of Proposition \ref{prop:Asymptotics-C}, we need a central limit theorem for the empirical covariance $\widehat{\sigma}_n$, which we present here.

\begin{lemma} 
	\label{lemma:CLT_sigma}
	Under Assumption \ref{ass:BoundedAway} and either Assumption \ref{ass:iid} or Assumption \ref{ass:Dependence}, it holds that
	$$\big( \Delta^\top \Omega \Delta \big)^{-1/2} \sqrt{n} (\hsigma - \sigma) \tod \mathcal{N}(0,1) 
	\qquad \text{ with } \qquad
	\Delta = (-q, -p, 1)^\top.$$
\end{lemma}

\begin{proof}[Proof of Lemma \ref{lemma:CLT_sigma}]
	Note that $\sigma = r - pq$ is a differentiable function of $p,q$ and $r$ (as the three partial derivatives are continuous) with gradient $ \Delta^\top = (-q, -p, 1)$.  Thus, the claim follows by using Lemma \ref{lemma:CLTjoint} and the delta method \citep[Theorem 3.1]{VanderVaart2000}.
\end{proof}

\begin{lemma}
	\label{lem:denominator}
	The expression $\pa(\qa-2\ra)+\ra(1-2\qa+2\ra)$ is non-zero in $D:= \big\{ (\pa,\qa,\ra)\in (0,1)^3 \mid \max(0,\pa+\qa-1) < \ra < \min(\pa,\qa) \big\}$.
\end{lemma}

\begin{figure}[tb]
	\centering
	\begin{tikzpicture}
		
		\draw[step=1cm, gray, very thin] (0,0) grid (10,10);
		
		\draw (0,0) rectangle (10,10);
		
		\fill[blue, opacity=0.7] (0,0) rectangle (5,5);
		\fill[orange, opacity=0.7] (0,5) rectangle (5,10);
		\fill[yellow, opacity=0.7] (5,0) rectangle (10,5);
		\fill[cyan, opacity=0.7] (5,5) rectangle (10,10);
		
		\fill[red] (5, 5) circle (5cm);
		\draw[red!50!black, line width=1mm] (5, 5) circle (5cm);
		
		\draw[very thick, -] (5,10) -- (10,5) node[midway, sloped, above] {$\pa+\qa\ge \frac{3}{2}$};
		\draw[very thick, -] (0,5) -- (5,10) node[midway, sloped, above] {$\qa\ge \frac{1}{2}+\pa$};
		\draw[very thick, -] (5,0) -- (0,5) node[midway, sloped, below] {$\pa+\qa\le \frac{1}{2}$};
		\draw[very thick, -] (5,0) -- (10,5) node[midway, sloped, below] {$\qa\le \pa-\frac{1}{2}$};

		\draw[very thick, ->] (0,0) -- (10,0) node[anchor=north west] {\textbf{p axis}};
		\draw[very thick, ->] (0,0) -- (0,10) node[anchor=south east] {\textbf{q axis}};
		
		\foreach \p in {1,2,3,4,5,6,7,8,9}
		\draw[thin] (\p cm,1pt) -- (\p cm,-1pt) node[anchor=north] {$\boldsymbol{0.\p}$};
		\foreach \q in {1,2,3,4,5,6,7,8,9}
		\draw[thin] (1pt,\q cm) -- (-1pt,\q cm) node[anchor=east] {$\boldsymbol{0.\q}$};
		
	\end{tikzpicture}
	\caption{Graphical illustration for the proof of Lemma \ref{lem:denominator}. %roots of $p(q-2r)+r(1-2q+2r)$: 
		The interior of the red disk contains complex roots (in $\ra$) and the real roots in the remaining four colored areas are treated with a case distinction, showing that in each case, the roots violate a condition of $\mathbb{D}$: 
		The roots in the blue and cyan areas are shown to be smaller than $\max(0, \pa+\qa-1)$ and the roots in the orange and yellow regions are shown to be larger than $\min(\pa,\qa)$.}
	\label{fig:proofsupport}
\end{figure}
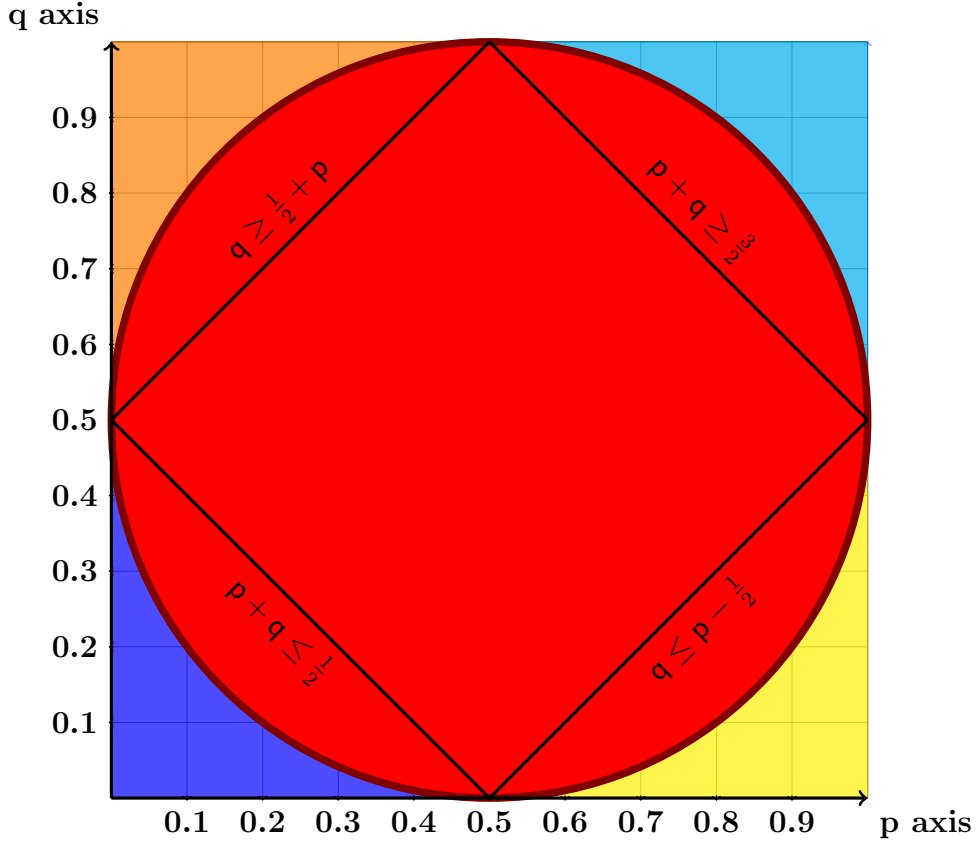

\begin{proof}[Proof of Lemma \ref{lem:denominator}]
	Straightforward calculations show that the function $\ra \mapsto \pa(\qa-2\ra)+\ra(1-2\qa+2\ra)$ has the two (possibly complex) roots 
	\begin{align*}
		\ra_1
		%		=&\ \frac{\sqrt{4q^2-4q+4p^2-4p+1}+2p+2q-1}{4}\\
		%		=&\ \frac{\sqrt{4q(q-1)+4p(p-1)+1}}{4}+\frac{p}{2}+\frac{q}{2}-\frac{1}{4}\\
		=&\ \frac{\sqrt{\qa(\qa-1)+\pa(\pa-1)+\frac{1}{4}}}{2}+\frac{\pa}{2}+\frac{\qa}{2}-\frac{1}{4},
	\end{align*}
	and 
	\begin{align*}
		\ra_2
		%		&\ -\frac{\sqrt{4q^2-4q+4p^2-4p+1}+2p+2q-1}{4}\\
		%		=&\ -\frac{\sqrt{4q(q-1)+4p(p-1)+1}}{4}+\frac{p}{2}+\frac{q}{2}-\frac{1}{4}\\
		=&\ -\frac{\sqrt{\qa(\qa-1)+\pa(\pa-1)+\frac{1}{4}}}{2}+\frac{\pa}{2}+\frac{\qa}{2}-\frac{1}{4}.
	\end{align*}
	%	It holds that
	%	\begin{align*}
		%		r_1=&\ \frac{\sqrt{4q^2-4q+4p^2-4p+1}+2p+2q-1}{4}\\
		%		=&\ \frac{\sqrt{4q(q-1)+4p(p-1)+1}}{4}+\frac{p}{2}+\frac{q}{2}-\frac{1}{4}\\
		%		=&\ \frac{\sqrt{q(q-1)+p(p-1)+\frac{1}{4}}}{2}+\frac{p}{2}+\frac{q}{2}-\frac{1}{4}
		%	\end{align*}
	%	and analogously,
	%	\begin{align*}
		%		r_2 
		%%		&\ -\frac{\sqrt{4q^2-4q+4p^2-4p+1}+2p+2q-1}{4}\\
		%%		=&\ -\frac{\sqrt{4q(q-1)+4p(p-1)+1}}{4}+\frac{p}{2}+\frac{q}{2}-\frac{1}{4}\\
		%		=&\ -\frac{\sqrt{q(q-1)+p(p-1)+\frac{1}{4}}}{2}+\frac{p}{2}+\frac{q}{2}-\frac{1}{4}.
		%	\end{align*}
	
	For $\ra_1$ and $\ra_2$, we get complex solutions with a non-zero imaginary part for any $\qa(\qa-1)+\pa(\pa-1)+\frac{1}{4}<0$, which is equivalent to $(\pa,\qa) \in U := \{x\in \mathbb{R}^2:\|x-(0.5, 0.5) \|_2 <  0.5\}$. 
	Here, $U$ is the open disk centered around $(0.5, 0.5)$ with a radius of $0.5$, shown in red in Figure \ref{fig:proofsupport}. 
	We do not have to consider these complex solutions as they are not in $\mathbb{D}$.
	%	Since $\mathbb{R}^2\times\mathbb{C}\backslash \mathbb{R}\cap \mathbb{D}=\emptyset$, there is nothing to be shown.
	We further define the closure of the open disk $U$ as $\overline{U}$, and its border as $\partial U := \overline{U} \backslash U$.
	
	For $(\pa,\qa) \in [0,1]^2\backslash U$, we distinguish the following four cases that correspond to the four colored areas in Figure  \ref{fig:proofsupport}.
	In each case, we show that the possible roots $\ra_1$ and $\ra_2$ are not in $\mathbb{D}$.
	\medskip
	
	\noindent $\bullet$ Case 1: Let $(\pa,\qa) \in [0, 0.5]^2\backslash U$ (dark blue area in Figure \ref{fig:proofsupport}). 
	For $\pa=0$, i.e., for any point on the left edge of the blue area, we get 
	\begin{align*}
		\ra_1=\frac{\sqrt{\qa(\qa-1)+\frac{1}{4}}}{2}+\frac{\qa}{2}-\frac{1}{4}=\frac{\sqrt{(\qa-\frac{1}{2})^2}}{2}+\frac{\qa}{2}-\frac{1}{4}= \frac{\frac{1}{2}-\qa}{2}+\frac{\qa}{2}-\frac{1}{4}=0.
	\end{align*}
	Furthermore, for any $(\pa,\qa) \in [0, 0.5]^2\backslash \overline{U}$, it holds that 
	\begin{align*}
		\frac{\partial \ra_1}{\partial \pa}=\frac{2\left(\pa+\sqrt{\qa(\qa-1)+\pa(\pa-1)+\frac{1}{4}}\right)-1}{4\sqrt{\qa(\qa-1)+\pa(\pa-1)+\frac{1}{4}}}\le 0,
	\end{align*}
	since $\pa+\sqrt{\qa(\qa-1)+\pa(\pa-1)+\frac{1}{4}} \le  \pa+\sqrt{\pa(\pa-1)+\frac{1}{4}} = \pa+|\pa-\frac{1}{2}|=\pa-\pa+\frac{1}{2}=\frac{1}{2}$. 
	Hence, when moving from the left edge where $\pa=0$ towards the right in the blue area, $\ra_1$ decreases (weakly) from a starting value of zero such that we can conclude that $\ra_1 \le 0$ for all $(\pa,\qa) \in [0, 0.5]^2\backslash \overline{U}$.
	
	We treat the boundary case $(\pa,\qa) \in [0, 0.5]^2\cap \partial U$ separately as the above partial derivative is not defined there: Here, it holds that $\qa(\qa-1)+\pa(\pa-1)+\frac{1}{4} = 0$ such that $\ra_1 = \frac{\pa+\qa}{2}-\frac{1}{4}\le 0$.
	% with equality only for $(p,q)=(0,\frac{1}{2})$ and $(p,q)=(\frac{1}{2},0)$.
	% For all other $p-q$-combinations, the sum $p+q$ is smaller than 0.5 and the inequality holds strictly.
	
	Hence, for any $(\pa,\qa) \in [0, 0.5]^2 \backslash U$, we conclude that $\ra_1 \le 0 = \max(0,\pa+\qa-1)$.
	As $\ra_2 \le \ra_1$ (for any real-valued solution in $\mathbb{D}$) the same claim follows directly for $\ra_2$.
	\medskip 
	
	\noindent $\bullet$ 
	Case 2: Let $(\pa,\qa) \in \big( [0, 0.5] \times [0.5, 1] \big) \backslash U$ (orange area in Figure \ref{fig:proofsupport}). For $\pa=0$, we get 
	\begin{align*}
		\ra_2 = -\frac{\sqrt{\qa(\qa-1)+\frac{1}{4}}}{2}+\frac{\qa}{2}-\frac{1}{4} = -\frac{\sqrt{(\qa-\frac{1}{2})^2}}{2}+\frac{\qa}{2}-\frac{1}{4}= \frac{\frac{1}{2}-\qa}{2}+\frac{\qa}{2}-\frac{1}{4} = 0 = \pa =  \min(\pa,\qa).
	\end{align*}
	Furthermore, for any $(\pa,\qa) \in \big( [0, 0.5] \times [0.5, 1] \big) \backslash \overline{U}$, it holds that 
	\begin{align}
		\label{eqn:Derivative_r2_largerone}
		\frac{\partial \ra_2}{\partial \pa}=\frac{1-2\pa}{4\sqrt{\pa(\pa-1)+\qa(\qa-1)+\frac{1}{4}}}+\frac{1}{2}\ge 1
	\end{align}
	because $$\frac{1-2\pa}{4\sqrt{\pa(\pa-1)+\qa(\qa-1)+\frac{1}{4}}}-\frac{1}{2}=\frac{1-2\left(\pa+\sqrt{\pa(\pa-1)+\qa(\qa-1)+\frac{1}{4}}\right)}{4\sqrt{\pa(\pa-1)+\qa(\qa-1)+\frac{1}{4}}}$$
	and $\pa+\sqrt{\qa(\qa-1)+\pa(\pa-1)+\frac{1}{4}}\le \pa+\sqrt{\pa(\pa-1)+\frac{1}{4}}=\pa+|\pa-\frac{1}{2}|=\pa-\pa+\frac{1}{2}=\frac{1}{2}$.  
	Hence, when moving from the left edge ($\pa=0$) towards the right in the orange area, $\ra_2$ grows faster than the upper bound $\pa=\min(\pa,\qa)$ in $\mathbb{D}$, such that we can conclude that $\ra_2 \ge \min(\pa,\qa)$ for all $(\pa,\qa) \in \big( [0, 0.5] \times [0.5, 1] \big) \backslash \overline{U}$.
	
	We again treat the boundary case $(\pa,\qa) \in \big( [0, 0.5] \times [0.5, 1] \big) \cap \partial U$ separately. Here, it holds that $\qa\ge\frac{1}{2}+\pa$ such that $\ra_2=\frac{\pa+\qa}{2}-\frac{1}{4}\ge \frac{\pa+\frac{1}{2}+\pa}{2}-\frac{1}{4}=\pa=\min(\pa,\qa)$.
	
	Hence, for any $(\pa,\qa) \in \big( [0, 0.5] \times [0.5, 1] \big) \backslash U$, we conclude that $\ra_2 \ge \min(\pa,\qa)$.
	As $\ra_1 \ge \ra_2$ (for any real-valued solutions in $\mathbb{D}$) the same claim holds for $\ra_1$.
	\medskip
	
	\noindent $\bullet$  
	Case 3: Let $(\pa,\qa) \in \big( [0.5,1] \times [0,0.5] \big) \backslash U$ (yellow area in Figure \ref{fig:proofsupport}).
	The argument in this case is equivalent to Case 2 by simply interchanging $\pa$ and $\qa$.
	
	%	For $q=0$, we get 
	%	\begin{align*}
		%		r_2 = -\frac{\sqrt{p(p-1)+\frac{1}{4}}}{2}+\frac{p}{2}-\frac{1}{4} = -\frac{\sqrt{(p-\frac{1}{2})^2}}{2}+\frac{p}{2}-\frac{1}{4} = -\frac{p-\frac{1}{2}}{2}+\frac{p}{2}-\frac{1}{4} =0=q= \min\{p,q\}.
		%	\end{align*}
	%	
	%	Furthermore, for $(p,q) \in \big( [0.5,1] \times [0,0.5] \big) \backslash \overline{U}$, it holds that 
	%	\begin{align*}
		%		\frac{\partial r_2}{\partial q}=\frac{1-2q}{4\sqrt{p(p-1)+q(q-1)+\frac{1}{4}}}+\frac{1}{2}\ge 1
		%	\end{align*}
	%	with the same reasoning as in and below \eqref{eqn:Derivative_r2_largerone} by interchanging $p$ and $q$.
	%%	 $$\frac{1-2q}{4\sqrt{p(p-1)+q(q-1)+\frac{1}{4}}}-\frac{1}{2}=\frac{1-2\left(q+\sqrt{p(p-1)+q(q-1)+\frac{1}{4}}\right)}{4\sqrt{p(p-1)+q(q-1)+\frac{1}{4}}}$$
	%%	and $q+\sqrt{q(q-1)+p(p-1)+\frac{1}{4}}\le q+\sqrt{q(q-1)+\frac{1}{4}}=q+|q-\frac{1}{2}|=q-q+\frac{1}{2}=\frac{1}{2}$.  
	%	Hence, $r_2$ again grows faster than the upper bound $q = \min(p,q)$ such that we can conclude $r_2 \ge \min(p,q)$.
	%	
	%	For the boundary $(p,q)\in \big( [0.5,1] \times [0,0.5] \big) \cap \partial U$, it holds that $q \le -\frac{1}{2}+p$ such that $r_2 = \frac{p+q-\frac{1}{2}}{2}\ge \frac{2q}{2} = q = \min(p,q)$.
	%	
	%		
	%	Hence, for any $(p,q) \in \big( [0.5,1] \times [0, 0.5] \big) \backslash U$, we conclude that $r_2 \ge \min(p,q)$.
	%	As $r_1 \ge r_2$ (for any real-valued solutions in $\mathbb{D}$) the same claim holds for $r_1$.
	%	\color{black}
	%	\medskip
	
	\medskip 
	\noindent $\bullet$  
	Case 4: $(\pa,\qa) \in [0.5, 1]^2 \backslash U$ (cyan area in Figure \ref{fig:proofsupport}). 
	For the right edge where $\pa=1$, we get 
	\begin{align*}
		\ra_1=\frac{\sqrt{\qa(\qa-1)+\frac{1}{4}}}{2}+\frac{\qa}{2}+\frac{1}{4}=\frac{\sqrt{(\qa-\frac{1}{2})^2}}{2}+\frac{\qa}{2}+\frac{1}{4}=\frac{\qa-\frac{1}{2}}{2}+\frac{\qa}{2}+\frac{1}{4}=\qa=\max(0, \pa+\qa-1).
	\end{align*}
	Furthermore, for $(\pa,\qa) \in [0.5,1]^2 \backslash \overline{U}$, it holds that 
	\begin{align*}
		\frac{\partial \ra_1}{\partial \pa}=\frac{2\pa-1}{4\sqrt{\pa(\pa-1)+\qa(\qa-1)+\frac{1}{4}}}+\frac{1}{2}\ge 1,
	\end{align*}
	since, similar to \eqref{eqn:Derivative_r2_largerone}, $$\frac{2\pa-1}{4\sqrt{\pa(\pa-1)+\qa(\qa-1)+\frac{1}{4}}}-\frac{1}{2}=\frac{2\left(\pa-\sqrt{\qa(\qa-1)+\pa(\pa-1)+\frac{1}{4}}\right)-1}{4\sqrt{\qa(\qa-1)+\pa(\pa-1)+\frac{1}{4}}}$$
	and $\pa-\sqrt{\qa(\qa-1)+\pa(\pa-1)+\frac{1}{4}}\ge \pa-\sqrt{\pa(p-1)+\frac{1}{4}}=\pa-|\pa-\frac{1}{2}|=\pa-\pa+\frac{1}{2}=\frac{1}{2}$. 
	Hence, when moving from the right edge ($\pa=1$) towards the left in the cyan area, $\ra_1$ decreases faster than the lower bound $\max(0, \pa+\qa-1)$ in $\mathbb{D}$, such that we can conclude that $\ra_1 \le \max(0, \pa+\qa-1)$ for all $(\pa,\qa) \in [0.5, 1]^2 \backslash \overline{U}$.
	
	For the boundary $(\pa,\qa) \in [0.5,1]^2 \cap \partial U$, it holds that $\ra_1=\frac{\pa+\qa}{2}-\frac{1}{4}\le \pa+\qa-1 \le \max(0,\pa+\qa-1)$.
	
	Since it holds that $\ra_2\le \ra_1$, the same claim follows for $\ra_2$ directly.
	\bigskip
	
	Hence, in all four cases above, the possible roots $\ra_1$ and $\ra_2$ are not in $\mathbb{D}$ such that we can conclude that the expression $\pa(\qa-2\ra)+\ra(1-2\qa+2\ra)$ is non-zero for all $(\pa,\qa,\ra) \in \mathbb{D}$.
\end{proof}

\FloatBarrier
\section{Additional Tables and Graphs} 
\label{sec:additional_graphs}

\begin{table}[h!] 
		\centering
		\begin{tabular}{cccc}
			\toprule
			measure & value& 90\% Fisher-CI & 90\% Standard-CI   \\
			\midrule 
			Phi coefficient & 0.23 & $[0.16,0.30]$ & $[0.16,0.30]$ \\
			Yule's $\myQ$ & 0.86 & $[0.59,0.96]$ & $[0.70,1.02]$\\
			Cole's coefficient & 0.83 & $[0.44,0.96]$ & $[0.61,1.00]$ \\
			\bottomrule
		\end{tabular}
	\caption{Point estimators and confidence intervals with (Fisher-CI) and without (Standard-CI) the Fisher transformation  for the data from Table \ref{tab:smallpox_contingency_table}.}
	\label{tab:smallpox_CIs}
\end{table}

\begin{figure}[h!] 
	\centering
	\includegraphics[width=1\linewidth]{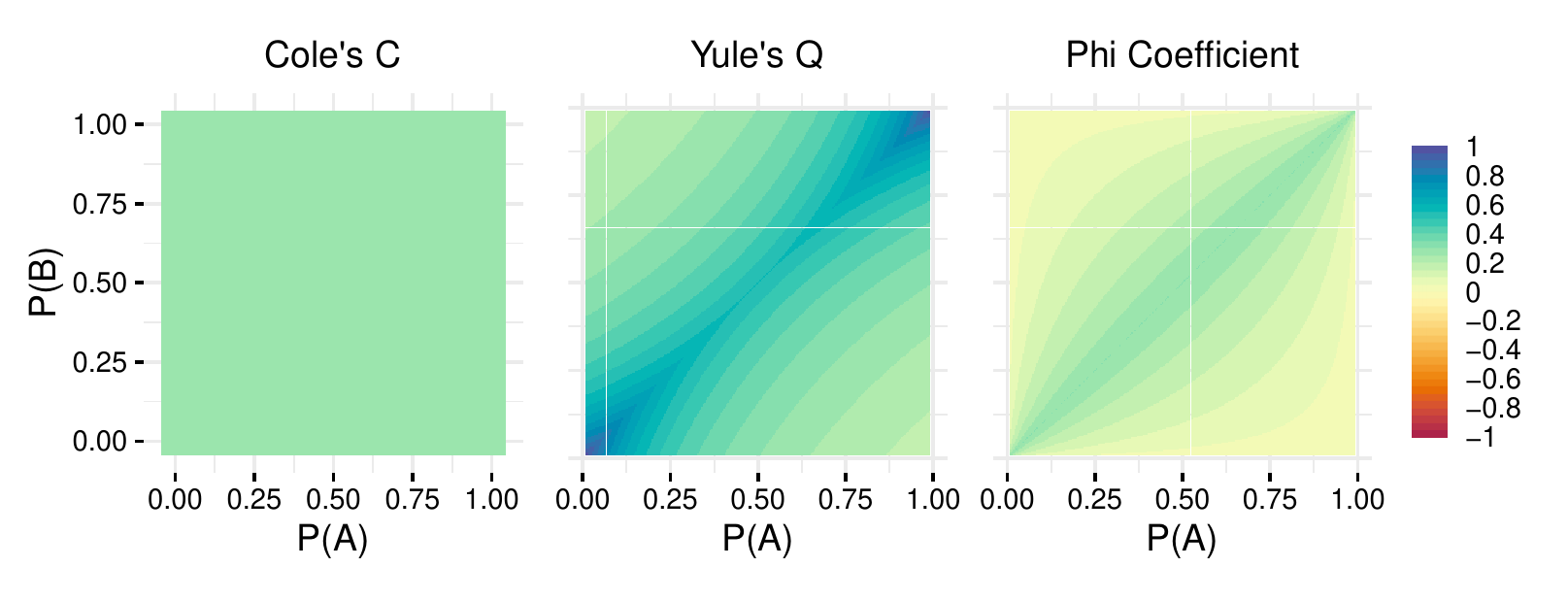}
	\caption{This figure depicts Cole's $\myC$ with a fixed value of 0.3 and the corresponding values of Yule's $\myQ$ and the phi coefficient for all combinations of marginal event probabilities $\P(A)$ and $\P(B)$.}
	\label{fig:colecomparison03}
\end{figure}

\begin{figure}[h!] 
	\centering
	\includegraphics[width=1\linewidth]{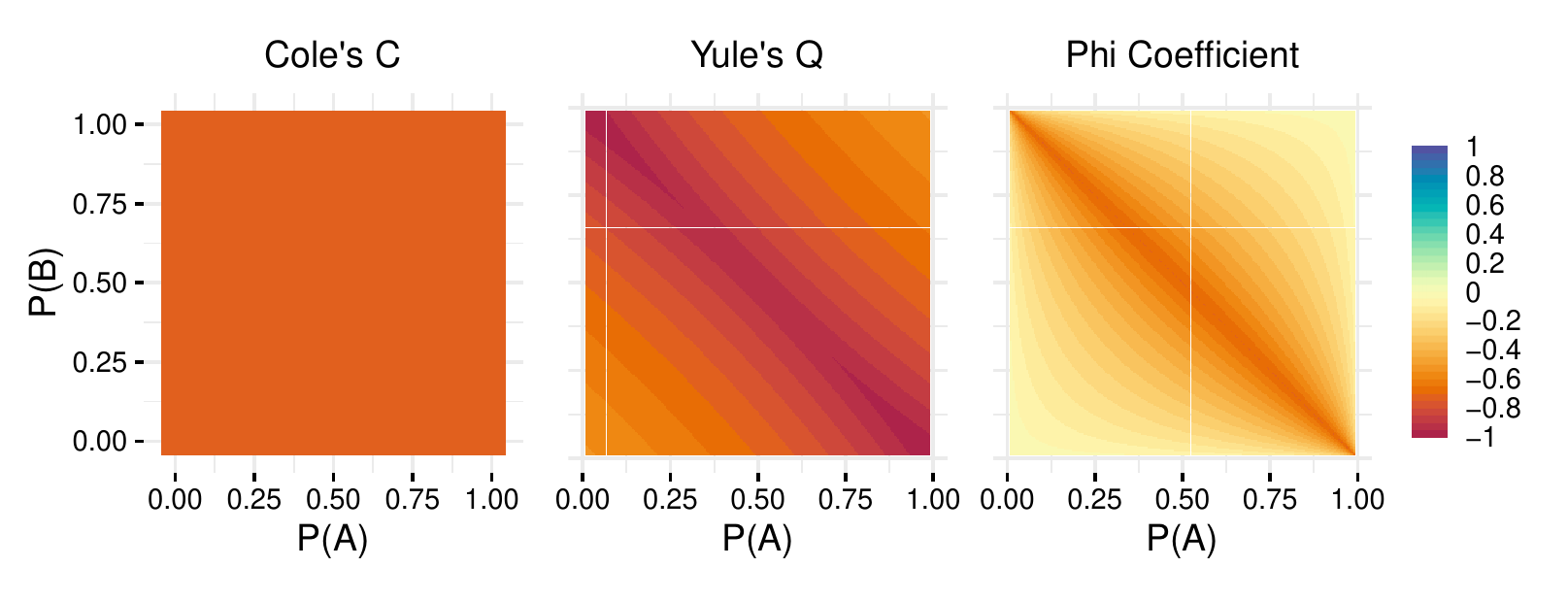}
	\caption{This figure depicts Cole's $\myC$ with a fixed value of -0.7 and the corresponding values of Yule's $\myQ$ and the phi coefficient for all combinations of marginal event probabilities $\P(A)$ and $\P(B)$.}
	\label{fig:colecomparison-07}
\end{figure}

\begin{table}[h!] 
	\centering
	\begin{threeparttable}
		\begin{tabular}{cc}
			\toprule\toprule
			\textbf{drug} & \textbf{sample size} \\ 
			\midrule
			alcohol (ALC) & $12{,}922$ \\
			cigarettes (CIG) & $13{,}176$ \\
			amphetamines (AMP) & $13{,}147$\\
			tranquilizers (TRQ) & $13{,}028$\\
			narcotics (NAR) & $12{,}883$\\
			LSD (LSD) & $10{,}980$\\
			psychedelics$^*$ (PSY) & $10{,}975$\\
			sedatives (SED) & $13{,}088$\\
			MDMA (MDM) & $\;\,6{,}552$\\
			cocaine (COK) & $13{,}024$\\
			crack (CRA) & $12{,}505$\\
			methamphetamines (MET) & $\;\,4{,}313$\\
			heroin (HER) & $13{,}005$\\
			\bottomrule
		\end{tabular}
		\begin{tablenotes}
			\item[*] without LSD
		\end{tablenotes}
	\end{threeparttable}
	\caption{Effective sample size for the respective drug and marijuana. The difference to the univariate sample sizes arises due to the necessity to delete data points whenever at least one of the observations is NA.} 
	\label{tab:samplesizes_marijuana}
\end{table}
%\lukas{Do you think it is worthy to include the relative consumption frequency here as well? ALso in the meth table}

\begin{figure}[h!]  
	\centering
	\includegraphics[width=1\linewidth]{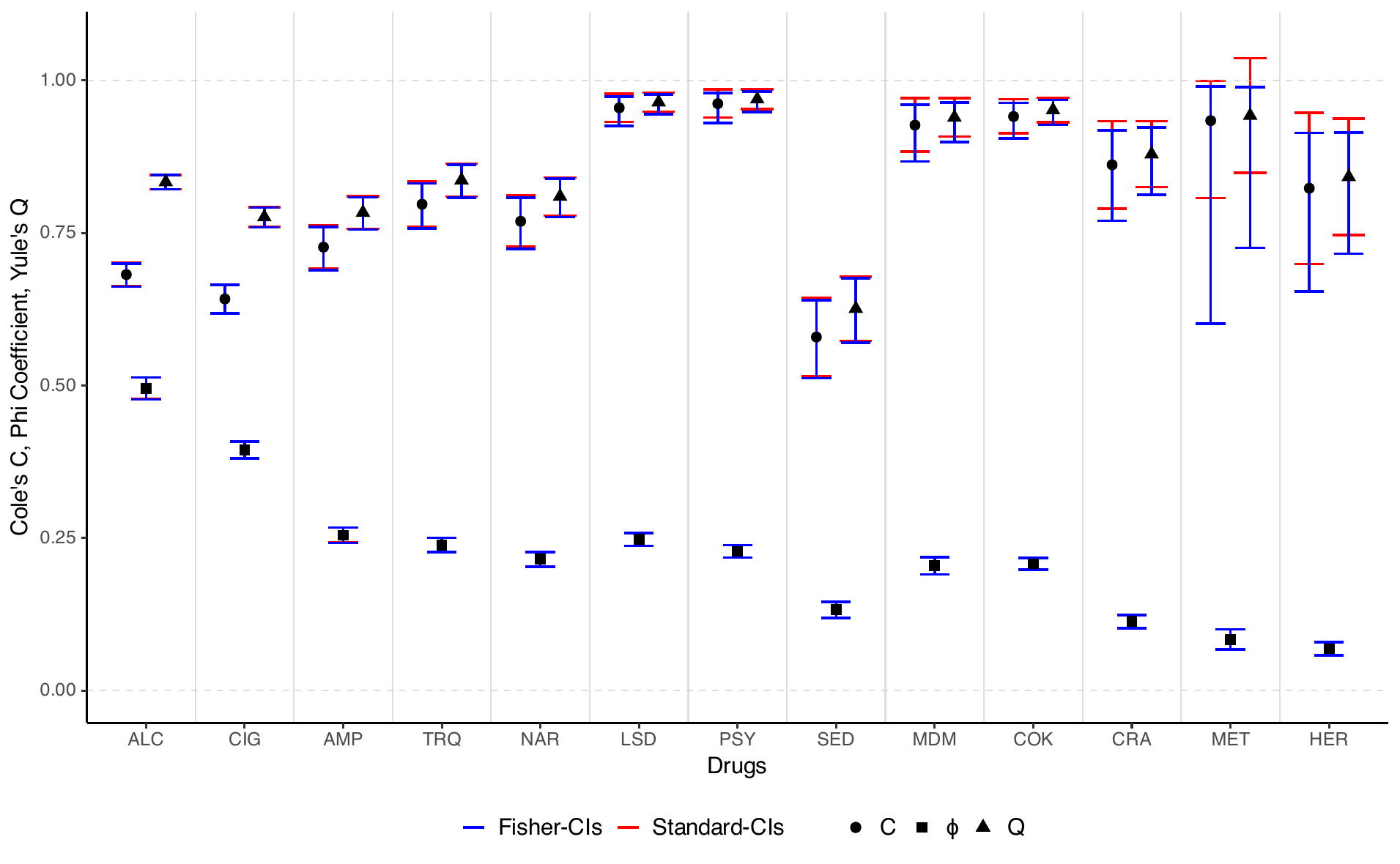}
	\caption{Point estimates $\myCn$, $\widehat{\phi}_n$ and $\myQn$ for marijuana compared with all other drugs and corresponding $90\%$ confidence intervals based on the Fisher transformation (blue) and on standard inference (red).}
	\label{fig:case_study_CIs_mari_comparison}
\end{figure}

\begin{table}[h!] 
	\centering
	\begin{threeparttable}
		\begin{tabular}{cc}
			\toprule\toprule
			\textbf{drug} & \textbf{sample size} \\ 
			\midrule
			alcohol (ALC) &  $4{,}220$\\
			marijuana (MAR) &  $4{,}313$   \\
			cigarettes (CIG) &  $4{,}372$\\
			amphetamines (AMP) &$4{,}354$ \\
			tranquilizers (TRQ) & $4{,}414$\\
			narcotics (NAR) & $4{,}423$\\
			LSD (LSD) &$2{,}163$ \\
			psychedelics$^*$ (PSY) &$2{,}170$ \\
			sedatives (SED) & $4{,}413$\\
			MDMA (MDM) & $4{,}410$\\
			cocaine (COK) &$4{,}408$ \\
			crack (CRA) &$4{,}428$ \\
			heroin (HER) & $4{,}421$\\
			\bottomrule
		\end{tabular}
		\begin{tablenotes}
			\item[*] without LSD
		\end{tablenotes}
	\end{threeparttable}
	\caption{Effective sample size for the respective drug and methamphetamines. The difference to the univariate sample sizes arises due to the necessity to delete data points whenever at least one of the observations is NA.} 
	\label{tab:samplesizes_meth}
\end{table}

\begin{figure}[h!] 
	\centering
	\includegraphics[width=1\linewidth]{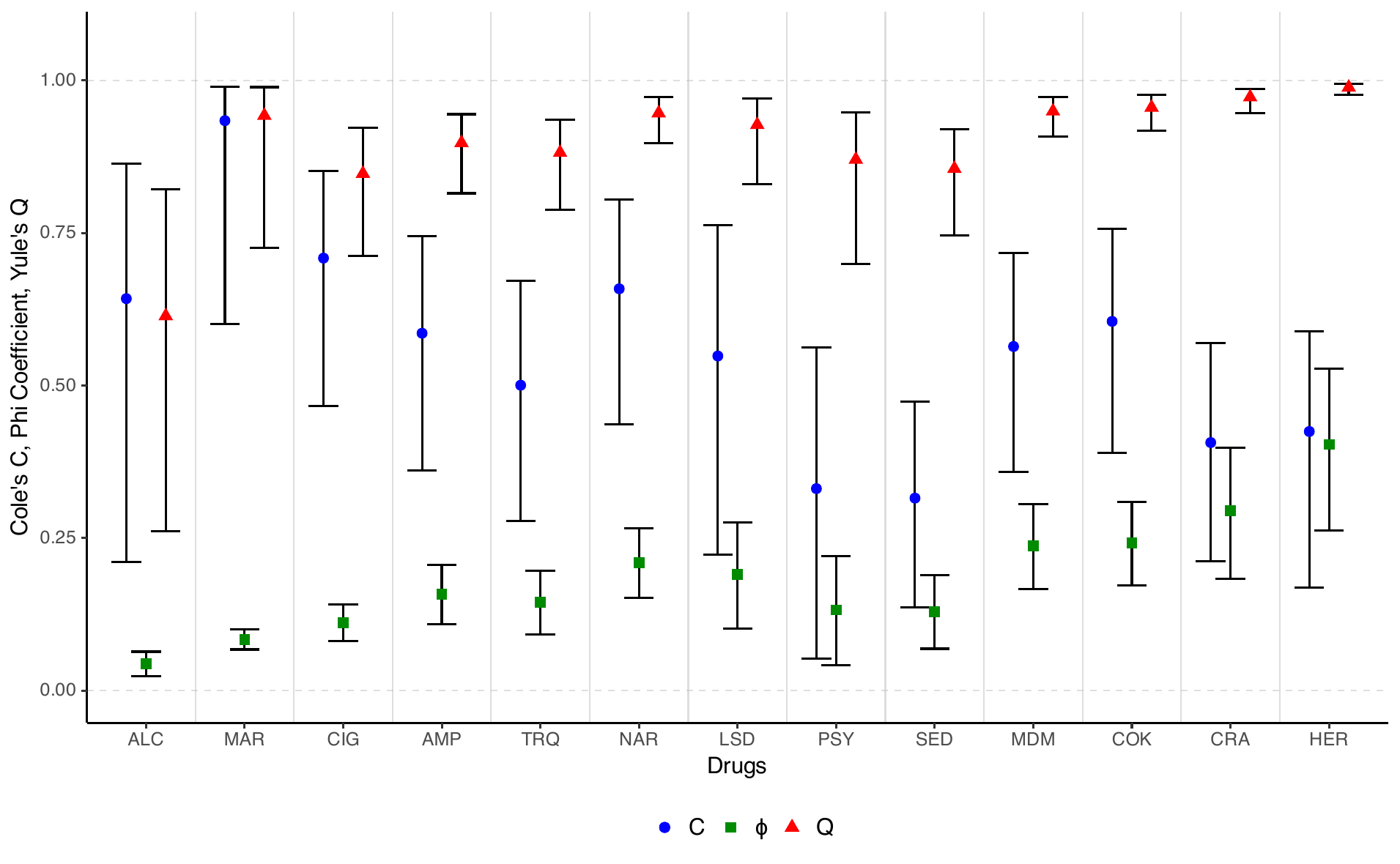}
	\caption{Point estimates $\myCn$, $\widehat{\phi}_n$ and $\myQn$ for methamphetamines compared with all other drugs and corresponding 90\% confidence intervals based on the Fisher transformation.}
	\label{fig:case_study_CIs_meth}
\end{figure}

\begin{figure}[h!] 
	\centering
	\includegraphics[width=1\linewidth]{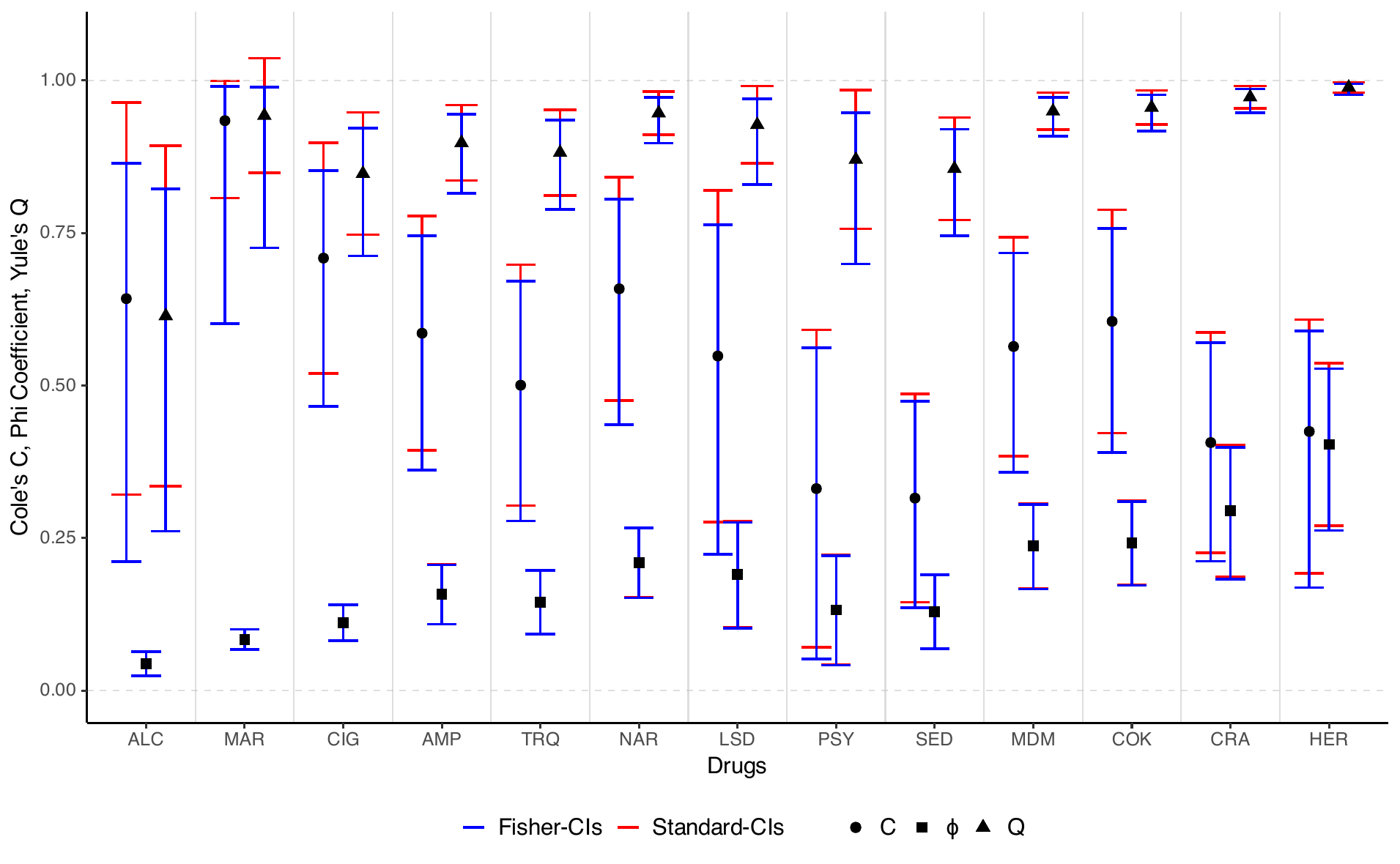}
	\caption{Point estimates $\myCn$, $\widehat{\phi}_n$ and $\myQn$ for methamphetamines compared with all other drugs and corresponding $90\%$ confidence intervals based on the Fisher transformation (blue) and on standard inference (red).}
	\label{fig:case_study_CIs_meth_comparison}
\end{figure}

\end{document}